\documentclass[11pt]{article}%
\usepackage{gensymb}
\usepackage{amssymb,fullpage,soul,nicefrac}
\usepackage{amsfonts}
\usepackage{amsmath}
\usepackage{subfig}
\usepackage{graphicx,enumerate}
\usepackage{fullpage}
\usepackage[round]{natbib}
\usepackage[pdftex]{hyperref}
\usepackage{soul}

\usepackage{amsfonts, amsthm, thmtools,graphicx,latexsym,url,epsf, algpseudocode, verbatim, setspace, float, paralist, mathtools}
\usepackage[usenames,dvipsnames,svgnames,table]{xcolor}

\hypersetup{
    colorlinks,
    linkcolor={blue!50!black},
    citecolor={blue!50!black},
    urlcolor={blue!50!black}
}

\usepackage{tikz, pgfplots}
\usetikzlibrary{arrows,calc}
\tikzset{
>=stealth',
help lines/.style={dashed, thick},
axis/.style={<->},
important line/.style={thick},
connection/.style={thick, dotted},
}

\newcommand{\hi}

\newcommand{\bmath}{\begin{equation}}
\newcommand{\emath}{\end{equation}}
\newcommand{\bmathnn}{\begin{eqnarray*}}
\newcommand{\emathnn}{\end{eqnarray*}}

\newcommand*{\threesim}{\mathrel{\vcenter{\offinterlineskip\hbox{$\sim$}\vskip-.35ex\hbox{$\sim$}\vskip-.35ex\hbox{$\sim$}}}}

\declaretheorem{theorem}
\declaretheorem{lemma}
\declaretheorem{proposition}
\declaretheorem{claim}

\declaretheorem{corollary}

\declaretheorem{remark}
\declaretheorem{fact}
\declaretheorem{definition}

\newtheorem*{theorem*}{Theorem}

\usepackage{ushort}

\def\what{\widehat}

\newcommand{\normone}[1]{\left\lVert#1\right\rVert_{1}}

\def\over{\ensuremath{\overline}}

\def\mcl{{\mc}_{\mathfrak{l}}}
\def\mcu{{\mc_{\mathfrak{u}}}}

\def\calM{\ensuremath{\mathcal{M}}}
\def\calN{\ensuremath{\mathcal{N}}}
\def\Q{\ensuremath{\mathcal{Q}}}

\def\H{\mathcal{H}}
\def\bbR{\mathbb{R}}

\def\expdist{\mathsf{Exp}}

\def\pihat{\ensuremath{\what{\pi}}}
\def\pil{\ensuremath{{{\pi}_{\mathfrak{l}}}}}
\def\piu{\ensuremath{{{\pi}_{\mathfrak{u}}}}}
\def\piy{\ensuremath{{{\pi}_{\mathsf{Y}}}}}
\def\pipy{\ensuremath{{{\pi}'_{\mathsf{Y}}}}}
\def\pippy{\ensuremath{{{\pi}''_{\mathsf{Y}}}}}

\def\xp{H^P}
\def\yp{E^P}
\def\xg{H^G}
\def\yg{E^G}

\def\1{{\bf{1}}}

\def\aref{\autoref}

\renewcommand{\P}[1]{{\mathbb{P}}\left[#1\right]}

\newcommand{\E}[1]{{\mathbb{E}}\left[#1\right]}
\newcommand{\EE}[2]{{\mathbb{E}}_{#1}\left[#2\right]}

\def\MC{\textup{MC}}
\def\mc{\ensuremath{\calM}}

\usepackage{color}

\def\calC{\ensuremath{\mathcal{C}}}
\def\calB{\ensuremath{\mathcal{B}}}
\def\calP{\ensuremath{\mathcal{P}}}

\def\opd{\operatorname{d}}

\def\xstar{\ensuremath{x^*}}

\usepackage{nomencl}
\makenomenclature

\usepackage{graphicx}
\begin{document}

\title{Matching in Dynamic Imbalanced Markets}
\author{Itai Ashlagi \and Afshin Nikzad \and Philipp Strack\thanks{We want to thank Mohammad Akbarpour, Aaron Bodoh-Creed, Yeon-Koo Che, Yuichiro Kamada, and Olivier Tercieux for useful comments and suggestions. Ashlagi: MS\&E,  Stanford, iashlagi@stanford.edu. Nikzad: Economics, University of California, Berkeley, nikzad@berkeley.edu.  Strack: Economics, University of California, Berkeley, pstrack@berkeley.edu. Ashlagi acknowledges the research support of the National Science Foundation grant (SES-1254768).}}


\maketitle

\begin{abstract}
We study dynamic matching in exchange markets with easy- and hard-to-match agents. A greedy policy, which attempts to match agents upon arrival, ignores the positive externality that waiting agents generate by facilitating future matchings. We prove that this trade-off between a ``thicker'' market and faster matching vanishes in large markets; A greedy policy leads to shorter waiting times, and more agents matched than any other policy. We empirically confirm these findings in data from the National Kidney Registry. Greedy matching achieves as many transplants as commonly-used policies (1.6\% more than monthly-batching), and shorter patient waiting times (23 days faster than monthly-batching).
\end{abstract}

%
\section{Introduction}
\label{sec:intro}

We study how to optimally match agents in a dynamic random  exchange market. Matching agents faster reduces  waiting times but at the same time makes the market thinner, leaving more agents without a compatible partner. This trade-off naturally arises for kidney exchange platforms that seek to  form exchanges between incompatible patient-donor pairs.\footnote{For some early work on kidney exchange in  static pools and the importance of creating a thick marketplace, see \citet{rothEfficient2007,roth2004kidney}.}  Waiting to match may increase the number of patients receiving a kidney, but comes at a cost: receiving a transplant earlier does not only improve the quality of life for the patient but also leads to substantial savings in dialysis costs for society.\footnote{The savings from a transplant over dialysis is estimated by over \$270,000 \citep{held2016cost} per year over the first five years.}
In the last decade kidney exchange platforms in  the United States  gradually moved from matching roughly every  month to matching daily.\footnote{The National Kidney Registry (NKR) and the Alliance for Paired Donation (APD) search for matches on a daily basis, whereas the United Network for Organ Sharing (UNOS) search for matches twice a week.} Practitioners are concerned that this  behavior, some of which is driven by competition between Kidney exchanges,\footnote{From personal communication with the kidney exchange directors.} is harmful,  especially for the most highly sensitized patients.\footnote{That is, patients who have common antibodies that will attack foreign tissue.} In contrast, kidney exchange programs in Canada, Australia, and  the Netherlands match periodically every 3 or 4 months \citep{ferrari2014kidney}.

%
%
%
%
This article analyses the trade-off between agents' waiting times and the percentage of matched agents in dynamic markets. 
We find that, maybe surprisingly, matching greedily minimizes the waiting time and simultaneously maximizes the chances to find a compatible partner for \emph{all} agents for sufficiently large markets.
We further quantify the inefficiency associated with other commonly used policies like monthly matching using data from the National Kidney Registry.

%
%
%
%
To analyze this question we propose a stochastic compatibility model with easy-to-match and hard-to-match agents. Easy-to-match agents can match with all other agents  with a positive probability $p$,  whereas hard-to-match agents can  match only with easy-to-match agents with a positive probability $q$.  The main focus of our analysis is on the case where the majority of agents are hard-to-match, which is inline with kidney exchange pools.
This model captures two empirical regularities of the patient-donor data from the National Kidney Registry (NKR): 
First, as the  market grows large, the fraction of patient-donor pairs that are matched in a maximal matching  does not approach 1, which is a consequence of the imbalance  between  different   pairs' blood types in kidney exchange \citep{saidman2006increasing,rothEfficient2007}.\footnote{See also \citet{Misaligned},  who study a production function of a kidney exchange platform in order to quantify the marginal benefits of different types of pairs and altruistic donors.} Second, as the market grows large, the fraction of agents that cannot be matched in {\em any} matching goes to zero.\footnote{A patient-donor pair cannot be matched in any matching if it cannot form a (two-way) exchange with any other patient-donor pair due to biological compatibility.} 
Our parsimonious two-type model captures the above regularities and no single-type model can account for both of them (Propositions \ref{prop:good-fit} and \ref{prop:no-simpler-model}).

%
%
%
We study a dynamic model based on the above two-type compatibility structure in which easy- and hard-to-match agents arrive to the market according to independent Poisson processes with rates $m_E$ and $m_H$. Agents depart exogenously at rate $d$.
%
%
The market-maker observes the realized  compatibilities and decides when to match compatible agents. We evaluate  a policy based on three measures: \textsl{match rate}, \textsl{matching time}, and \textsl{waiting time}. The match rate is the probability with which an agent is matched. The waiting time is the average difference between the time an agent arrives and the time she leaves, either matched or unmatched. The matching time measures how long an agent has to wait on average before being matched, conditional on being matched.

%
%
%
We start by analyzing the {\it greedy policy}, which  matches every agent upon its arrival if possible.
We first derive the distribution of the number of hard- and easy-to-match agents waiting in the market in steady state.
 We find that, as the market grows large, many hard-to-match agents will wait in the market for a compatible partner at any point in time. As a consequence, almost every easy-to-match agent is matched with a hard-to-match agent immediately upon arrival and the probability that an easy-to-match agent leaves the market unmatched converges to zero (Proposition \ref{thm.opt2}).
As their match rate is close to one and their waiting time is close to zero the greedy policy is asymptotically optimal for easy-to-match agents in large markets.
As hard-to-match agents are incompatible with each other and almost every easy-to-match agent is paired with a hard-to-match agent, greedy also maximizes the match rate of hard-to-match agents.
Maybe less intuitively, greedy-matching also minimizes the waiting time of hard-to-match agents compared to \emph{any} other policy when the market grows large (Proposition \ref{prop:upperbound}).
We establish this result by first showing that weakly more hard-to-match agents wait for a partner in any other policy. Then, we use a version of Little's law which implies that the average number of hard-to-match agents waiting in the market is proportional to their waiting time.
Together, this establishes that greedy matching will perform weakly better than any other policy in a sufficiently large market.

The main challenge in the proof is analyzing the steady state distribution of a two-dimensional random walk which keeps track of the number of easy- and hard-to-match agents waiting in the market. Instead of analyzing the two-dimensional process directly ---which is in general intractable--- we use coupling techniques to derive an upper and a lower bound on the marginal distribution of hard-to-match agents. These bounds allow us to derive the distribution of the number of easy-to-match agents waiting in steady state.

%
%
Next, we quantify the inefficiency associated with batching policies, which are commonly used in practice. A batching policy periodically (e.g. monthly) matches as many agents as possible. We derive a lower bound on the waiting time and an upper bound on the match rate under batching policies. As the batching period gets longer the bound on match rate decreases and the bound on the waiting time increases. Our bounds imply, that in a large market, greedy matching dominates any batching policy, as it leads to \emph{strictly shorter} waiting times and \emph{strictly higher} match rates. Quantitatively, our results imply that in a large market the match rate of any batching policy, for both easy- and hard-to-match agents, is at most the match rate under the greedy policy minus half the length of the batching period multiplied by the departure rate. For example if agents exogenously depart on average after a year and the batching period is one month (i.e. \nicefrac{1}{12} of a year) the batching policy will match at least $\nicefrac{1}{24} = 4.2\%$ fewer agents.

%
%
We also analyze the {\it patient matching policy} introduced by \cite{Akbarpour}. This policy assumes that agents' exogenous  departure times are observable. It matches an  agent upon departure if possible, and otherwise the agent leaves the market unmatched.
  We show that the patient policy leads to the same match rate as the greedy policy when the market becomes large. In both policies almost all easy-to-match agents are matched almost upon arrival in a large market. We show that hard-to-match agents wait longer (in first order stochastic dominance) under the patient policy compared to the greedy policy.
  Quantitatively, the waiting time of hard-to-match agents under the greedy policy equals the waiting time under the patient policy multiplied by $\big(1-\frac{m_E}{m_H}\big)$ where $m_E,m_H$ are the arrive rates of easy and hard-to-match agents. For example when $\nicefrac{1}{3}$ of the agent are easy-to-match (i.e. $2 m_E = m_H$) hard-to-match agents will wait twice as long under the patient policy.%
 %
\footnote{The differences with \citet{Akbarpour} are discussed in detail in Sections \ref{sec:lit} and \ref{sec.mainresults}.}


%
%
%
Finally, we test whether the large-market predictions of our model hold in data from the National Kidney Registry (NKR). This data differs from our assumptions along two dimensions: first, because of blood and tissue types it does not match our stylized two-type compatibility structure. Second, it is unclear that the market is sufficiently large for our results to apply, because only a finite number of agents arrive every year (around 360/year). Nevertheless, the data confirms the predictions of our model (Section \ref{sec:simulationsData}):
As the market becomes large, the waiting and matching times of patient-donor pairs who are ``easier'' to match approach $0$, but the waiting and matching time of pairs which are ``harder'' to match do not (c.f. Table \ref{tab:alldata}).\footnote{To capture the theoretical notions of hard-to-match and easy-to-match in our data set, we categorize patient-donor pairs based on notions of  ``over-demanded'' and ``under-demanded'' notions, as defined in Section \ref{sec:simulationsData}.} We further find that batching policies result in no improvement to the match rate and lead to longer waiting times relative to greedy matching (c.f. Table \ref{tab:alldata} and Figure \ref{fig-fmwm}). Lastly, under greedy matching, the waiting and matching times are significantly lower than under patient matching. At the same time, we do not find systematic differences between the match rates under greedy and patient matching (Figure \ref{fig.largesimul} and Table \ref{tab:alldata}).

\subsection{Related literature}
\label{sec:lit}


The most closely related literature studies dynamic matching on networks when agents' preferences are based on compatibilities, motivated by kidney exchanges. This literature was initiated by \citet{Utku}. It is useful to organize this literature into two perspectives:

The papers taking the first perspective seek to minimize the waiting time of agents in the market assuming that agents never depart exogenously. \citet{Utku} analyzes a  kidney exchange model with  linear waiting costs where compatibility is deterministically determined by  blood type. He finds that for pairwise exchanges, greedy matching is  optimal.\footnote{Further, using thresholds to facilitate three-way matchings is beneficial, but generates  relatively small improvements.}
\citet{anderson2013efficient} consider a model where compatibilities are based on  a random  graph model, in which each  agent is compatible with any other agent with some fixed probability. They find that greedy matching minimizes average waiting time, as the compatibility probability tends to zero.\footnote{Their results hold under different types of feasible matchings:  pairwise, pairwise and three-way cycles, and chains.}  \citet{ashlagi2016matching} add asymmetric types to this random compatibility model where one type has a non-vanishing probability of being matched with any other agent. In particular, in contrast to our model, any two types can potentially match.\footnote{Note that types differ only by the probability to match with other agents.} They too find that greedy matching minimizes waiting time.\footnote{They also study  the relationship between the balance in the market between types of agents and the desired matching technology.} This strand of literature thus establishes in various models that greedy matching minimizes average waiting time.

The second perspective one could take to dynamic matching in kidney exchanges is to analyze how many agents are matched. \citet{Akbarpour} consider a model with departures, in which each  agent is compatible with any other agent with some fixed probability. They find that the patient policy leads to an exponentially smaller loss rate (i.e., fraction of unmatched agents) compared to the greedy policy.

Each of the above perspectives studies one of two objectives: either to minimize the time until an agent is matched or to minimize the number of agents that leave the market unmatched. Given the different objectives the two perspectives lead to different conclusions about the optimality of greedy and suggest a trade-off between matching agents quickly and matching as many agents as possible. Our main contribution with respect to this literature is to study this trade-off and show that it vanishes in large kidney-exchange markets with asymmetric agents.

 Technically, our paper is the first to analyze a model with both exogenous departures and heterogeneous agents and to also analyze the distribution of waiting and matching times, rather than just averages.

From a modeling perspective there are two major differences between our paper and most of the the above literature. First, compatibilities in our model depend also on the agent's type (i.e. donor's blood type).\footnote{One may interpret our model as combining both blood types and randomness due to tissue-type incompatibilities. We note  that   \citet{dickerson_dynamic_2012}  develop a heuristic to approximate the full dynamic program and overcome the infamous ``curse of dimensionality.''}  Second, in contrast to \citet{anderson2013efficient} and \citet{Akbarpour} we focus on markets, where the matching probabilities do not vanish. Their models intend to capture sparse compatibility networks and small markets. In contrast we are interested in large markets where the agents' types are independent of the market size as arguably natural for kidney exchanges.
Whether a given market is approximated by either compatibility model  depends highly  on the specific context, and is ultimately an empirical question. Our simulations reveal that kidney exchange markets of even  moderate size behave as predicted by our large market model.

Similar type of asymmetries across agents also appear in \cite{nikzad2017financing}. They are concerned with a proposal for global kidney exchange, which incorporates international pairs to domestic  kidney exchange pools.  In their model, there is a continuum of international pairs who do not get matched to each other and a continuum of domestic pairs who can get matched to each other and to the international pairs. The compatibilities (between measures of pairs) are determined by a ``matching function''.  They do a steady-state analysis to answer whether the savings from dialysis can cover the surgery costs of international pairs.

The effectiveness of thickening the market by waiting to increase the number of matches has also been studied in markets other than kidney exchanges.
\cite{ridesharing} compare the match rates of greedy and patient policies in ride-sharing markets for matching drivers with passengers and find that ``the patient algorithm may not necessarily generate more matches than the greedy algorithm with heterogeneous compatibility at some parameterization.''
Finally,  recent and indirectly, related papers study  dynamic matching when agents' have preferences that do not depend only on compatibility. These papers find that greedy policies are inefficient \citep{baccara2015optimal,doval2014theory} since some waiting can improve the quality of matches.\footnote{See also related results in queueing models \citep{Leshno,bloch2014dynamic}.}

\section{The compatibility graph}
\label{sec.model}\label{SEC.MODEL}

A kidney exchange pool can be represented by a {\it compatibility graph} $G$. Each node in the graph represents an agent (a patient-donor pair), and a link between two nodes exists if and only if the two corresponding agents are {\it compatible} with each other (so a bilateral exchange between the nodes is feasible). We restrict attention to bilateral exchanges.

A {\it matching} $\mu$ is a set of non-overlapping  compatible pairs of agents. Denote by $M(G)$ the set of matchings in $G$.\footnote{The paper restricts attention  to matching only pairs of agents and not through chains. For the  effect of matching through chains see, e.g., \citet{ashlagi2011nead} and \citet{anderson2013efficient}.}  For every compatibility graph $G$ let  $|G|$ denote the number of agents in the graph, and for every matching $\mu$ let  $|\mu|$ denote the number of agents in that matching.

Define the {\it (normalized) size of the maximum matching} (SMM) in a graph $G$ to be the fraction of matched agents in a  maximum matching:
\begin{align*}
 \text{SMM} &= \max_{\mu \in M(G)} \frac{|\mu|}{|G|} \, .
\end{align*}

Define the {\it fraction of agents without a partner} (FWP) to be the fraction of agents that are not matched in any matching (thus have no compatible agent):
\begin{align*}
	 \text{FWP} &= \frac{|\{ i \in G \colon (i,j) \notin M(G) \text{ for all } j \}|}{|G|} \,.
\end{align*}

Figure \ref{fig:features} depicts the average SMM and FWP for randomly drawn subsets of the patient-donor population acquired from the National Kidney Registry (NKR), the Alliance for Paired Donation (APD), and the United Network for Organ Sharing (UNOS). This data includes 4132 patient-donor pairs.  Two features stand out: first, even as the pool grows large, the size of the maximum matching stays bounded away from 1, i.e., $\text{SMM} < 1$. This is a natural consequence of the different blood types \citep{rothEfficient2007}.  Second, when the market grows large, the fraction of  pairs  that have no compatible pairs decreases. Roughly 8\% of pairs are incompatible with any other pair in this data (this is the FWP).\footnote{In practice some patients can receive a kidney from  blood-type incompatible donors due to advanced technology. For the sale of simplicity we ignore this in our simulations, but it is worth noting that the FWP drops to roughly 4\% when  this form of compatibility is allowed.} Since compatibility depends only on the characteristics of the patients and donors, it is independent of pool size, and thus in a sufficiently large  pool one would expect that the FWP would further decrease to zero.
\begin{figure}[ht]
\begin{center}
\includegraphics[scale=0.6]{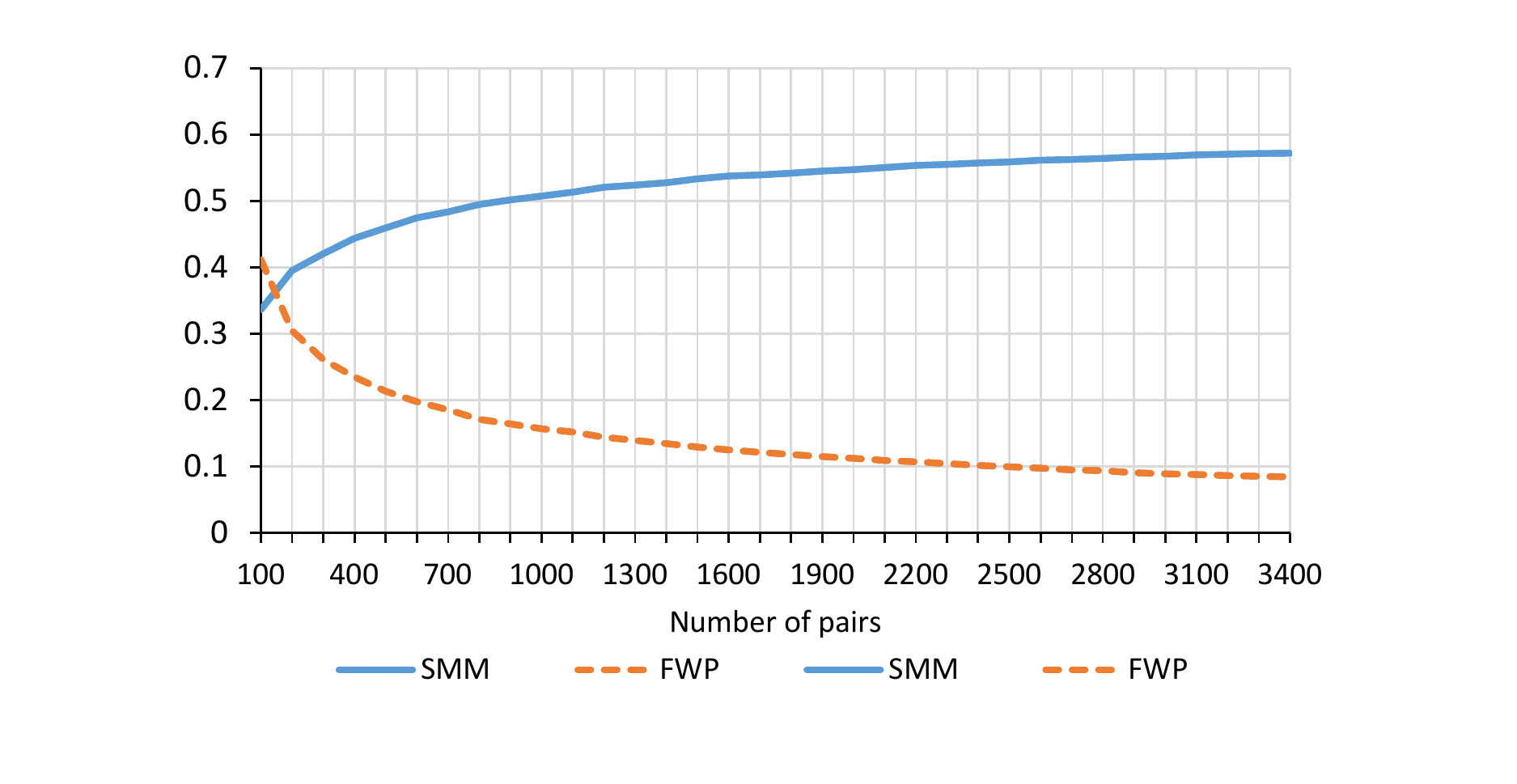}
\vspace{-0.9cm}
\caption{\label{fig:features}The average fraction of patient-donor pairs without a compatible partner in blue and the (normalized) size of the  maximum matching in red, for a random subset of patient-donor pairs from  NKR, APD, and UNOS data.}
\end{center}
\end{figure}

\begin{fact}\label{fac:empirical-regularity} As the kidney exchange patient-donor pool grows large, the  compatibility graph (Figure \ref{fig:features}) is such that the size of the maximal matching (SMM) stays bounded away from $1$ and the fraction of patient-donor pairs without a compatible partner (FWP) goes to $0$.
\end{fact}

The change in both the SMM and  FWP measures  captures the benefit of a larger market. Since a  matching policy in a dynamic environment  trades off the benefits of a larger market with the waiting costs incurred by the agents, having a model that accurately represents the SMM and the FWP is important to correctly describe the costs and benefits of waiting to match.

\subsection{A compatibility model}\label{sec:compatibility-model} To capture the features of kidney exchange identified in Fact \ref{fac:empirical-regularity} we adopt  a stylized and tractable model with random compatibilities. There are two types of agents, {\it easy-to-match}  or {\it hard-to-match}, denoted  by E and H, respectively. There are more hard-to-match than easy-to-match agents.
Any pair of hard-to-match and easy-to-match agents are compatible independently with probability $p>0$, any  pair of two easy-to-match agents are compatible  independently with probability $q>0$, and  no pair of hard-to-match agents are compatible with each other (see Figure \ref{fig:random-compatibility}).

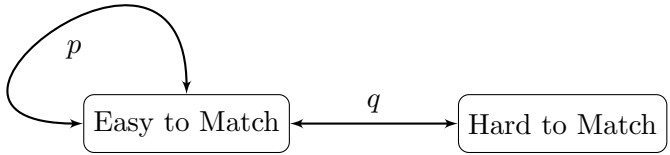
\begin{figure}[ht]
\begin{center}
\tikzstyle{int}=[draw,  rounded corners, minimum size=2em]
\tikzstyle{init} = [pin edge={to-,thin,black}]
\begin{tikzpicture}[node distance=5cm,auto,>=latex']
\clip(-3,-0.5) rectangle (7,1.7);
    \node [int] (a) {Easy to Match};
    \node [int] (c) [right of=a] {Hard to Match};
    \path[thick,<->] (a) edge node {$q$} (c);
    \draw[thick,<->] (a) to [out=90,in=180,looseness=5,] node {$p$} (a);
\end{tikzpicture}
\end{center}
\caption{\label{fig:random-compatibility}
The random compatibility model.}
\end{figure}

Proposition \ref{prop:good-fit} shows that this simple model is indeed able to capture the features of real kidney exchanges identified in Fact \ref{fac:empirical-regularity}:

%
\begin{proposition}\label{prop:good-fit}
Consider  a compatibility graph with  $m$ easy-to-match agents and $(1+\lambda)m$ hard-to-match agents where $\lambda>0$. Compatibilities between pairs of agents  are generated as described in Section \ref{sec:compatibility-model}. As  $m$  grows  large the $\mathrm{SMM}$ goes to $\frac{2}{2+\lambda}$ and the $\mathrm{FWP}$ goes to zero almost surely:
\begin{align}
	\lim_{m \to \infty} \mathrm{SMM} &= \frac{2}{2+\lambda},\label{eq:SMM-large-m}\\
	\lim_{m \to \infty} \mathrm{FWP} &= 0 \label{eq:FWP-large-m}\,.
\end{align}
\end{proposition}

That the size of the maximal matching cannot exceed $\frac{2}{2+\lambda}$ is intuitive: since H agents cannot match with each other and there are more H agents than E agents, some H agents must remain unmatched when the pool is large. An upper bound on the fraction of agents that can be matched equals twice the fraction of E agents $\frac{1}{2+\lambda}$. Furthermore, note that this fraction is achieved whenever there exists a matching, in which  all E agents are matched with H agents. It follows from a standard result in random graph theory that the probability that such a ``perfect matching" exists approaches 1 as the pool grows large.
Furthermore, as the pool grows large any H agent will be compatible with some E agent, since  compatibilities between agents are drawn independently. Thus, the fraction of agents without a partner converges to 0.

The parameter $\lambda$ of the model measures the degree of imbalance between hard- and easy-to-match agents. So $\lambda=0$ corresponds to a balanced pool. Figure \ref{fig:features} suggests that the size of the maximal matching in the national kidney exchange data converges to roughly $60\%$ when the pool becomes large, implying that $\lambda \approx 1.33$ in the context of our model.
This reduced-form calibration is roughly consistent with a value of $\lambda\approx 1.33$  that one obtains for the same data when defining hard-to-match agents directly as those who cannot match with each  other due to blood-type incompatibilities (see Section \ref{sec:simulationsData}).

Proposition \ref{prop:good-fit} establishes that our two-type model can match the empirical behavior of the $\textrm{SMM}$ and $\textrm{FWP}$ measures. Proposition \ref{prop:no-simpler-model} establishes that no model with a single type can replicate the empirical features of real kidney exchanges observed in Fact \ref{fac:empirical-regularity}, even when allowing  the probability  of compatibility between two agents  to depend on the market size in arbitrary ways.

\begin{proposition}\label{prop:no-simpler-model} Consider a model with $m$ homogeneous agents, in which every  pair of  agents are compatible independently with probability $p(m)>0$ that may depend on the market size.
The following two conditions cannot be satisfied simultaneously almost surely: 
\begin{align}
	\lim_{m \to \infty} \mathrm{SMM} &<1, and \label{eq.f21}\\
	\lim_{m \to \infty} \mathrm{FWP} &= 0 \label{eq.f22}\,.
\end{align}
\end{proposition}

The proof of Proposition \ref{prop:no-simpler-model} is constructive. It begins with assuming that every agent has a compatible partner when the pool grows large, i.e., that \eqref{eq.f22} is satisfied. It then constructs an algorithm which selects a single matching for any given compatibility graph and proves that the matching selected by this algorithm will include all agents with high probability as the pool grows large. This implies that \eqref{eq.f21} and \eqref{eq.f22} cannot be simultaneously satisfied in \textsl{any} random graph model with homogeneous agents. %

Economically, this observation implies that heterogeneity of agents plays a major role in kidney exchanges.\footnote{This is consistent with \citet{rothEfficient2007} and \citet{Misaligned}, who demonstrate that the types of patients and donors play a crucial role for efficiency.}
Our  two-type model  is arguably the simplest random compatibility model that captures these features of the compatibility graph.

\section{Dynamic matching}
\label{sec:model}

\subsection{Model}
\label{subsec:model}

We  embed the static compatibility model from Section \ref{sec:compatibility-model} in a dynamic model to allows to study  matching policies in a dynamic setting.
We consider an infinite-horizon dynamic model, in which agents can match bilaterally.  Easy-to-match  agents arrive to the market according to a Poisson process with rate $m$, and hard-to-match agents arrive to the market according to an independent Poisson process with rate $(1+\lambda)m$. We assume that the majority of agents are hard-to-match, that is $\lambda>0$,  unless explicitly stated otherwise.


An agent that arrives to the market at time $t$  becomes  {\it critical} after $Z$ units of time in the market, where $Z$ is distributed exponentially with mean $d$, independently between agents.  We refer to $d$  as the {\it exogenous departure rate}, or  {\it departure rate} for the sake of brevity.  The latest time an agent can match is the time she becomes critical, $t+Z$; immediately after  this time  the agent  leaves the market unmatched.

\paragraph{Matching policies.}  Denote by $G_t$ the compatibility graph induced by the agents that are present at time $t$.
A {\it dynamic matching policy} selects at any time $t$ a matching $\mu_t\in M(G_t)$, which may be empty. Whenever a non-empty matching is selected, all matched agents leave the market. 


Several kidney exchange platforms in the United States typically match in a greedy manner, attempting to match a patient-donor pair as soon as it arrives to the market (see \citet{frequency}). A tractable approximation of this behavior is a greedy matching policy.
\begin{definition}[Greedy] In the \emph{greedy} policy an  agent is matched  upon arrival with a compatible agent if such an agent exists. If she is compatible with  more than one agent, $H$ agents are prioritized over $E$ agents and otherwise ties are broken randomly.
\end{definition}


Some platforms identify matches periodically (thus less frequently than a greedy matching policy), allowing the pool to thicken and possibly offer more matching opportunities. For example, UNOS matches twice a week, whereas national platforms in the  United Kingdom and the Netherlands identify matches every three months \citep{eurokpd}.
 This behavior is approximated with the following batching policy.


\begin{definition}[Batching] A \emph{batching} policy executes a maximal match every $T$ days. If there are multiple maximal matches, select randomly one that maximizes the number of matched H agents.
\end{definition}

The last policy we consider is a patient matching policy, proposed by \citet{Akbarpour}, which attempts to match an agent only once she becomes critical. In the context of kidney exchange this means that two patient-donor pairs in the pool are matched  only if the condition of one of these pairs is such that it cannot match at a later point in time for medical or any other reason.\footnote{Such a policy is practical if the times at which pairs become critical are observable.}


\begin{definition}[Patient] In the \emph{patient} policy an agent that becomes critical is matched with a compatible agent if one exists. If she is compatible with  more than one agent, H agents are prioritized over $E$ agents, and ties are broken randomly.
\end{definition}

Observe that  the greedy and  patient  policies match at most two agents at any given time because no two agents ever arrive or become critical at the same time.  The batching policy, however, can match multiple agents in a given period.


\paragraph{Measures for performance.}  To study the performance of a matching policy we focus  on two measures. One is the  {\it match rate} of each type $\Theta\in\{E,H\}$, which is the fraction of agents of type $\Theta$ that match at the steady state. The other is the {\it expected waiting time} (or simply {\it waiting time}) an agent of type $\Theta$  spends in the market, whether eventually matched or not. For type  $\Theta\in\{H,E\}$ we denote its match rate by $q_{\Theta}$ and its expected waiting time by $w_{\Theta}$, where the policy will be clear from the context. Another measure we analyze is the {\it matching time} of a  type $\Theta$, which is the average time agents of type $\Theta$ who eventually match spend in the market.

In kidney exchange the match rate corresponds to the probability of exchanging a kidney with another patient-donor pair. Because waiting for a kidney is often spent on dialysis, which is costly (both financially and physically), these quantities have a direct impact on welfare.


We are interested in optimal policies for large pools. Formally, we consider the following  notion of optimality:
\begin{definition}[Asymptotic optimality]
A policy is \textbf{asymptotically optimal} if for every $\epsilon>0$ there exists an ${m}^\star$ such that if the arrival rate is large $m\geq {m}^\star$, \textbf{every} type of agent improves its match rate and expected waiting time by at most $\epsilon$ when changing to any other  policy.
\end{definition}
This optimality notion is demanding, since it requires the policy to be optimal for every type of agent simultaneously. It is unclear whether an asymptotically optimal policy exists, since a policy  that is optimal for H agents might be suboptimal for E agents.

\subsection{Results}\label{sec.mainresults}


In this section we present the main result of the paper and discuss its implications. We discuss the logic underlying the results  in Section \ref{sec.mainresults.disc}. The main finding is a characterization of the match rates and waiting times associated with the greedy, batching and patient matching policies.

\begin{theorem}\label{thm.opt} \label{thm.optimalpolicy}
The greedy policy is asymptotically optimal, whereas the batching and patient  policies are not asymptotically optimal.
\end{theorem}
We further compute the match rates and expected waiting times under these policies.
\begin{proposition}\label{thm.opt2} \label{thm.optimalpolicy2}
As  the arrival rate $m$ grows large:
\begin{enumerate}
\item[(i)] The match rates of hard- and easy-to-match pairs under greedy approach $(q^G_H,q^G_E)=(\frac{1}{1+\lambda},1)$, respectively, and their expected waiting times approach $(w^G_H,w^G_E)=(\frac{ \lambda\, d}{1+\lambda},0)$.
\item[(ii)] The batching policy, which matches every $T$ periods, achieves match rates of at most $(q^B_H,q^B_E)=(\frac{1-e^{-T/d}}{ (1+\lambda) T /d} ,\frac{1-e^{- T/d}}{T/d})$.
 Furthermore, the expected waiting and matching time for each type $\Theta\in\{\textrm{E},\textrm{H}\}$
is at least $w^B_{\Theta}=d (1-q_{\Theta})$.
 Also,  $q^B_{\Theta}<q^G_{\Theta}$,  whereas $q^B_{\Theta}$ approaches $q^G_{\Theta}$ as $T$ approaches $0$. In addition,
$w^B_{\Theta}>w^G_{\Theta}$,  whereas $w^B_{\Theta}$ approaches $w^G_{\Theta}$ as $T$ approaches $0$.
\item[(iii)]
The match rates of hard- and easy-to-match pairs under
the patient policy approach  $\frac{1}{1+\lambda}$ and $1$, respectively, and their expected waiting times approach $d$ and $0$, respectively.
\end{enumerate}
\end{proposition}

%

Figure \ref{fig:illustration.mainresults} illustrates the match rates and waiting times of $H$ and $L$ agents under the different policies as found in Theorem \ref{thm.opt2}.
\begin{figure}[h!]
\centering
\includegraphics[scale=0.375]{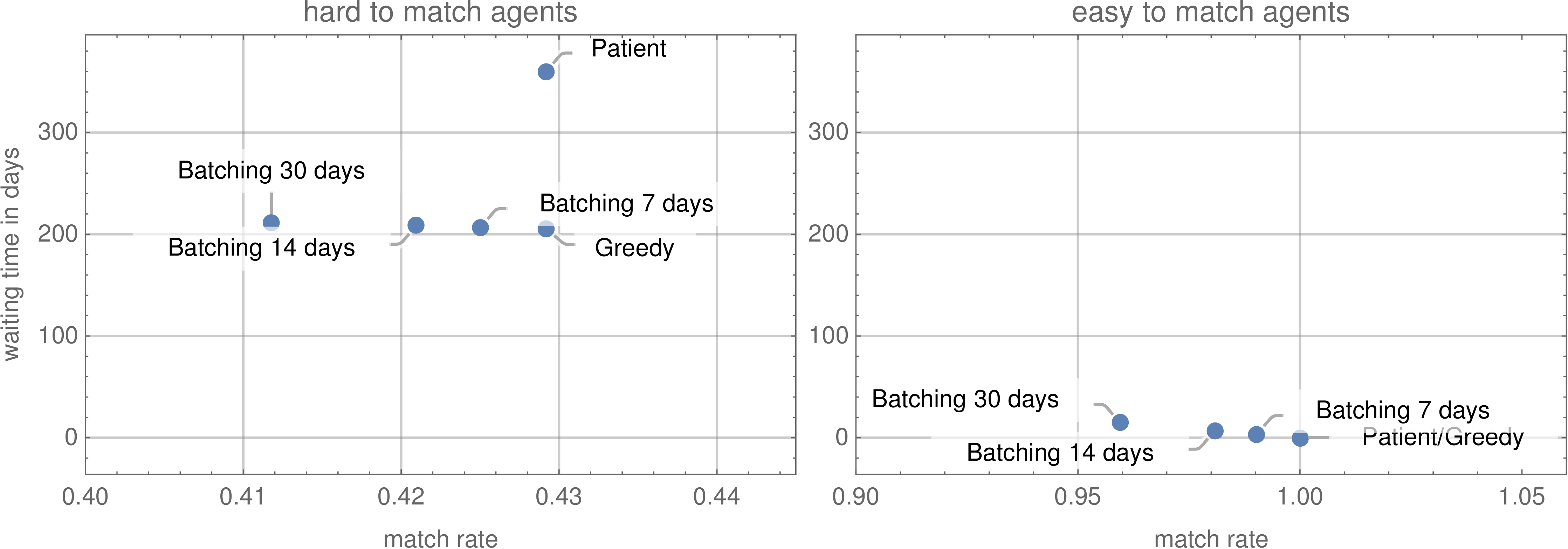}
\caption{\label{fig:illustration.mainresults} Illustration of Theorem \ref{thm.optimalpolicy} when there are twice as many hard-to-match agents as easy-to-match agents ($\lambda=1.33$) and  expected departure time is $360$ days. The blue dots represent the predictions of our model as the arrival rate goes to infinity.}
\end{figure}

In Figure \ref{fig:illustration.mainresults} we calibrated the model such that it matches the data from the National Kidney Registry, where agents depart on average after 360 days and $\lambda \approx 1.33$. As Figure \ref{fig:illustration.mainresults} illustrates, the batching policy  leads agents to wait longer and get matched with a smaller probability than under the greedy approach. The losses resulting from this are substantial. For example, under a monthly batching policy hard-to-match agents wait on average $6$ days and easy-to-match agents $15$ days longer and get matched with $1.7\%$ and $4\%$ lower probability. Similarly, the patient policy matches equally many agents as  the greedy policy, but leads to a substantially longer expected waiting time for hard-to-match agents (155 more days).

We now provide a rough intuition for the  differences among greedy, batching and  patient matching policies. In Section \ref{sec.mainresults.disc} we provide a more precise intuition and a sketch of the argument for the various parts of the results. As there are more hard- than easy-to-match agents, hard-to-match agents will accumulate and a large number of them will be present at any time under any policy. This implies that under  greedy matching,  easy-to-match agents will have upon arrival, with high probability, a compatible hard-to-match agent  and are therefore  matched immediately. As a consequence, every easy-to-match agent is matched with a hard-to-match agent, which implies that  greedy matching  asymptotically achieves the optimal match rate.

Under the batching policy each agent waits at least from the time of her arrival until the next time a matching is identified. Thus each agent waits on average at least half the length of the batching interval. Furthermore, each agent departs during that time with strictly positive probability. Thus, easy-to-match agents are worse off under the batching policy than under the greedy policy  where they get matched immediately with probability 1. As some easy-to-match agents leave the market unmatched, hard-to-match agents are matched with a smaller probability. Consequently there are, on average, more hard-to-match agents waiting in the market. Little's law, which states that the arrival rate multiplied by  the average waiting time equals the average number of waiting agents, implies that hard-to-match agents also wait longer under any batching policy than under a  greedy matching policy. As both types  are worse off, batching policies are not asymptotically optimal.

By analyzing the dynamics of the market we show that under  patient matching, so many  hard-to-match agents accumulate that an easy-to-match agent will match, with  high probability, with a critical hard-to-match agent almost immediately upon arrival. This implies that  patient matching asymptotically achieves the optimal match rate. As hard-to-match agents  get matched only when they become critical, the distribution and expectation of their waiting time is the same as if they do not match at all. Hence, hard-to-match agents get matched faster under  a greedy policy,  implying that  patient matching  is not asymptotically optimal.

Finally, observe that the smaller the  imbalance in the market, the faster hard-to-match agents match under the greedy policy. Under patient matching, however, the  waiting time  distribution of hard-to-match agents is independent of the market imbalance. So as $\frac{m_H}{m_E}$ approaches 1, the ratio between the average waiting times under patient and greedy matching policies approaches infinity.

\begin{remark}
{\em
For  completeness, we  prove  in Appendix \ref{sec:small-lambda} a counterpart for some parts of Theorem \ref{thm.opt} for the empirically irrelevant case  $\lambda<0$. In that appendix we establish the intuitive fact that, when the majority of agents are easy to match, the expected waiting and matching times approach $0$ under both greedy and patient matching policies as the market size $m$ approaches infinity.
}
\end{remark}


\begin{remark}
{\em
\cite{Akbarpour} find, in apparent contrast to our results, that patient  matching performs better than greedy matching. The difference stems from a combination of two factors. First, in their model,
there is a single type of agent. Second, the likelihood that two agents  match vanishes in the arrival rate as  $p(m)=\frac{c}{m}$ for some constant $c>0$.
We set the likelihood of matching to be independent of the size of the market and allow for agents to have differing abilities to match with other agents (in line with the  empirical structure of kidney exchanges -- recall Proposition \ref{prop:good-fit} and Proposition \ref{prop:no-simpler-model}).

A further difference is that \cite{Akbarpour} measure the ratio between loss rates of different matching policies, whereas we are interested in the match rate.\footnote{In our model, as the market grows large, the ratio between the loss rates under the greedy and patient matching policies approaches $1$ as does the ratio between the match rates.  Their model, however, predicts an exponential ratio between the loss rates, with the exponent being proportional to the average degree of the compatibility graph. As the average degree grows larger, the ratio between loss rates grows. However, the ratio between match rates approaches $1$. (Loss rates under both policies approach $0$ and match rates approach $1$.)} Analyzing the match rate allows us to show that greedy matching is an asymptotically optimal policy whenever the agent has risk-neutral expected utility preferences. While the ratio between loss rates is an intuitive measure it has no direct relation to expected utility preferences.

}
\end{remark}

\subsection{Discussion of results}\label{sec.mainresults.disc}

In this section we provide a proof sketch for the various parts of \autoref{thm.opt} and \autoref{thm.opt2} as well as additional results on the matching time distributions.
It is  useful to first establish  an upper bound on the performance of \textsl{any} policy.
\begin{proposition}[Upper bound on the performance of any policy]
\label{prop:upperbound}
For {any} $\epsilon$ there exists $m_\epsilon$ such that for any market size $m > m_\epsilon$ and {any} policy the match rate of H agents is at most $\frac{1}{1+\lambda}+\epsilon$ and the expected waiting and matching time is at least $\frac{\lambda\, d}{1+\lambda}-\epsilon$.
\end{proposition}

Proposition \ref{prop:upperbound} is shown by considering the hypothetical situation where each E agent can  match with any H agent and thus in a large market no  E agent remains unmatched. Since H agents cannot match with other H agents, this provides an upper bound on the probability of an H agent being matched, which is at most the ratio of E to H agents, $\frac{1}{1+\lambda}$.

\subsubsection{Greedy matching policy}



Next we analyze the performance of  greedy matching  as the market grows large. The following proposition includes the results in  the first part of \autoref{thm.opt2}.
\begin{proposition}[Performance of the greedy matching policy]
\label{prop:greedy} Consider the greedy  policy as the market grows large ($m\to\infty$). The match rate of H agents converges to $\frac{1}{1+\lambda}$ and the waiting and matching times converge to an exponential distribution with mean $\frac{\lambda\, d}{1+\lambda}$. The match rate of E agents converges to $1$ and their waiting and matching times converge to $0$.
\end{proposition}

We first provide intuition for the waiting time distribution. Consider   greedy matching  in a deterministic setting where every E agent is compatible with every H agent,  and agents arrive and depart deterministically. In this setting  E agents will be matched upon arrival with H agents. This means that there will be no E agents waiting in the market, and their waiting time equals zero. Denote by $x$ the steady state number of H agents present in the market. Per unit of time, $(1+\lambda) m$ H agents arrive to the market and $m$ of them are  matched with E agents. Further, $\frac{x}{d}$ of the waiting agents are expected to depart unmatched per unit of time. In the steady state the number of departing agents equals the number of unmatched arriving agents. Thus, $x$ solves the balance equation
\[
\frac{x}{d} = \lambda \,m \,\,\,\,\Rightarrow\,\,\,\, x=\lambda \, m \, d\,.
\]
Thus, if the matching partner for an E agent is chosen at random, each H agent has a probability of $\frac{m}{\lambda \,m \,d} =\frac{1}{\lambda\,d}$ of being chosen per unit of time. The time at which a never-departing H agent would be matched is therefore exponentially distributed with rate $\frac{1}{\lambda\,d}$. The time until an H agent departs the market is exponentially distributed with rate $d$. Since the minimum of two exponentially distributed random variables is again exponentially distributed with rate equal to the sum of the rates, the waiting time of an H agent is exponentially distributed with rate $\frac{1}{\lambda\,d} + \frac{1}{d} = \frac{1+\lambda}{\lambda\,d}$, and thus with mean $\frac{\lambda\,d}{1+\lambda}$.

The formal proof of Proposition \ref{prop:greedy} is more complex. The main idea is to show that for a sufficiently large market $m$ the steady state of the model with random compatibilities is close to the steady state of the model where every E agent is compatible with every H agent and agents arrive and depart deterministically. Our results can thus also be understood as a motivation for studying deterministic models. To show this approximation, consider the number of waiting H agents at time $t$, denoted by $x_t$, and the number of waiting E agents, denoted by $y_t$. We show that under  greedy matching  $(x_t,y_t)_t$ is a two-dimensional continuous-time Markov process.\footnote{See also \citet{burqOR}, who analyze a two-dimensional Markov chain.} We derive the fixed-point equation, which characterizes the steady state distribution of this process and show that it admits a unique solution. We then prove that the stationary distribution must have exponential tails, and the rate at which the density of the stationary distribution decays in the tails shrinks at least at a rate proportional to $\frac{1}{\sqrt{m}}$. We use this to show that the steady-state number of E and H agents in the market is not more than a factor of $\sqrt{m}$ away from the solution for the fixed-point equation. As this distance grows slow relative to the market size, the random fluctuations of $(x,y)$ are well approximated by their expectations which correspond to the dynamics of the deterministic setting described earlier. A complication in this analysis is the two-dimensional nature of the stochastic process $(x,y)$, which requires analyzing some auxiliary problems that we describe in detail in the appendix.

Next, we explain why the matching time of H agents follows the same distribution as their waiting time. Consider an arriving H agent. Let $t_d$ be the (random) departure time for this agent. Let $t_m$ be the time at which the agent would be matched with another agent, assuming that the agent never departs.
The waiting time of $h$ then is $t=\min\{t_d,t_m\}$. The conditional distribution of the minimum of two independent exponential random variables $t_d,t_m$ is independent of which one is smaller.\footnote{ That is,  for any $z>0$, $\P{t_m < z \big| t_m=t} =  \P{t < z}.$
This holds because
\[
\P{t_m < z \big| t_m=t} = \frac{\P{\min\{t < z,t_m=t\}}}{\P{t_m =t}}=\frac{\P{t < z}\cdot \P{ t_m=t}} { \P{t_m =t} } = \P{t < z}.
\]
The second equality follows from the fact that the events $t < z$  and  $t_m=t$ are independent, which holds because
$t_d,t_m$ are independent exponential random variables.} Thus, the matching time $t_m$ of H agents has the same distribution as $t$.

As Proposition  \ref{prop:greedy} shows,   greedy  matching achieves the upper bound  derived for arbitrary policies in Proposition \ref{prop:upperbound} and we conclude that  greedy matching  is asymptotically optimal.

\subsubsection{Patient matching policy}

We next quantify the performance of  patient matching policy. The following proposition includes the third part of  \autoref{thm.opt2}.
\begin{proposition}[Performance of the patient matching]
\label{prop:patient} Consider the patient policy when the pool grows large ($m\to\infty$). The match rate of H agents converges to $\frac{1}{1+\lambda}$ and the waiting and matching time converge to an exponential distribution with mean $d$.  The match rate of E agents converges to $1$ and their waiting and matching times converge to $0$.
\end{proposition}

To get a rough intuition for this result, again consider the hypothetical case where all H and E agents can match and agents arrive deterministically. There exists a steady state under the patient policy such that there are no E agents in the market  and the number of H agents in the market is approximately $(1+\lambda) \, m\, d  $. The reason this is a steady state is as follows. H agents get critical and attempt to match with E agents at a rate of $\frac{(1+\lambda)\, m\, d}{d} = (1+\lambda)\, m$, whereas E  agents enter the market only at rate $m$. This means that  E agents are matched immediately and the steady-state number of E agents in the market remains zero. Since no E agent becomes critical, there are no H agents who are matched with a critical E agent. Thus, H agents depart at the rate $\frac{(1+\lambda)\,m\,d}{d}= (1+\lambda)\,m$ and arrive at the rate $(1+\lambda)\,m$, which implies that this is a steady state. A standard argument implies that the steady state is unique. Since H agents are the ones that initiate matches, their average waiting time equals the average time $d$ until they depart exogenously. By the same argument given for  greedy matching  it follows that the waiting and matching times have the same distribution.

The main technical difficulty in proving Proposition \ref{prop:patient} is the same as in the proof Proposition \ref{prop:greedy}; we need to approximate the stochastic process describing the number of waiting E and H agents with a more tractable process and prove that the steady state of this approximation converges to the steady state of the heuristic model discussed above when the market becomes large.

\subsubsection{Batching policy}

Figures \ref{fig.bmatchrate} and \ref{fig.bwaitingtime} graphically compare the bounds provided by  \autoref{thm.opt2}  on the match rate and waiting times of H agents under  the batching policy   to the match rates and waiting times of H agents under the greedy and patient matching policies.


\begin{figure}[ht]
\centering
\begin{minipage}{0.45\textwidth}
  \centering
  \includegraphics[width=1\linewidth]{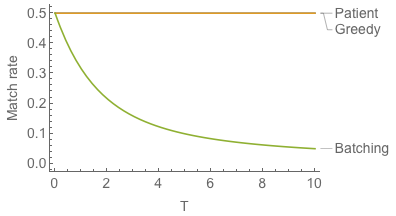}
  \captionof{figure}{Upper bound on the match rate of H agents under the batching policy when $\lambda=1$.}
  \label{fig.bmatchrate}
\end{minipage}%
\hspace{0.4cm}
\begin{minipage}{.45\textwidth}
  \centering
  \includegraphics[width=1\linewidth]{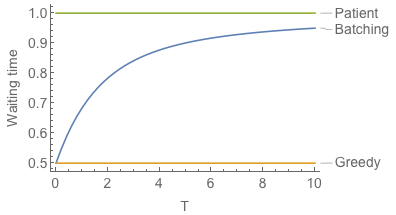}
  \captionof{figure}{Lower bound on the waiting time  of H agents under the batching policy when $\lambda=1$.}
  \label{fig.bwaitingtime}
\end{minipage}
\end{figure}

The bounds given by   \autoref{thm.opt2} for H agents are derived by analyzing a simpler stochastic process in which
(i) easy-to-match nodes are not compatible, and (ii) the probability of compatibility between an easy-to-match and a hard-to-match node is $1$. A straightforward coupling exercise shows that the match rate of H agents is larger in the simplified process than in the original process and their waiting time smaller. This allows us to analyze the simplified process instead of the original process. The bounds that we derive this way are in fact tight (up to vanishing factors) for the original process as well.\footnote{This can be proved formally using proof techniques similar to the ones we used to  analyze the greedy matching policy.}
Figures \ref{fig.bmatchrate} and \ref{fig.bwaitingtime} illustrate the convergence of these bounds to those for the greedy policy  (given in Proposition \ref{prop:greedy}) as $T$ approaches $0$.

The bounds given in  \autoref{thm.opt2} for E agents are calculated simply based on the fact that an arriving E agent should wait until the next matching period and may not get matched if she becomes critical before that. The bounds for E agents also are  tight for the original process and  approach to those of the greedy policy  as $T$ approaches $0$.






\subsection{On convergence rates and the effect of market imbalance}
\label{sec:imbalance}

In the previous section we have shown  that   greedy matching is optimal when the market is large. Here, using simulations, we
(i) explore  the convergence of match rates and waiting times under greedy matching as the market grows large for  a fixed $\lambda$, and (ii)  demonstrate  that  greedy matching is also optimal in small imbalanced markets with a large share of hard-to-match agents.

Figure \ref{fig:arrivalm} plots the match rates for both types of agents and the matching and waiting times  for a  fixed   $\lambda=1$ while varying the arrival rate  $m$. Observe that the measures converge quickly to their limit value as $m$ increases (note that the theory predicts that in the limit one half of the hard-to-match agents are matched).

\begin{figure}[H]
\centering
\includegraphics[scale=0.65]{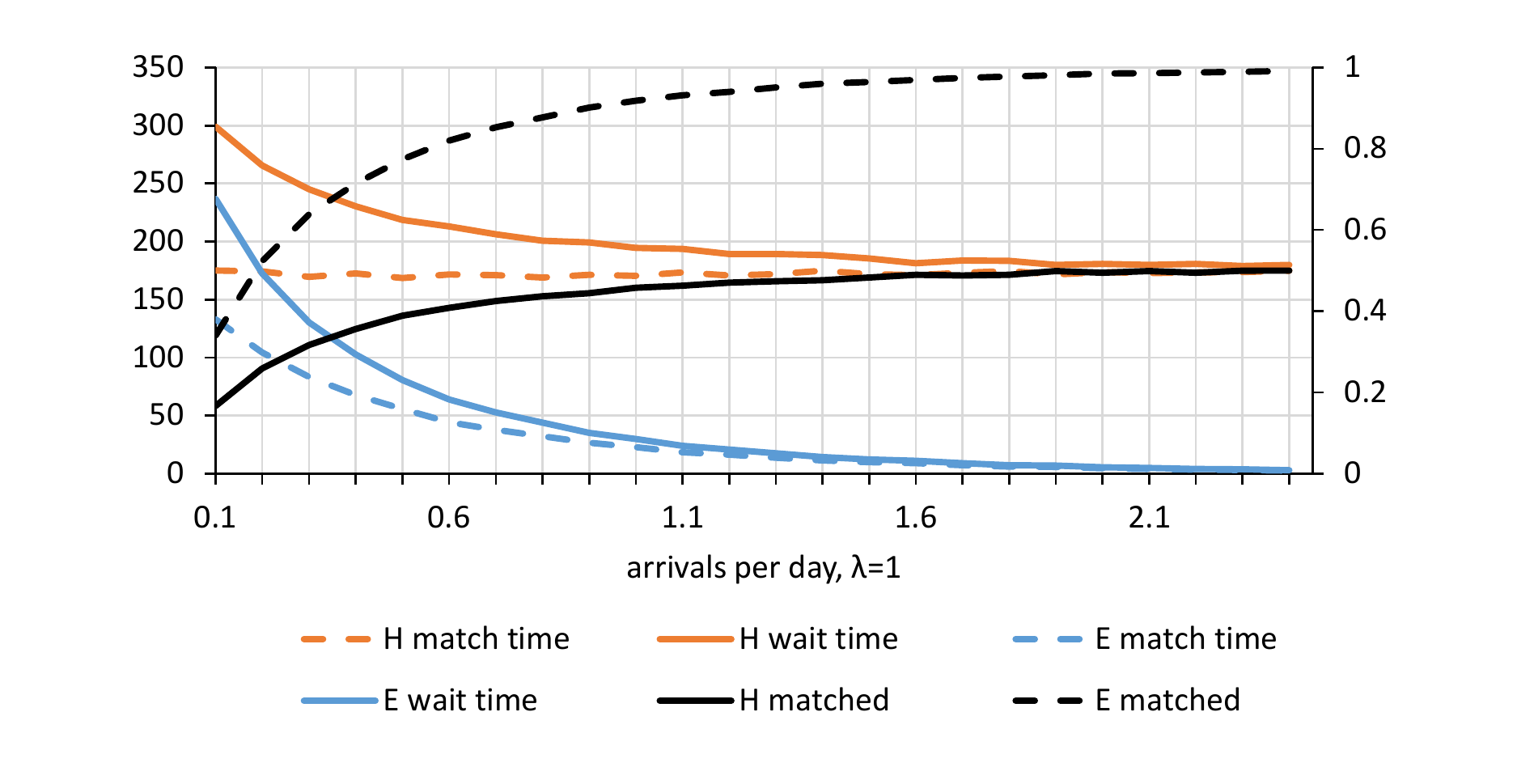}
\vspace{-0.5cm}
\caption{\label{fig:arrivalm} Sensitivity analysis over the arrival rate $m$ for the greedy matching policy. $\lambda$ is set to $1$ and agents depart (exogenously), on average, after $360$ days. The arrival rate in the NKR data is roughly $1$. The left y-axis represents times in days and the right y-axis represents the fraction matched.}
\end{figure}

Figure \ref{fig:arrivalm2} illustrates the optimality of greedy matching in small imbalanced markets.
It plots the same measures as we will see in Figure \ref{fig:arrivalm}, but this time  we fix a small arrival rate  $m=0.25$  and vary the imbalance parameter $\lambda$.
We observe that the average number of matches for both types of agents quickly converges to their upper bound as the imbalance parameter $\lambda$ increases.

To gain some intuition  about the effect of imbalance on the optimality of greedy matching, it  is useful to consider  the  {\em loss ratio}  of a policy, defined  to be $1$ minus the ratio of the expected number of agents matched under that policy per unit of time to the expected number of agents matched by the omniscient policy per unit of time.\footnote{The omniscient policy has access to the  whole sample path of the stochastic process, involving arrivals, departures, and compatibilities, and therefore makes the maximum number of matches.}

\begin{figure}[H]
\centering
\includegraphics[scale=0.65]{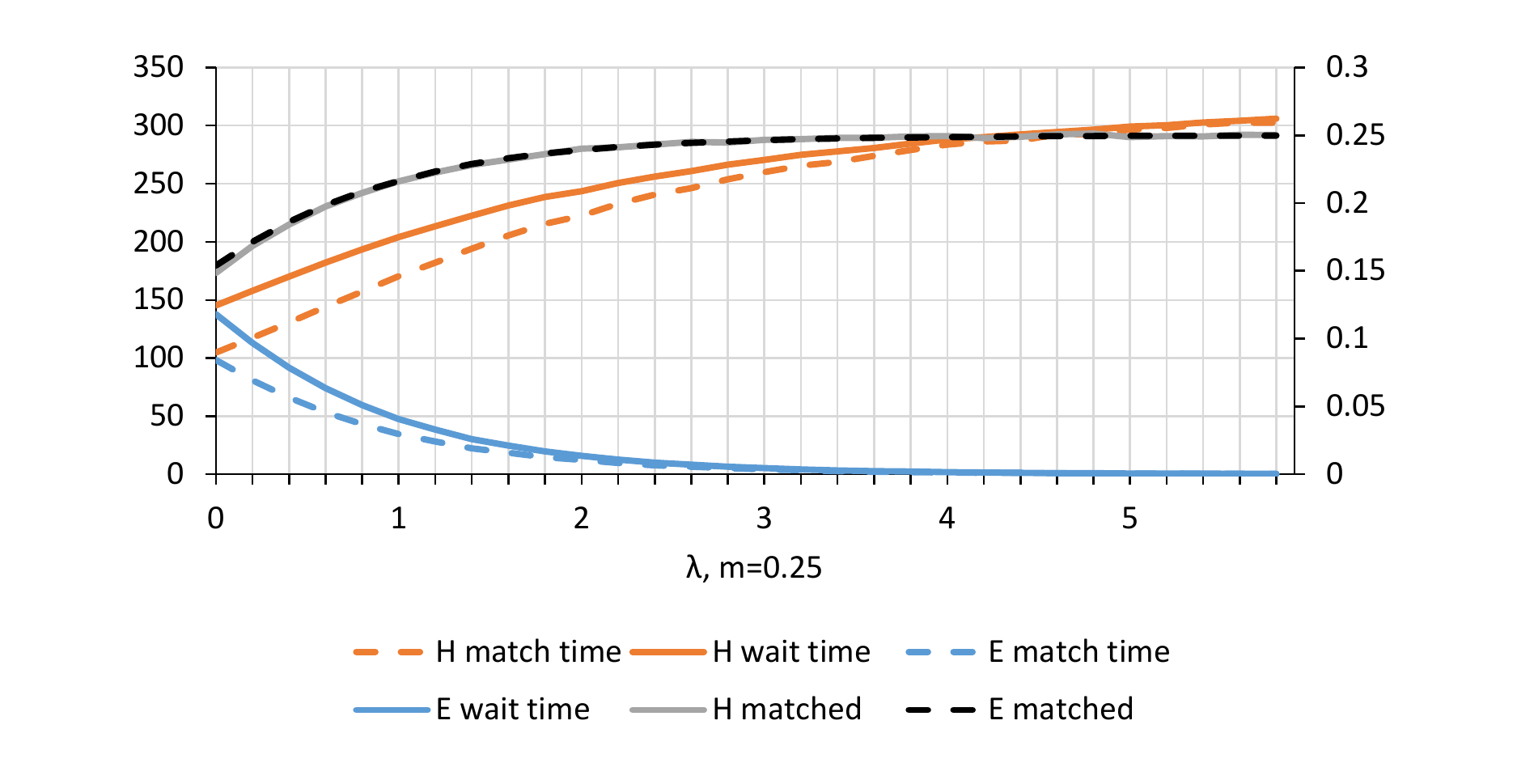}
\vspace{-0.5cm}
\caption{\label{fig:arrivalm2} Sensitivity analysis over the imbalance $\lambda$  for the greedy matching policy. $m$ is fixed to $0.25$ per day and average days in the pool is $360$.}
\end{figure}

We argue that the  loss ratio under greedy matching can be bounded  by $O(e^{-p\lambda m/2})$, which approaches  zero exponentially fast in the imbalance parameter, $\lambda$.
The key observation is that the steady-state distribution of the number of $H$ agents in the market  first-order stochastically dominates the Poisson distribution with parameter $\lambda m$. (This can be verified by a straightforward coupling argument.) Standard tail bounds of the Poisson distribution  imply that, upon the arrival of an E agent, the probability  that the number of H agents in the market is at least $\lambda m/2$ is $O(e^{-\lambda m})$. Conditioned on having at least $\lambda m/2$ H agents in the  market, the chance that an arriving E agent does not have a compatible H agent is bounded by $(1-p)^{-\lambda m/2}\leq e^{-p\lambda m/2}$. A union bound then implies that the probability  that an E agent is not matched with  an H  agent, and therefore the loss ratio, is bounded by $e^{-p\lambda m/2}+e^{-\lambda m}=O(e^{-p\lambda m/2})$.

\section{Empirical Findings}
\label{sec:simulationsData}

While we have formally shown that greedy matching is optimal in a large or imbalanced market, whether real markets are sufficiently large or imbalanced for this prediction to hold is ultimately an empirical question. To address this question we complement our theoretical findings with simulations using comparability data from a real kidney exchange platform (the National Kidney Registry, or NKR).
These simulations indicate that greedy matching dominates the batching  and  patient  policies even in relatively small markets. We further provide small market simulations based on the stylized model (and not real  data) in Appendix \ref{sec-simulations-stylized}, which also indicate that greedy matching  outperforms other commonly used policies.

The data from the NKR includes 1364 de-identified patient-donor pairs from July 2007 to December 2014 (the focus of the paper is on bilateral matching and we therefore ignore altruistic donors in the data).
The data includes patients' and donors' blood types and antigens as well the antibodies for each patient, which  allowed us  to verify (virtual) compatibility between each donor and each patient.  On  average, approximately one patient-donor pair arrives per day to the NKR, and the average exogenous departure time of a pair  is estimated to be 360 days.\footnote{Hazard rates vary slightly across pair types,  but for the sake of simplicity  we aggregate all pairs and used a simple hazard rate  model to estimate departures rate. For more detailed estimates see \citet{Misaligned}.} We note that our simulation results are very similar when merging the APD, UNOS, and NKR data.

In our simulations, arrivals of patient-donor pairs to the pool are generated by a Poisson process with a fixed arrival  rate. Arrival rates  are varied  from $0.01$ to $4$ pairs per day, capturing market sizes from one-tenth to four times the size of the NKR.  Varying the rate of arrival allows us to observe the effect of thickening the market exogenously (see also \citet{Misaligned}).  Each arriving pair is sampled uniformly at random (with replacement) from the NKR data.
Pairs depart from the pool according to an independent exponential random variable unless  matched prior to that point. The mean of the exponential random variable is set to $360$ (days), based on the empirical  estimate. The compatibility of two pairs is determined from real blood types, antigen and antibody compatibility.\footnote{In practice some patients can receive a kidney from a donor with an incompatible blood type, but for simplicity  of exposition we ignore this possibility here. Similar findings hold without this assumption after adjusting the classification of hard- and easy-to-match pairs in the data.}

We simulate greedy, patient, and batching matching policies until the arrival of the 200,000th pair  to the pool and report statistics for waiting time, matching time, and match rate by taking averages over all or a predefined subset of pairs  that belong to some collection of  types. A batching policy  matches  every  $T$ days the maximum number of pairs in the market, while prioritizing hard-to-match pairs. We experimented with batching frequencies $T=7,30,$ and $60$ days.\footnote{We also examined  weighted optimization using various weights and found no any qualitative differences. This is consistent with \citet{burqOR}, which found that  prioritization is negligible when there are more hard-to-match agents in the markets.}



Figure \ref{fig-fmwm} reports the fraction of matched pairs (left) and average waiting times (right) under  greedy,  patient, and batching matching policies for different arrival rates. Patient matching results in the highest fraction matched, and  the greedy and batching policies with $T=7$ result in a slightly lower match rate (large batch sizes lead to lower match rates). Moreover, the average waiting time under greedy matching is the smallest among all policies.
Table \ref{tab:alldata} reports the fraction of matched pairs, average waiting time, and average matching time across all pairs in the simulation. We also note that the difference between the average matching times under batching and greedy matching is between 10 and 40 days (depending on the arrival rate and the batch size).

An interesting empirical  observation is that waiting for more pairs to arrive does not increase the match rate (which can be seen by comparing greedy and batching policies as shown in Figure \ref{fig-fmwm}); however, increasing the arrival rate increases the match rate (see also \citet{frequency}).



\begin{figure}[htb!]
\centering
\begin{tabular}{cc}
\hspace{-1cm}
\subfloat[Fraction matched]{\label{fig:fm}\includegraphics[scale=0.5]{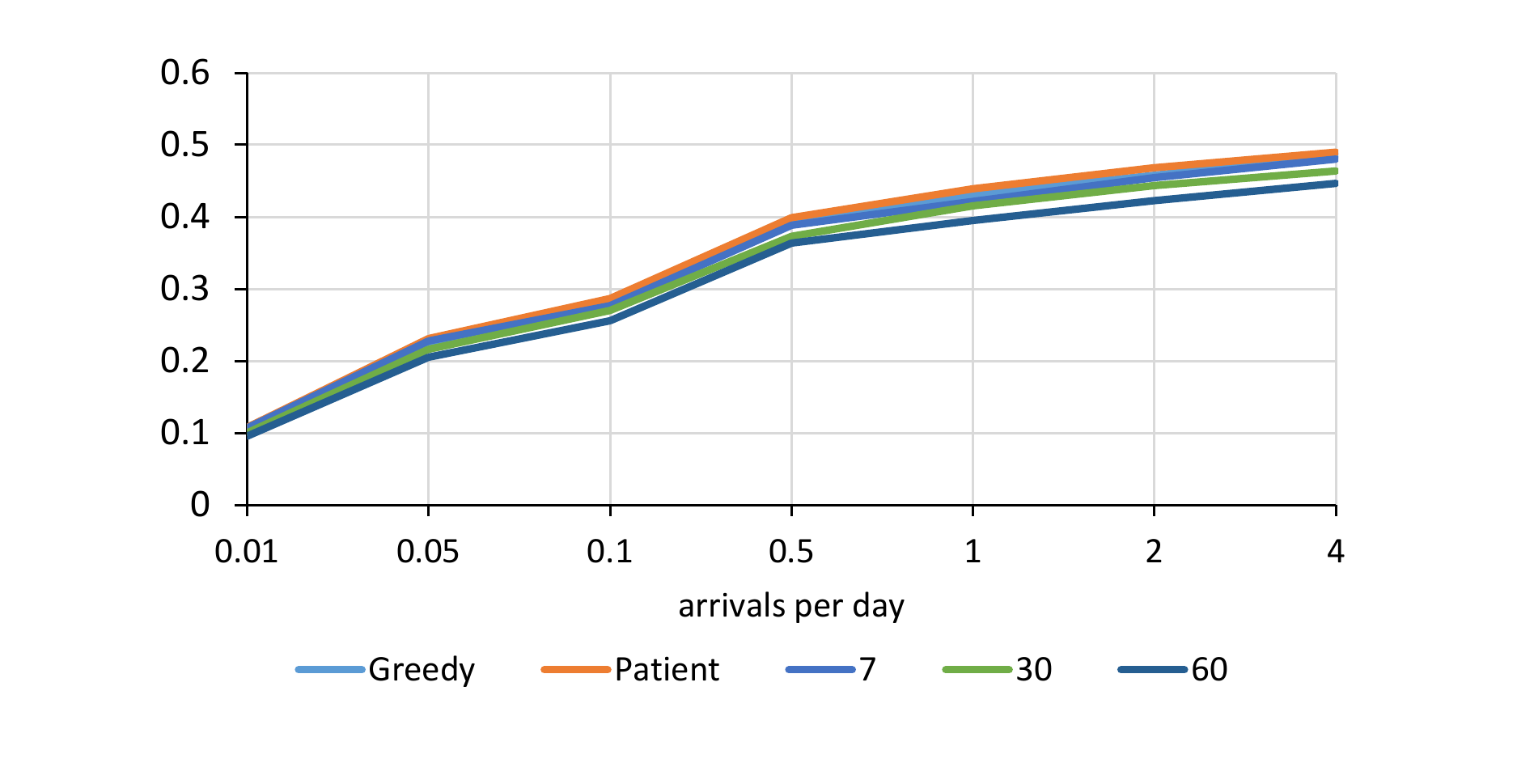}}%
\hspace{-1cm}
\subfloat[Average waiting time]{\label{fig:wt}\includegraphics[scale=0.5]{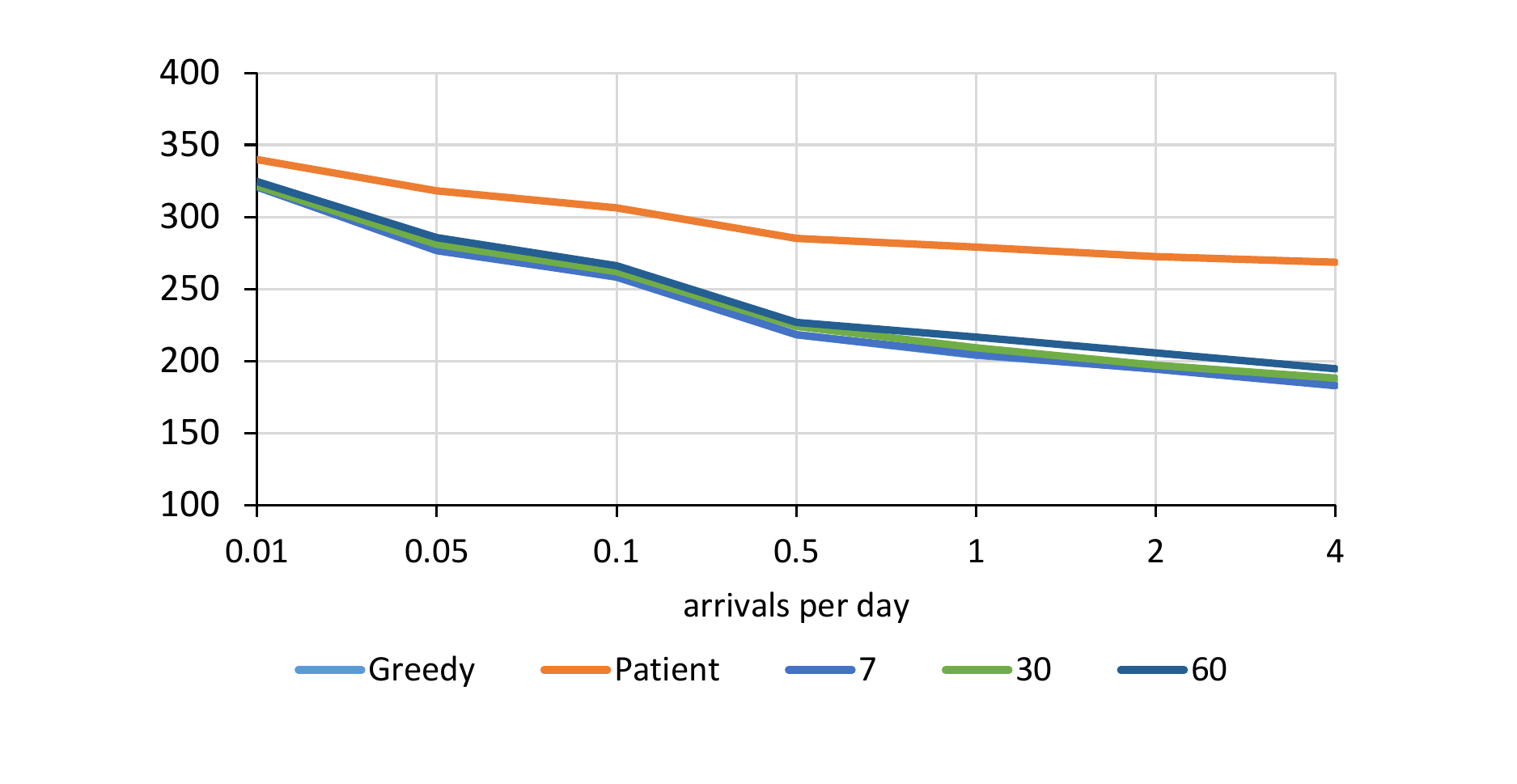}}
\end{tabular}
\caption{\footnotesize{Fraction of pairs matched (left) and average waiting time in days (right) under greedy, patient, and batching matching policies in simulations using NKR data. The x-axis represents the arrival rate measured by the number of pairs arriving on average per day.}}
\label{fig-fmwm}
\end{figure}

\begin{table}[H]
  \centering
  \footnotesize
\begin{tabular}{|l|*{27}{c|}}
  \hline
   arrival rate & \multicolumn{9}{c}{match rate} & \multicolumn{9}{|c}{matching time}&\multicolumn{9}{|c|}{waiting time} \\
    & \multicolumn{3}{c}{Greedy} & \multicolumn{3}{c}{Patient} & \multicolumn{3}{c}{Batch30} & \multicolumn{3}{|c}{Greedy} & \multicolumn{3}{c}{Patient} &  \multicolumn{3}{c}{Batch30} & \multicolumn{3}{|c}{Greedy} & \multicolumn{3}{c}{Patient} & \multicolumn{3}{c|}{Batch30} \\\hline

 \multicolumn{1}{|c|}{0.01} & \multicolumn{3}{c}{0.108} & \multicolumn{3}{c}{0.108} &  \multicolumn{3}{c}{0.10} &  \multicolumn{3}{|c}{150.07} &  \multicolumn{3}{c}{301.50}  &  \multicolumn{3}{c}{176.27} & \multicolumn{3}{|c}{320.82} & \multicolumn{3}{c}{329.90} & \multicolumn{3}{c|}{321.63} \\ \hline

\multicolumn{1}{|c|}{0.05} & \multicolumn{3}{c}{0.225} & \multicolumn{3}{c}{0.231} &  \multicolumn{3}{c}{0.216} & \multicolumn{3}{|c}{130.45} & \multicolumn{3}{c}{253.41} &  \multicolumn{3}{c}{164.17}  & \multicolumn{3}{|c}{277.84} & \multicolumn{3}{c}{318.33} & \multicolumn{3}{c|}{281.02} \\ \hline

\multicolumn{1}{|c|}{0.1} & \multicolumn{3}{c}{0.283} & \multicolumn{3}{c}{0.286} &  \multicolumn{3}{c}{0.27} & \multicolumn{3}{|c}{119.79} &  \multicolumn{3}{c}{233.50}  & \multicolumn{3}{c}{160.62}  &  \multicolumn{3}{|c}{258.10} & \multicolumn{3}{c}{306.44} &  \multicolumn{3}{c|}{261.98} \\ \hline

\multicolumn{1}{|c|}{0.5} & \multicolumn{3}{c}{0.391} & \multicolumn{3}{c}{0.392} &  \multicolumn{3}{c}{0.373} & \multicolumn{3}{|c}{98.91} & \multicolumn{3}{c}{ 204.53}  & \multicolumn{3}{c}{ 131.68}  & \multicolumn{3}{|c}{218.39} & \multicolumn{3}{c}{285.20 }  & \multicolumn{3}{c|}{224.13 } \\ \hline

                \multicolumn{1}{|c|}{1} & \multicolumn{3}{c}{0.431} & \multicolumn{3}{c}{0.439} &   \multicolumn{3}{c}{0.415} & \multicolumn{3}{|c}{92.12} & \multicolumn{3}{c}{198.36}  & \multicolumn{3}{c}{115.17}  & \multicolumn{3}{|c}{204.07} &  \multicolumn{3}{c}{279.23} & \multicolumn{3}{c|}{209.31}  \\ \hline

                \multicolumn{1}{|c|}{2} & \multicolumn{3}{c}{0.458} & \multicolumn{3}{c}{0.468} & \multicolumn{3}{c}{0.443} & \multicolumn{3}{|c}{85.63} & \multicolumn{3}{c}{192.00}  & \multicolumn{3}{c}{100.33} & \multicolumn{3}{|c}{194.59} &  \multicolumn{3}{c}{ 272.53}  & \multicolumn{3}{c|}{ 197.12} \\ \hline

                \multicolumn{1}{|c|}{4} & \multicolumn{3}{c}{0.485} & \multicolumn{3}{c}{0.489} &  \multicolumn{3}{c}{0.463} & \multicolumn{3}{|c}{77.13} & \multicolumn{3}{c}{183.51}  & \multicolumn{3}{c}{89.28} & \multicolumn{3}{|c}{183.59} & \multicolumn{3}{c}{ 268.62} & \multicolumn{3}{c|}{ 188.17} \\ \hline

\end{tabular}
   \caption{\footnotesize{Fraction of pairs matched, average matching time,  and average waiting times in days over all pairs in simulations using NKR data.} }
    \label{tab:alldata}
\end{table}%

We also compute the average matching time and waiting time  for different types of pairs. Figure \ref{fig.od-ud} reports these statistics for {\em over-demanded}  pairs (left) and {\em under-demanded}  pairs (right).  Under-demanded pairs are ABO incompatible with  each other and consist of the following types: O-X patient-donor pairs where  X$\neq$O; A-AB; and B-AB.\footnote{An X-Y patient-donor pair contains a patient with bloodtype X and a donor with bloodtype Y.} Over-demanded  pairs are bloodtype-compatible with each other and consist of the following  types: X-O patient-donor pairs where X$\neq$O; AB-A; and AB-B. Over-demanded and under-demanded pairs are roughly equivalent to easy-to-match and hard-to-match agents in our model; over-demanded pairs can potentially match with each other  as well as other with under-demanded pairs whereas under-demanded pairs can  match only with over-demanded pairs.  We note that some patients are much more sensitized (that is, they have a variety of antibodies that will attack foreign tissue) than others even within types, and in fact patients in over-demanded pairs are typically more sensitized than patients in under-demanded pairs.
\begin{figure}[htb!]
\centering
\begin{tabular}{cc}
\hspace{-1.2cm}
\subfloat[Over-demanded pairs]{\label{od-waiting}\includegraphics[scale=0.5]{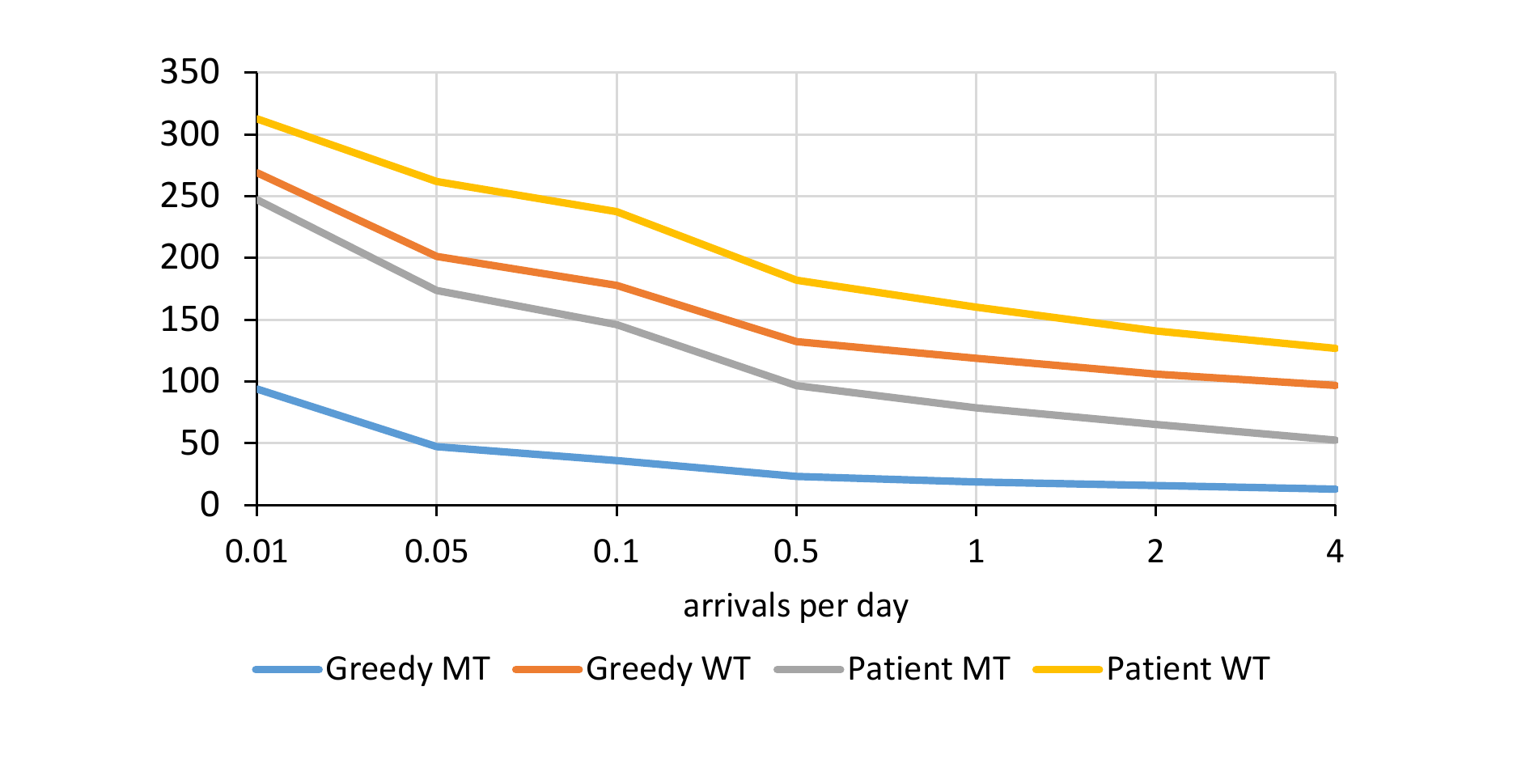}}
\hspace{-0.4cm}
\subfloat[Under-demanded pairs]{\label{ud-waiting}\includegraphics[scale=0.5]{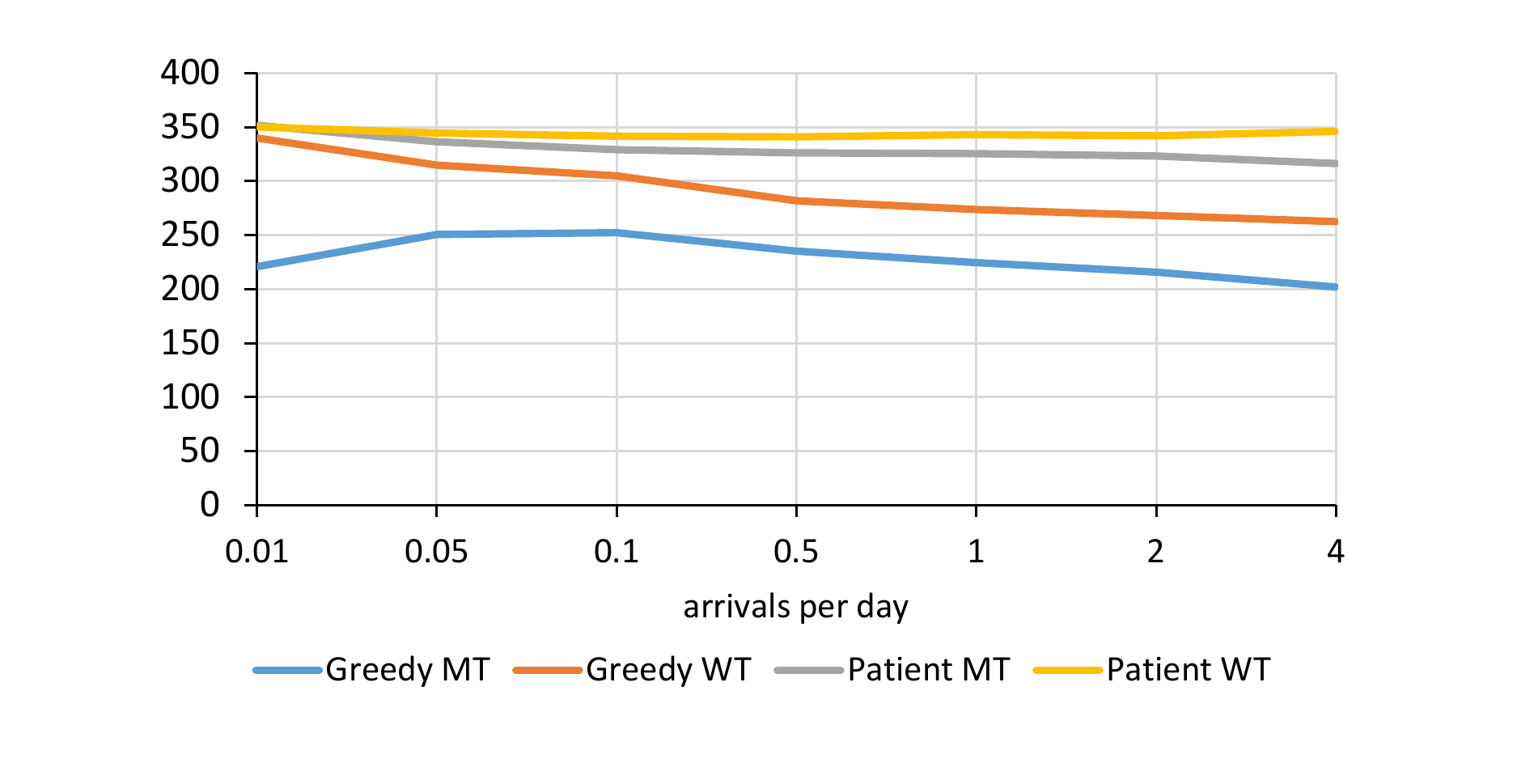}}
\end{tabular}
\caption{\footnotesize{Average matching times (MT) and waiting times (WT) in days under  greedy (G) and patient (P) matching policies in simulations using NKR data. The x-axis represents the arrival rate, which is the mean number of pairs arriving per day.}}
\label{fig.od-ud}
\end{figure}


Observe in Figure \ref{fig.od-ud} that the matching and waiting times of over-demanded pairs steadily decrease as the market becomes thicker, whereas the average matching and waiting times of under-demanded pairs changes little. This finding is in line with the  predictions from \autoref{thm.optimalpolicy2}. Despite the heterogeneity in the data,  the theoretical predictions (of the stylized two-type model) are aligned with the experiments when we categorize pairs as either  over-demanded or under-demanded. Moreover, the patterns hold even though patients belonging to over-demanded pairs are, on average, more sensitized than those in under-demanded pairs.\footnote{In fact, more than 40\% of patients in over-demanded pairs have less than a 5\% chance of being tissue-type compatible with a random donor. Furthermore, about 10\% of over-demanded pairs are not compatible with any other pair within this data set, which is why the average matching and waiting times do not drop all the way to zero.}
Figure \ref{fig-fmwt-95} plots the fraction  matched  and average waiting time of  over-demanded pairs in which the patient also has a Panel Reactive Antibody of at most 95 (that is, at least a 5\% chance of being tissue-type compatible with a random donor). So even though many of them are quite sensitized, almost all of them get matched and their waiting time is very low as the market grows large.

\begin{figure}[h!]
\centering
\includegraphics[scale=0.55]{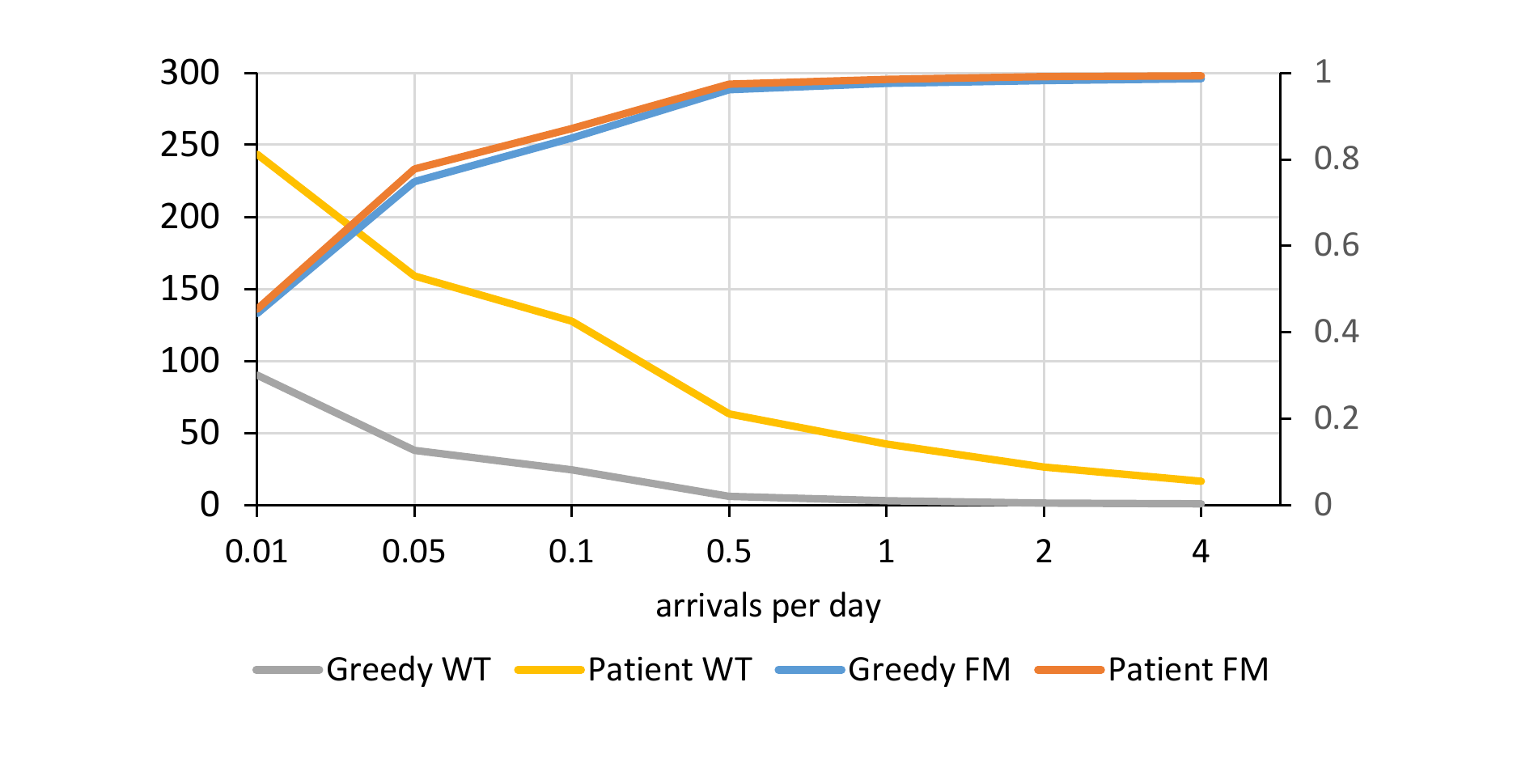}
\vspace{-0.5cm}
\caption{\footnotesize{Fraction of pairs matched (FM) and average waiting time (WT) for over-demanded pairs whose patients have PRA at most 95. The x-axis represents the arrival rate measured by the number of pairs arriving on average per day. The left y-axis represents the average waiting time in days matched and the right y-axis represents the average fraction matched.}}
\label{fig-fmwt-95}
\end{figure}

Next, we run  greedy and patient matching  under the base case scenario (with an arrival rate of  1 pair per period) until 700,000 pairs arrive, which again are drawn uniformly at random with replacement from the NKR data. For each  pair, we compute the average  waiting time over the copies of this pair sampled  in the simulation as well as the fraction of the copies that are matched (i.e., the empirical probability of getting matched). This gives, for each of the 1364 pairs in the original data set, an average waiting time and an empirical probability of being matched   under  both the greedy and patient matching policies.  The results appear given in Figure~\ref{fig.largesimul}. Figure \ref{st-pareto} shows that for each pair in the data set, the expected waiting time under greedy matching is shorter than the expected waiting time under patient matching (as all of the blue dots are above the 45$\degree$ line). This observation  suggests that the waiting time distribution  under greedy matching stochastically dominates the waiting time distribution  under patient matching.\footnote{A detailed analysis  of the simulation results confirms that this is indeed the case. We omit the details.}
Figure \ref{match-gp} reports, for  an arrival rate of $1$, the match rates under greedy and patient matching policies for each pair in the data. Observe that for most pairs the empirical probabilities of matching under the greedy and patient policies are  ``close" to each other.


\begin{figure}[htb!]
\centering
\begin{tabular}{cc}
\hspace{-1.2cm}
\subfloat[waiting times]{\label{st-pareto}\includegraphics[width=8cm,height=5.3cm]{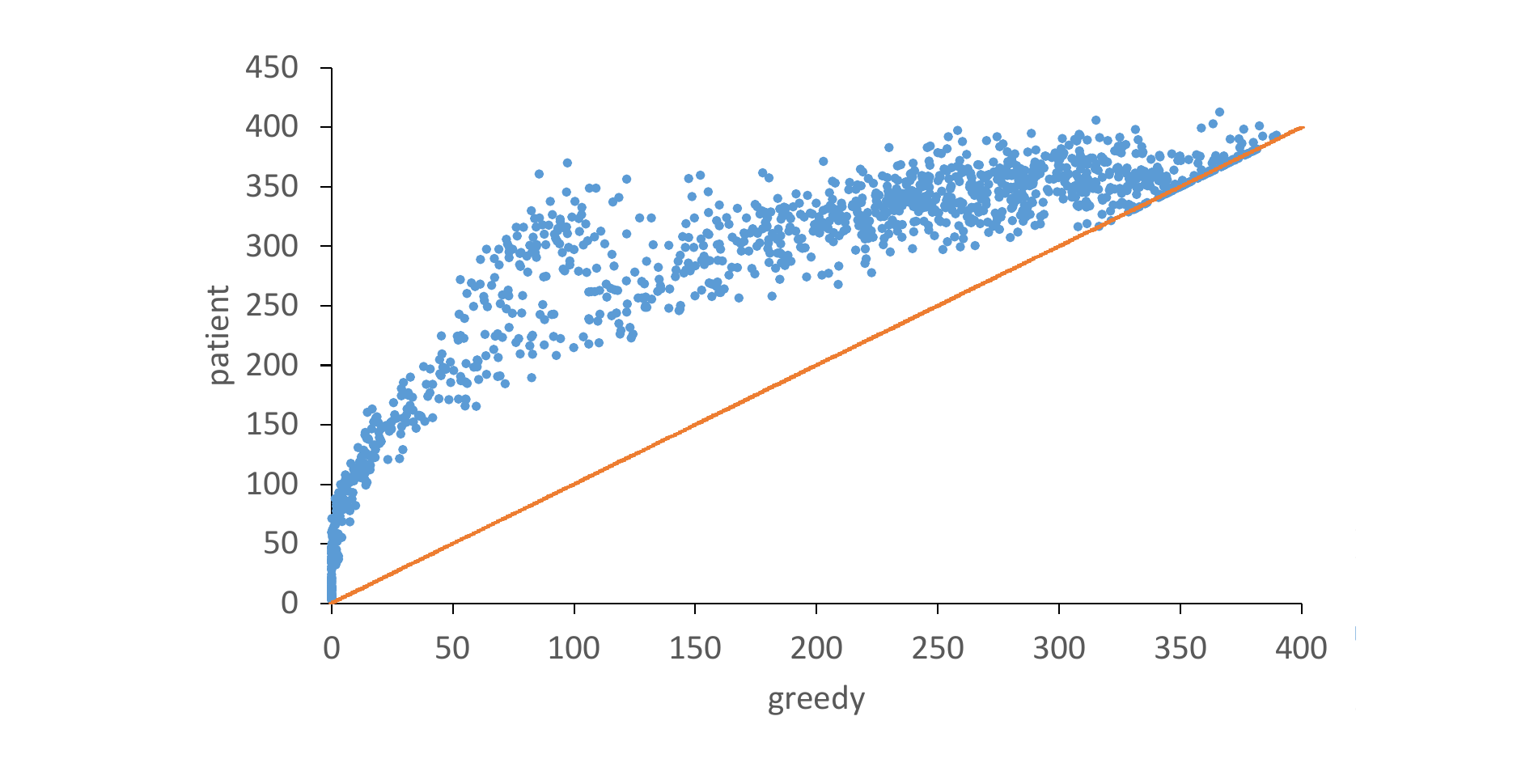}}%
\subfloat[chance of matching]{\label{match-gp}\includegraphics[width=8cm,height=5.3cm]{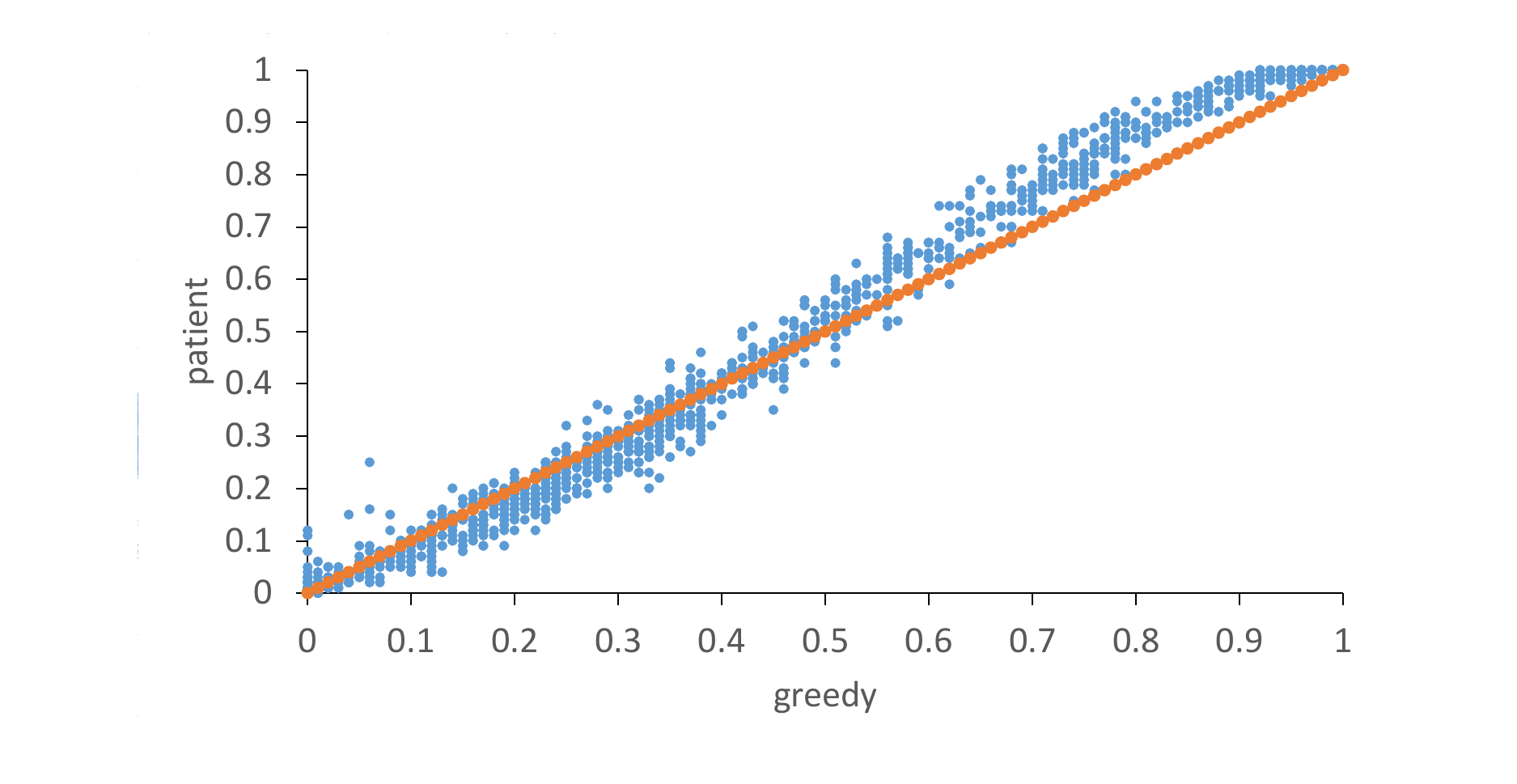}}
\end{tabular}
\caption{\footnotesize{Each dot represents one of the 1364 pairs in the data. The left figure scatters the average waiting time for each pair (averaged over its copies). The right figure  scatters the  empirical probability of being matched for each pair under the greedy and patient policies. The horizontal and vertical axes correspond to the greedy and patient policies,  respectively.}}
\label{fig.largesimul}
\end{figure}

Interestingly, Figure \ref{match-gp} suggests that, under greedy matching, the easy-to-match pairs are slightly worse-off because they are matched with slightly lower probability, whereas hard-to-match pairs are  better off.

\section{Final comments}

This paper studies matching policies in a random dynamic market, in which some agents are easier to match than others. We show theoretically as well as empirically that when the market is sufficiently large, the greedy  matching policy is optimal for all types of agents. This finding has direct practical implications for kidney exchanges that may not employ greedy matching policies out of concern  that greedy matching may harm highly sensitized patients.\footnote{See, e.g., \citet{ferrari2014kidney}.}

Our numerical simulations further provide evidence  that matching frequently does not harm the number of transplants even for practical market sizes. While we only simulated pairwise matchings and ignored frictions that occur in practice,\footnote{For instance, matches do not always translate into transplants due to refusals or blood test (crossmatch) failures.}
the simulations of \citet{frequency} account for such frictions and show, using simulations, that among batching policies, the policy with the shortest batching time window (essentially, the greedy policy) is optimal among a class of batching policies.


This paper has some limitations in the context of kidney exchange. First, we focus on matching only pairs of agents. Some kidney exchange programs, however, match pairs using chains initiated  by altruistic donors or cycles with three pairs. We note that  numerous programs have very low enrollment of nondirected (or altruistic) donors that  initiate chains and, in some  countries like France, Poland, and Portugal, chains are not even feasible  since altruistic donation is illegal \citep{eurokpd}.\footnote{A simple extension of our model to have directed links  will in fact predict that chains are not beneficial in large markets, which contrasts with the finding by  \citet{anderson2013efficient}. The difference stems from the fact that their model assumes vanishing probabilities. Moreover, this prediction is likely to hold in such an extended model with a {\it fixed} market size (but is hard to analyze).}



Finally we discuss a few technical limitations. First, matches are chosen uniformly at random. If, however, longer-waiting agents were to be assigned  a higher priority, the  waiting time and matching time  would not be distributed exponentially (while their average values remains the same). But as the market grows large, the match rates  converge to optimal under both policies, and we conjecture that the waiting and matching  times under greedy matching would stochastically dominate their counterparts under  patient and batching policies. Another limitation is that we focus on the number of matches and not the  quality of matches. Studying matching policies under heterogeneous match values remains an intriguing  challenge.

\appendix

\section*{Appendixes}

\section{Proofs from Section \ref{SEC.MODEL}}\label{sec.app.facts}
\begin{proof}[Proof of Proposition \ref{prop:good-fit}]
Note that as there are more E agents than H agents and H agents cannot match to themselves $\frac{2}{2+\lambda}$ is an upper bound on the fraction of agents which can be matched for any $m$. Note, that the size of the maximal matching (SMM) equals $\frac{2}{2+\lambda}$ if the bipartite graph with $m$ easy-to-match agents and $m$ hard-to-match agents on the other side admits a perfect matching. Becasue the probability that such a perfect matching exists converges to one as $m \to \infty$ (see for example Theorem 5.1 page 77 in \cite{frieze2015introduction}), it follows that $\text{SMM} \to \frac{2}{2+\lambda}$.

The probability that a hard-to-match agent has no partner is given by $(1-p)^m$. Because the compatibilities between hard-to-match and  easy-to-match agents are drawn independently the probability that all hard-to-match agents have at least one partner is given by $$(1-(1-p)^m)^{m (1+\lambda)}\,.$$
This probability converges to one as $m\to \infty$. The same argument shows that the probability that all easy-to-match agents have at least one partner converges to one.
\end{proof}

\begin{proof}[Proof of Proposition \ref{prop:no-simpler-model}]
The proof is by contradiction; suppose such $p(m)$ exists.
The chance that an agent has no other compatible agents is $(1-p(m))^m$. If $p(m)=O(1/m)$, then we have
$$(1-p(m))^m\leq e^{-m p(m)} = e^{-O(1)},$$
which implies that  \eqref{eq.f22} cannot be satisfied. Therefore, suppose that $p(m)=\frac{\omega(m)}{m}$, where $\lim_{m\to\infty} \omega(m)=\infty$. Next, we use this property to show that \eqref{eq.f21} cannot be satisfied.

The proof is constructive. We propose a simple algorithm that chooses a matching $\mu$ with size $|\mu|$ such that $\lim_{m\to\infty}\frac{|\mu|}{m}=1$.
Our algorithm is a greedy algorithm, defined as follows. It orders agents of the graph from $1$ to $m$, and visits the agents one by one. When visiting agent $i$, if there are no agents left that are compatible with agent $i$, then the algorithm passes agent $i$ and moves to agent $i+1$. Otherwise, the algorithm chooses one of the neighbors of agent $i$ arbitrarily, namely agent $j$, and adds the pair $(i,j)$ to the matching. The algorithm then visits the next available agent in the ordering. This process continues until the algorithm visits all agents.

We claim that the algorithm produces a matching $\mu$ which satisfies  $\lim_{m\to\infty}\frac{|\mu|}{m}=1$. Let $\phi(m)$ be a function that grows faster than $\frac{m}{w(m)}$ but slower than $m$. Then, during the algorithm, so long as there are $\phi(m)$ agents left in the graph, the chance that a visited node has no compatible agents is
$$(1-{w(m)}/{m})^{\phi(m)}\leq e^{-\frac{w(m)\phi(m)}{m}}=o(1).$$
That said, so long as there are  $\phi(m)$ agents left in the graph, the agent visited by the algorithm will be matched with a probability at least $q(m)$ where $\lim_{m\to\infty} q(m)=0$. By linearity of exception, the expected number of agents that are left unmatched by the end of the algorithm is then at most $\phi(m)+(m-\phi(m))\cdot q(m)$. Noting that $$\lim_{m\to\infty} \frac{\phi(m)+(m-\phi(m))\cdot q(m)}{m}=0$$ completes the proof.

\end{proof}

\begin{proof}[Proof of Proposition \ref{prop:upperbound}]
Consider a random process which is similar to the original random process (described by $\mc$), with the following differences:
\begin{enumerate}
\item Matches are made greedily upon arrival of nodes, and only between an easy-to-match and a hard-to-match node.
\item Easy-to-match nodes do not leave the pool before getting matched.
\item The probability of compatibility of an easy-to-match and a hard-to-match node is $1$.
\end{enumerate}

We represent the difference between the number of hard-to-match and easy-to-match nodes in this random process with a one-dimensional \MC{} $\mcl$, which is defined as follows. The state space of $\mcl$ is  $V(\mcl)=\mathbb{Z}$. The \MC{} is in state $x$ when the number of hard-to-match nodes minus the number of easy-to-match nodes is $x$. The transition rates from state $x$ to its left and right neighbors are respectively defined by
\begin{align*}
l_x= m + \max\{x,0\},
\hspace{0.5cm}
r_x= (1+\lambda)m.
\end{align*}
It is straightforward to verify that $\mcl$ is ergodic and therefore has a unique stationary distribution that we denote by $\pi$.
For notational simplicity, let $\over{x}=\EE{\pi}{\max\{x,0\}}$ and $\xstar=\lambda m$.

The expected waiting time under any policy $\tau$ is at least  $\frac{\over{x}}{m(1+\lambda)}$ and the expected match rate is at most $1-\frac{\over{x}}{m(1+\lambda)}$. The proof of this statement is by a straightforward coupling of the sample paths of $\mcl$ and the stochastic process corresponding to $\tau$. We skip the tedious details.

The proposition is proved in two steps. The first step shows that $\over{x}\leq \xstar +o(\xstar)$, and the second step shows that $\over{x}\geq \xstar -o(\xstar)$. This would prove the claim.

\paragraph{Proof of Step i.}
We start by writing the balance equations, according to which
\begin{align}
\frac{\pi_{i+1}}{\pi_i}= \frac{r_{i+1}}{l_i}.\label{eq.gbalance}
\end{align}

Suppose $i=\xstar+y$, for some $y>0$. Then,
\begin{align}
\frac{\pi_{i+1}}{\pi_i} = \frac{(1+\lambda)m}{x^*+y+m}=\frac{(1+\lambda)m}{x^*+y+m}=1-\frac{y}{x^*+y+m}.\nonumber
\end{align}
Let $n=x^*+m$. When $y>2n^{1/2}$, then the above equation implies that
\begin{align}
\frac{\pi_{i+1}}{\pi_i} &\leq 1-n^{1/2} = 1-n^{1/2}, \nonumber\\
\frac{\pi_{i}}{\pi_{x^*}} &\leq 1-n^{1/2}. \nonumber
\end{align}
The above equations imply that for any $i$ with $i-x^*>2n^{1/2}$ we have
\begin{align}
\frac{\pi_{i}}{\pi_{x^*}} &\leq \left(1-{n}^{-1/2}\right)^{i-x^*-2{n}^{1/2}} \nonumber\\
&\leq e^{-{n}^{-1/2}\cdot (i-x^*-2{n}^{1/2})}  \nonumber \\
&= e^{2-\frac {i-x^*}{\sqrt{n}} }\label{eq.rbnd}
\end{align}

Now, we can establish that $\over{x} \leq x^*+O(\sqrt{n})$. This is done by applying \eqref{eq.rbnd} as follows
\begin{align}
\bar{x} =\EE{x\sim\pi}{\max\{x,0\}}
&\leq x^*+ \sum_{i=0}^{\infty}  \pi_{x^*+i\sqrt{n}} \cdot (i+1)\sqrt{n} \nonumber\\
&\leq x^*+ \sum_{i=0}^{\infty}   e^{2-\frac {x^*+i\sqrt{n}-x^*}{\sqrt{n}} }  \cdot (i+1)\sqrt{n} \label{eq.nddapplybnd1}\\
&= x^*+ \sqrt{n} \sum_{i=0}^{\infty}   e^{2-i }  \cdot (i+1) = x^*+ O(\sqrt{n})=x^*+o(x^*) \label{eq.nddapplybnd2},
\end{align}
where \eqref{eq.nddapplybnd1} holds by \eqref{eq.rbnd},  and  \eqref{eq.nddapplybnd2} holds because $x^*> n^{\delta}$ for a $\delta>1/2$.

\paragraph{Proof of Step ii.} In this step we prove that $\bar{x}\geq x^*-o(x^*)$. The proof is similar to the previous step.
Now, suppose $i=x^*-y$, for some $y>0$. Equation \eqref{eq.gbalance} then implies
\begin{align}
\frac{\pi_{i}}{\pi_{i+1}} = \frac{x^*-y+\lambda}{n}=\frac{n-y}{n}=1-\frac{y}{n}.\label{eq.nddbl}
\end{align}
The above bound, followed by a calculation similar to how we  derived \eqref{eq.rbnd} implies
\begin{align}
\frac{\pi_{i}}{\pi_{x^*}}\leq e^{2-\frac {x^*-i}{\sqrt{n}}} \label{eq.lbnd}
\end{align}
for all $i<x^*$.

Also, observe that
\begin{align}
\frac{\pi_{i}}{\pi_{i+1}}= \frac{1}{1+\lambda} \label{eq.lbnd-negpart}
\end{align}
holds for all $i<0$.  Define $\pi_{<0}=\sum_{i<0}\pi_i$. Equations \eqref{eq.lbnd} and \eqref{eq.lbnd-negpart} imply that
\begin{align}
\pi_{<0}\leq e^{2-\frac{\xstar}{\sqrt{n}}}= o(1).\label{eq.lbnd-negpartsum}
\end{align}

Next, we use \eqref{eq.lbnd} and \eqref{eq.lbnd-negpartsum} to provide a lower bound for $\bar{x}$:
\begin{align}
\bar{x}=\EE{x\sim\pi}{\max\{x,0\}}  &\geq (1-\pi_{<0})\cdot x^* - \sum_{i=0}^{x^*/\sqrt{n}} \pi_{x^*-i\sqrt{n}} \cdot  (i+1) \sqrt{n}\nonumber\\
&\geq (1-\pi_{<0})\cdot x^* - \sqrt{n}\sum_{i=0}^{x^*/\sqrt{n}} \pi_{x_*} \cdot e^{2-i}  \cdot  (i+1) \label{eq.nddapplybnd3} \\
&=(1-\pi_{<0})\cdot  x^*-O(\sqrt{n})=x^* - o(x^*)\label{eq.nddapplybnd4}
\end{align}
where \eqref{eq.nddapplybnd3} holds by \eqref{eq.lbnd}  and \eqref{eq.nddapplybnd4} holds by \eqref{eq.lbnd-negpartsum}.

\end{proof}

\section{Proof of Theorem \ref{thm.optimalpolicy}}\label{sec.app.rdmp}
Theorem \ref{thm.optimalpolicy} is a direct consequence of Propositions  \ref{thm.opt2} and \ref{prop:upperbound}.
Proposition \ref{prop:upperbound} was proved in Section \ref{sec.app.facts}. To prove the theorem, we prove the rest of the propositions.
Before starting the proofs, we state some preliminaries in Section \ref{sec.mcprel}. The greedy, patient, and batching policies are then analyzed in Sections \ref{sec.app.mainganalysis}, \ref{sec.app.mainpresults}, and \ref{appendix.batching}, respectively.


\begin{proof}[Proof of Proposition \ref{thm.opt2}]
The first part of this proposition is a direct consequence of Proposition \ref{prop:greedy}.
The second part of the proposition is proved in  Section \ref{appendix.batching}.
The third part of the proposition is a direct consequence of Proposition \ref{prop:patient}.
\end{proof}

\begin{proof}[Proof of Proposition \ref{prop:greedy}]
The claim on the match rate under greedy matching is proved in \aref{lem.gmatchrate}. The claim on waiting time is proved in   \aref{lem.distofEGw8t} and \aref{lem.hard2matchwt}. These lemmas are stated and proved in Section \ref{sec.app.mainganalysis}.

\end{proof}

\begin{proof}[Proof of Proposition \ref{prop:patient}]
The claim on the match rate under  patient matching is proved in \aref{lem.pmatchrate}. The claim on the waiting time is proved in \aref{lem.distofEPw8t} and \aref{lem.distofHPw8t}. These lemmas are stated and proved in Section \ref{sec.app.mainpresults}.

\end{proof}


\section{Preliminaries} \label{sec.mcprel}
We use the terms {\em E pool} and {\em H pool} to denote the pools containing E agents and H agents, respectively.
The terms {\em E arrival} and {\em H arrival} respectively refer to the arrival of an E agent and the arrival of an H agent.
The term {\em departure clock} refers to the exponential random variable that determines the exogenous time that an agent becomes critical if she does not get matched prior to that time. When the departure clock {\em ticks}, the agent leaves the pool if she has not left with a match already.
We use the term {\em departure} to refer to the event of an agent leaving the pool, either because of being matched to another agent or because her departure clock ticks. Similarly, the terms {\em E departure} and {\em H departure} refer to the departure of an E agent and the departure of an H agent.

Each of the policies, greedy or patient, defines a continuous-time stochastic process whose state space is
$$\{(x,y): 0\leq x+y\leq C, \ x,y \in \mathbb{Z}\},$$
where $(x,y)$ denotes a state with $x$ hard-to-match and $y$ easy-to-match agents. These stochastic processes could be modeled as a Markov chain in the natural way. (The Markov chains would have the same state space as the above) We call the corresponding Markov chains the {\em greedy Markov chain}, and the {\em patient Markov chain}, respectively.

Under either of these policies, an agent might have to search for a compatible match. This happens in greedy matching upon the arrival of agents, and in patient matching upon their departure. In either of the policies, we suppose that an agent searches for a compatible match in the following way: she orders all the agents in the H pool in a random order,\footnote{That is, she draws a permutation over all the H agents, uniformly at random.} and gets matched to the first agent compatible with her, in that order. If no compatible agent is found, then the agent orders all the agents in the E pool in a random order  and gets matched to the first compatible agent in that ordering.
It also helps the analysis to define the notion of {\em offer}. When a searching agent $a$ checks her compatibility with another agent $b$, if $a$ and $b$ are compatible, then we say that $a$ makes an offer to $b$. By the definition of our policies, offers are always accepted. In our analysis, however, we sometimes consider scenarios in which an agent does not accept any offers. We always assume all offers are accepted unless we explicitly mention otherwise.


Suppose $f,g:\mathbb{R}_+\to\mathbb{R}_+$. We use the notation $g=o(f)$ when $\lim_{n\to\infty}\frac{g(n)}{f(n)}=0$ and  $g=O(f)$ when $\lim_{n\to\infty}\frac{g(n)}{f(n)} < \infty$.

We say an event $E$ holds with high probability (whp) when $\lim_{m\to\infty}\P{\bar{E}} = 0$. For notational simplicity, we also write this as
$\P{\bar{E}}=o(1)$. We say an event $E$ holds with low probability (wlp) when $\bar{E}$ holds whp.
We say an event $E$ holds with very high probability (wvhp) when $\P{\bar{E}} \leq e^{- m^{\alpha}}$, for some constant $\alpha>0$.
We say an event $E$ holds with very low probability (wvlp) when $\bar{E}$ holds wvhp.

\subsection{Markov chains}
In case of existence, we denote the unique stationary distribution of a Markov chain $\calN$ by $\pi(\calN)$. Sometimes we slightly abuse the notation and use $\pi$ to denote the  stationary distribution of a Markov chain that is clearly known in the context.

\begin{proposition}\label{pro.gamarnik}
\cite{anderson2013efficient} Suppose that $X_t$ is positive recurrent and that there exists $\alpha, \beta, \gamma > 0$, a set $B\subset \mathcal{X}$, and functions $U: \mathcal{X}\to \bbR_+$ and $f: \mathcal{X}\to \bbR_+$ such that for $x\in \mathcal{X}\backslash B$,
\begin{align}
\EE{x}{U(X_1)-U(X_0)} \leq -\gamma f(x), \label{eq.lyagamma}
\end{align}
and for $x\in B$,
\begin{align}
f(x) \leq \alpha,\\
\EE{x}{U(X_1)-U(X_0)} \leq \beta.
\end{align}
Then
\begin{align}
\E{f(X_{\infty})} \leq \alpha + \frac{\beta}{\gamma}.
\end{align}
\end{proposition}

\subsubsection*{Embedded Markov Chain}\label{sec.embeddedmc}
To apply Proposition \ref{pro.gamarnik}, we need to work with a discrete-time \MC. The Markov chains that we have referred to, however, so far have been continuous-time Markov chains. Rather than working with the continuous time Markov chain $\calN$, we typically work with a well-known
discrete-time Markov chain that is called the embedded Markov chain for $\calN$, which we denote by  $\what{\calN}$.
For completeness, we state the definition of $\what{\calN}$ below.
Let $\calN$ be a continuous-time Markov chain with transition rate $n_{i,j}$ from state $i$ to state $j$, for any $i\neq j$. Let $N$ be the {\em transition rate matrix} for $\calN$, i.e., $N_{i,j}=n_{i,j}$ for $i\neq j$, and the entries on the diagonal of $N$ are set so each row in $N$ sums to $0$.

\begin{definition}
The embedded Markov chain of $\calN$, denoted by $\what{\calN}$, is a discrete-time Markov chain with $V(\what{\calN})=V(\calN)$. The transition probability from state $i$ to state $j$ in $\what{\calN}$ is denoted by $\what{n}_{i,j}$ and is defined by
\begin{align*}
\what{n}_{i,j}=
\begin{cases}
\frac{n_{i,j}}{\sum_{k\neq i} n_{i,k}} & \textrm{if $i \neq j$}\\
0 & \mbox{if $i=j$}.
\end{cases}
\end{align*}
\end{definition}

Given that a finite-state Markov chain is ergodic, it is straightforward to show that its embedded Markov chain has a unique stationary distribution that we denote by $\pihat$ \cite{wikiemb}.
Intuitively, $\what{\calN}$ is a discrete-time Markov chain that ``behaves'' similar to $\calN$. We formalize the notion of similarity that we use in our analysis below.
\begin{definition}\label{def.mcapprxmts}
Let $\calN=\langle\calN_{m}\rangle_{m\geq 0}$ and  $\calN'=\langle\calN'_{m}\rangle_{m\geq 0}$ represent two infinite sequences of Markov chains where $V(\calN_{m})=V(\calN'_{m})$ and
$|V(\calN_{m})| < |V(\calN_{m+1})|$ for all $m\geq 0$. Suppose $\pi_m, \pi'_m$ respectively denote the unique stationary distributions for $\calN_m, \calN'_m$. We say $\calN$ approximates $\calN'$, and denote it by $\calN\threesim\calN'$, if there exist constants $m_0>0$ and $\ushort{\theta}, \bar{\theta}>0$ such that for all $m>m_0$ and all $x\in V(\calN_m)$ we have
$$ \ushort{\theta} \cdot \pi'_m (x) \leq \pi_m (x) \leq \bar{\theta}\cdot \pi'_m (x).$$
\end{definition}

The next lemma shows that, under certain conditions, the expected values of a function are not ``too far'' from each other, where the expectations are taken over the stationary distribution of a Markov chain and the stationary distribution of its embedded Markov chain.
We have not tried to state the most general version of this lemma, by any means; the following version comes in handy in the analysis.
\begin{lemma} \label{lem.nochangeinf}
Let $\calN=\langle\calN_{m}\rangle_{m\geq 0}$ and  $\calN'=\langle\calN'_{m}\rangle_{m\geq 0}$ be two infinite sequences of Markov chains with state spaces defined on $\mathbb{Z}$ such that $\calN\threesim \calN'$. Also, let $f:\mathbb{Z}\to \mathbb{R}_{+}$. Then, if \footnote{Recall that $o(m)$ denotes a term that grows with a rate slower than $m$.}
 $$ \EE{x\sim \pi(\calN_{m})}{f(x)} \leq  o(m),$$
we would also have
$$\EE{x\sim \pi(\calN'_{m})}{f(x)} \leq o(m).$$
\end{lemma}

\begin{proof}
Let $\pi,\pi'$ respectively denote $\pi(\calN_{m}), \pi(\calN'_{m})$, and $\ushort{\theta},\bar{\theta}$
be the coefficients defined in Definition \ref{def.mcapprxmts}. Then,
\begin{align*}
o(m) &\geq \EE{x\sim \pi}{f(x)}\\
&= \sum_{x} \pi(x) f(x)\\
&\geq \sum_{x} \ushort{\theta} \pi'(x) f(x) = \ushort{\theta} \cdot \EE{x\sim \pi'}{f(x)},
\end{align*}
which implies $\EE{x\sim \pi'}{f(x)} \leq o(m)$.
\end{proof}

\section{Analysis of Greedy Matching}\label{sec.app.mainganalysis}
In this section, we first analyze the stochastic process corresponding to the greedy policy. The main technical results established by this analysis are stated in \aref{sec.g.mainf}. Using these results we will be able to prove the results on the match rate and waiting time for greedy matching (mentioned in \aref{sec.mainresults}).
We establish the result on the match rate in \aref{sec.gmrate} and the waiting time result in \aref{sec.gw8t}.

\subsection{A Simplifying Assumption} \label{sec.finitestmc}
Before we start the analysis, we make an additional assumption for the sake of technical simplicity. As we will explain, our theorems quantitatively remain the same with or without this assumption.
\paragraph{The Finite Capacity Assumption.}
We suppose that the exchange program has capacity $C=\kappa\cdot m$, where $\kappa>0$ is a fixed constant that could be chosen arbitrary large. That is, if an arriving agent does not find a match immediately upon arrival and the total number of agents currently in the system is at least $\kappa\cdot m$, then that agent will not be added to the system.

For example, by letting $C=100 m$, we restrict the maximum possible number of agents at each moment in the system to be bounded by the expected number of pairs that arrive in the next $100$ time units.\footnote{Since $1$ to $5$ years is a reasonable estimate for each time unit in our model, the capacity in this case would be about the number of arriving agents in the next 100 to 500 years.} For practical purposes, this is a very reasonable assumption: through the analysis, we can see that the chance of having $100 m$ agents in the system is a low-probability event; more precisely, the stationary probability of this event is bounded by $2^{-m(96-2\lambda)}$). This fact also gives an intuitive explanation of  why our quantitative results do not rely on this assumption. When this assumption is dismissed, the proof would follow similarly while bearing some additional notation. To avoid   tedious notation, we present the analysis of greedy matching under this assumption.

\subsection{Modeling the dynamics}\label{sec.mcdef}

We  use a two-dimensional Markov chain, which we denote by $\mc$, to model the dynamics of the pool size.
First we set up some notation before proceeding to the description. For brevity, we use the abbreviation {\em\MC{}} for {\em Markov chain}.
For \MC{} \calM{}, let $V(\calM)$ denote the state space of $\calM$.
We represent each state by a pair $(x,y)$ where $x,y$ respectively denote the number of H agents and the number of E agents. In other words, we have
$$V(\mc)=\{(x,y): 0\leq x+y\leq C, \ x,y \in \mathbb{Z}\},$$
where we recall that the term $C$ is the capacity parameter defined in the Finite Capacity Assumption.

The definition of \mc{} would be completed by defining the transition rates. A transition can only happen from a state $(x,y)$ to its (at most) four {\em neighbors}, which are $$\{(x',y')\in V(\mc): |x-x'|+|y-y'|=1\}.$$
See Figure \ref{fig:2dMarkovchain} for a visual description of the neighbors.
To simplify the definition of transition rates from a node to its neighbors, we define the following notations:
Let $M_x=(1-p)^x$ and $N_y=(1-q)^y$.
For each node $(x,y)$, we denote the transition rates from this node to its four neighbors on the top, right, bottom, and left of it by $u_{x,y}, r_{x,y}, d_{x,y}, l_{x,y}$. These rates are defined as follows:
\begin{itemize}
\item If $x+y<C$, $u_{x,y}=m M_x N_y$ is the transition rate from the node $(x,y)$ to node $(x,y+1)$, otherwise $u_{x,y}=0$.
\item If $x+y<C$, $r_{x,y}=m(1+\lambda) N_y$ is the transition rate from the node $(x,y)$ to node $(x+1,y)$, otherwise $r_{x,y}=0$.
\item If $y>0$, then $d_{x,y}=m M_x (1-N_y) +m(1+\lambda)(1-N_y) + y $ is the transition rate from the node $(x,y)$ to node $(x,y-1)$, otherwise $d_{i,j}=0$.
\item If $x>0$, then $l_{x,y}=m(1-M_x) + x$ is the transition rate from the node $(x,y)$ to node $(x-1,y)$, otherwise $l_{x,y}=0$.
\end{itemize}
Given the above transition rates for $\mc${}, it is straightforward to verify that it models our stochastic process.

\begin{figure}
\centering
\begin{tikzpicture}[scale=0.8, inner sep=0, minimum size=17mm]
\tikzstyle{every node}=[circle,draw];
\node at (0,0) (b) {$x, y$};
\node at (-4.8,0) (a) {$x-1, y$};
\node at (4.8,0) (c) {$x+1, y$};
\node at (0, 4.8) (y) {$x, y+1$};
\node at (0, -4.8) (z) {$x, y-1$};

\node at (-0.3, 5.3) [draw=none] (up1) {};
\node at (0.3, 5.3) [draw=none] (up2) {};

\node at (-0.3, -5.3) [draw=none] (dn1) {};
\node at (0.3, -5.3) [draw=none] (dn2) {};

\node at  (-5.7,-.5) [draw=none] (d) {};
\node at  (5.7,.5) [draw=none] (e) {};
\node at  (-5.7,.5) [draw=none] (d1) {};
\node at  (5.7,-.5) [draw=none] (e1) {};
\node at  (-8,-.5) [draw=none] (d2) {};
\node at  (8,-.5) [draw=none] (e2) {};
 \begin{scope}[nodes = {draw = none}]
\draw[->,line width=1.5pt]
  (b) ->  (a) node [midway, above=-0.4cm]  {\small $m(1-M_x) + x$};
  \draw[->,line width=1.5pt]
  (b) ->  (c) node [midway, above=-0.4cm]  {\small $m(1+\lambda) N_y$};
  \draw[->,line width=1.5pt]
  (b) ->  (y) node [midway, right=0.1cm]  {\small $m M_x N_y$};
  \draw[->,line width=1.5pt]
  (b) ->  (z) node [midway, right=0.1cm]  {\small $m M_x (1-N_y) +m(1+\lambda)(1-N_y) + y $};
\end{scope}

\draw [dotted, line width=2.8pt] (a)+(-1.4,0) -- ++(-2.0,0);
\draw [dotted, line width=2.8pt] (c)+(1.4,0) -- ++(2.0,0);
\draw [dotted, line width=2.8pt] (y)+(0,1.4) -- ++(0,2.0);
\draw [dotted, line width=2.8pt] (z)+(0,-1.4) -- ++(0,-2.0);

\end{tikzpicture}
\caption{An illustration of the transitions from node $(x, y)$ Markov chain.}
\label{fig:2dMarkovchain}
\end{figure}
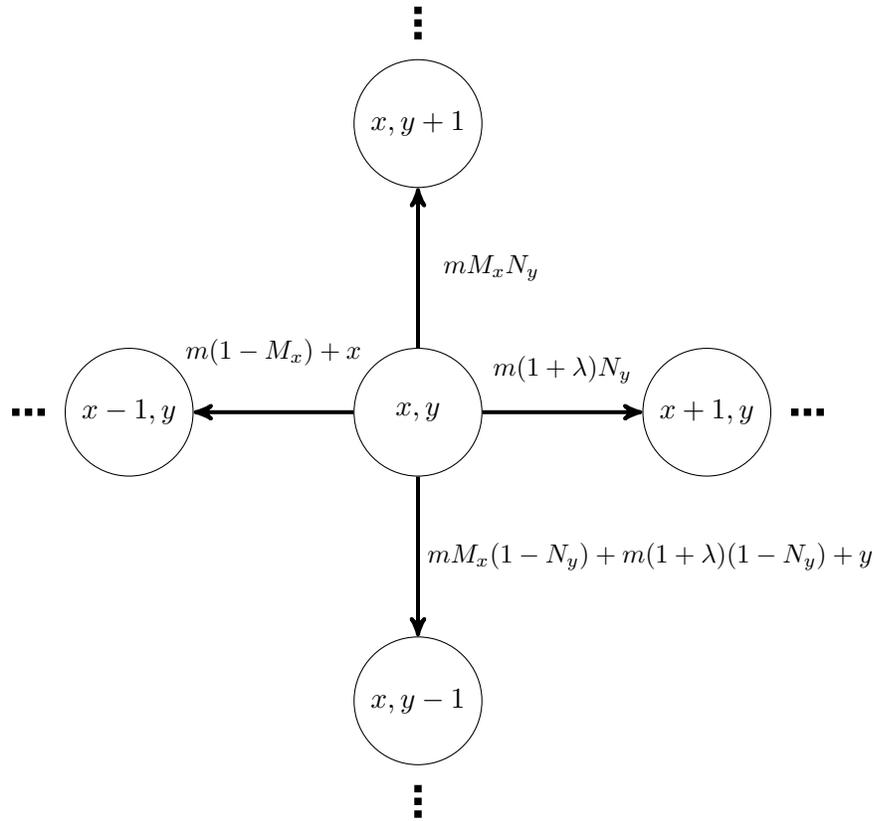

\begin{proposition}
\mc{} is ergodic and has a unique stationary distribution, $\pi$.
\end{proposition}
\begin{proof}
It is enough to show that \mc{} is ergodic, i.e., that it is irreducible and positive recurrent. Irreducibility is trivial: there is only one communication class.
To prove that \mc{} is positive recurrent, it is enough to note that it has a  finite state space; this implies that \mc{} is ergodic and has a unique stationary distribution. 
\end{proof}
We remark that the existence of the stationary distribution does not rely on the Finite Capacity Assumption (nor do the statements of our theorems, as mentioned earlier).

\begin{definition}
Let $\pi$ denote the stationary distribution of \mc{} and $\pi({i,j})$ denote the probability assigned to state $(i,j)$ in $\pi$. We define $\pi({i,j})=0$ when $i+j>C$. Let $\pi(i) = \sum_{j=0}^\infty \pi({i,j})$ and $\pi^j = \sum_{i=0}^\infty \pi({i,j})$.
\end{definition}

\subsection{The main analytical results for greedy matching}\label{sec.g.mainf}
The next two theorems are the core technical results for the analysis of greedy matching. They will be used for proving our results on the match rate and the waiting time distribution under greedy matching (Subsections \ref{sec.gmrate} and \ref{sec.gw8t}).
\begin{theorem}\label{thm.expected}
$\EE{\pi}{x} = \lambda m + O(\sqrt{m})$ and $\EE{\pi}{y} = O(1)$.
\end{theorem}

\begin{theorem}\label{thm.cncntrtn}
There exists constants $\theta_{\mathsf{X}}, \theta_{\mathsf{Y}} > 0$ such that for any $x,y\geq 0$, we have
\begin{align*}
\pi(x) &\leq e^{-\theta_{\mathsf{X}} \cdot\frac{ |x-x^*|}{\sqrt{m}} },\\
\pi(y) &\leq  e^{-\theta_{\mathsf{Y}} \cdot { y}},
\end{align*}
where  $\xstar=\lambda m$.
\end{theorem}

Computing $\EE{\pi}{y}$ and showing that $y$ is concentrated around its mean is the relatively easier part of Theorems \ref{thm.expected} and \ref{thm.cncntrtn}; this is done in \aref{sec.ycncntrtn}.
The main idea to prove the other parts (i.e., computing $\EE{\pi}{x}$ and showing that $x$ is concentrated around its mean) is
defining two other Markov processes that are coupled with $\mc$, namely $\mcl,\mcu$. These two processes are defined such that the number of unmatched hard-to-match agents in $\mc$ stochastically dominates the number of unmatched hard-to-match agents in $\mcl$ and is stochastically dominated by
the number of unmatched hard-to-match agents in $\mcu$.

In Sections \ref{sec.mcl} and \ref{sec.mcu}, we will show that
\begin{align}
\EE{\pil}{x} &= \lambda m + O(\sqrt{m}), \label{eq.mclexpc}\\
\EE{\piu}{x} &= \lambda m + O(\sqrt{m}),\label{eq.mcuexpc}
\end{align}
where $\pil, \piu$ respectively denote the (unique) stationary distributions of $\mcl, \mcu$. This implies that $\EE{\pi}{x} =  \lambda m + O(\sqrt{m})$.
To prove the concentration result, we will show that for the random processes $\mcl, \mcu$, there exist constants $\theta_{\mathfrak{l}}, \theta_{\mathfrak{u}}>0$ such that
\begin{align}
\pil(x) &\leq e^{-\theta_{\mathfrak{l}}} \cdot\frac{ |x-x^*|}{\sqrt{m}}, \label{eq.conlb}\\
\piu(x) &\leq e^{-\theta_{\mathfrak{u}}} \cdot\frac{ |x-x^*|}{\sqrt{m}}. \label{eq.conub}
\end{align}
This is done in Sections \ref{sec.mcl} and \ref{sec.mcu}.
Then, we use the fact that the PMF of $x$ in $\mc$ stochastically dominates the PMF of $x$ in $\mcl$ and is stochastically dominated by
the PMF of $x$ in $\mcu$. This fact, together with \eqref{eq.conlb} and \eqref{eq.conub}, implies that there exists a constant $\theta_{\mathsf{X}}>0$ such that
$$\pi(x) \leq e^{-\theta_{\mathsf{X}} \cdot\frac{ |x-x^*|}{\sqrt{m}} }.$$

We prove  \eqref{eq.mcuexpc} and \eqref{eq.conub} in \aref{sec.mcu}, and \eqref{eq.mclexpc} and \eqref{eq.conlb} in \aref{sec.mcl}. We introduced the  required notions and tools for the analysis in \aref{sec.mcprel}.

\subsection{Definition of $\mcu$} \label{sec.mcu}
Consider a random process which is similar to the original random process (described by $\mc$), with the following differences:
\begin{enumerate}
\item Matches are made only between hard-to-match nodes and easy-to-match nodes.
\item Matches are made greedily only upon arrival of easy-to-match nodes.
\item Easy-to-match nodes do not stay in the pool: they leave right after their arrival, if they are not matched with a hard-to-match node.
\end{enumerate}

We represent the number of hard-to-match nodes in this random process with a one-dimensional \MC{} $\mcu$, which is defined as follows. The state space of $\mcu$ is defined as
$V(\mcu)=\{0,1,2,\ldots\}$. The \MC{} is in state $x$ when the number of hard-to-match agents in the pool is $x$. The transition rates from state $x$ to its left and right neighbors are respectively defined by $l_x = m (1-N_x)+x$ and $r_x = m(1+\lambda)$. It is straightforward to verify that $\mcu$ is ergodic and therefore has a unique stationary distribution that we denote by $\piu$.

\begin{lemma}
$ \EE{\piu}{|x-x^*|} = O(\sqrt{m}) $.
\end{lemma}
\begin{proof}
The proof is by \aref{pro.gamarnik}.
First, we define the functions $f,U$ to be:
\begin{align*}
f(x) =&\, |x-x^*|, \\
U(x) =&\, (x-x^*)^2.
\end{align*}
Also, we define $B=\{x: |x-x^*| < \sqrt{m}\}$.
We  apply \aref{pro.gamarnik} to the  embedded Markov chain corresponding to $\mcu$, which we denote by $\what{\mcu}$. This will imply that
$$\EE{\what{\piu}}{|x-x^*|} = O(\sqrt{m}).$$
If this is proved, the lemma is then concluded by \aref{lem.nochangeinf};  note that \aref{lem.nochangeinf} is applicable since, by \aref{lem.ubapprx}, we have $\mcu\threesim \what{\mcu}$.

Next, we bound $\EE{{\what{\piu}}}{\Delta U(x)}$, by considering two cases: $x>x^*$ and $x<x^*$.
For notational simplicity, we will drop the index ${{\what{\piu}}}$ from the expectation in the rest of this proof. Also, let $S_x=l_x+r_x$.
When $x>x^*$, we have
\begin{align*}
\E{\Delta U(x)} &= \left(m(1+\lambda) - m(1-M_x) -x \right)\cdot \frac{2(x-x^*)}{S_x} + 2/S_x\\
&\leq   -\Omega({(x-x^*)}/{\sqrt{m}}),
\end{align*}
where the last inequality holds, since $x-x^*> \sqrt{m}$.

Now, suppose $x<x^*$. Then
\begin{align*}
\E{\Delta U(x)} &= \left(m(1-M_x) +x - m(1+\lambda) \right)\cdot \frac{2(x^*-x)}{S_x} + 2/S_x\\
&\leq -\Omega({(x-x^*)}/{\sqrt{m}}).
\end{align*}
Therefore, we can set $\gamma =\Omega(1/\sqrt{m})$.
Also, it is straightforward to verify that $\alpha=O(\sqrt{m})$ and $\beta=O(1)$, by the definition of $B$.
This implies that $\EE{\what{\piu}}{f(x)} = O(\sqrt{m})$, which proves the promised claim.
\end{proof}

\begin{claim}\label{lem.ubapprx}
Let $\what{\mcu}$ denote the embedded Markov chain of $\mcu$. Then, $\mcu\threesim \what{\mcu}$.
\end{claim}
\begin{proof}
Since $\mcu$ is ergodic, then  $\what{\mcu}$ has a unique stationary distribution. Let $\piu, \what{\piu}$ denote the stationary distributions of
 $\mcu$ and $\what{\mcu}$. Also, suppose $M$ denotes the transition rate matrix of $\mcu$. Then, $\piu$ can be written in terms of $\what{\piu}$ as follows \cite{wikiemb}:
 \begin{align*}
 \piu= -\frac{\what{\piu}  (\textrm{diag}(M))^{-1}}{\normone{\what{\piu} (\textrm{diag}(M))^{-1}}}.
 \end{align*}
 Thus, for any state $x$, we can write
 \begin{align}
 \piu (x)= \what{\piu}(x)\cdot \frac{  -1/M_{x,x}}{\normone{\what{\piu} (\textrm{diag}(M))^{-1}}}. \nonumber
 \end{align}
Given this equation, the lemma would be proved if there exist constants $\ushort{\theta},\bar{\theta}>0$ such that
 \begin{align}
\ushort{\theta} < \frac{  -1/M_{x,x}}{\normone{\what{\piu} (\textrm{diag}(M))^{-1}}} < \bar{\theta}.\label{eq.main4apprx}
 \end{align}

 To see why \eqref{eq.main4apprx} holds, first note that for any $x$, we have $-M_{x,x}=\Theta(m)$, by the Finite Capacity Assumption.
From this, we can imply that $\normone{\what{\piu} (\textrm{diag}(M))^{-1}}=\Theta(m)$, since the left-hand-side is a convex combination of the diagonal entries. These two facts together imply \eqref{eq.main4apprx}, which concludes the lemma.
\end{proof}

\subsection{Definition of $\mcl$} \label{sec.mcl}
Consider a random process which is similar to the original random process (described by $\mc$), with the following differences:
\begin{enumerate}
\item Matches are made greedily upon arrival of nodes, and only between an easy-to-match and a hard-to-match node.
\item Easy-to-match nodes do not leave the pool before getting matched.
\item The probability of compatibility of an easy-to-match and a hard-to-match node is $1$.
\end{enumerate}

We represent the difference between the number of hard-to-match and easy-to-match nodes in this random process with a one-dimensional \MC{} $\mcl$, which is defined as follows. The state space of $\mcl$ is  $V(\mcl)=\mathbb{Z}$. The \MC{} is in state $x$ when the number of hard-to-match nodes minus the number of easy-to-match nodes is $x$. The transition rates from state $x$ to its left and right neighbors are respectively defined by
\begin{align*}
l_x= m + \max\{x,0\},
\hspace{0.5cm}
r_x= (1+\lambda)m.
\end{align*}
It is straightforward to verify that $\mcl$ is ergodic and therefore has a unique stationary distribution that we denote by $\piu$.

\begin{lemma}
$ \EE{\pil}{|x-x^*|} = O(\sqrt{m}) $.
\end{lemma}
\begin{proof}
The proof is by \aref{pro.gamarnik}.
First, we define the functions $f,U$ to be:
\begin{align*}
f(x) =&\, |x-x^*|, \\
U(x) =&\, (x-x^*)^2.
\end{align*}
Also, we define $B=\{x: |x-x^*| < \sqrt{m}\}$.
We  apply \aref{pro.gamarnik} to the embedded Markov chain corresponding to $\mcl$, which we denote by $\what{\mcl}$. This would imply that
$$\EE{\what{\pil}}{|x-x^*|} = O(\sqrt{m}).$$
If this is proved, the lemma is then concluded by \aref{lem.nochangeinf}; note that \aref{lem.nochangeinf} is applicable since, by \aref{lem.ubapprx2}, we have $\mcu\threesim \what{\mcu}$.

Next, we bound $\EE{{\what{\pil}}}{\Delta U(x)}$, by considering two cases: $x > x^*$, and $x<x^*$.
For notational simplicity, we will drop the index ${{\what{\pil}}}$ from the expectation in the rest of this proof. Also, let $S_x=l_x+r_x$.
When $x>x^*$, we have
\begin{align*}
\E{\Delta U(x)} &= \left(m(1+\lambda) - m -x \right)\cdot \frac{2(x-x^*)}{S_x} + 2/S_x\\
&\leq   -\Omega({(x-x^*)}/{\sqrt{m}}),
\end{align*}
where the last inequality holds since $x-x^*> \sqrt{m}$.

Next, we consider the second case, $x<x^*$. See that
\begin{align*}
\E{\Delta U(x)} &= \left(m + \max\{x,0\} - m(1+\lambda) \right)\cdot \frac{2(x^*-x)}{S_x} + 2/S_x\\
&\leq   -\Omega({(x^*-x)}/{\sqrt{m}}).
\end{align*}
Therefore, we can set $\gamma =\Omega(1/\sqrt{m})$.
Also, it is straightforward to verify that $\alpha=O(\sqrt{m})$ and $\beta=O(1)$, by the definition of $B$.
This implies that $\EE{\what{\pil}}{f(x)} = O(\sqrt{m})$, which proves the claim.
\end{proof}

\begin{claim}\label{lem.ubapprx2}
Let $\what{\mcl}$ denote the embedded Markov chain of $\mcl$. Then, $\mcl\threesim \what{\mcl}$.
\end{claim}
\begin{proof}
The proof is identical to the proof of \aref{lem.ubapprx}.
\end{proof}

\subsection{Analysis of the Markov chain on the vertical axis}\label{sec.ycncntrtn}
\begin{lemma}\label{lem.ycncntrtn}
Let $\pi$ denote the stationary distribution of $\mc$. Then, $\EE{\pi}{y} = O(1)$. Moreover, $y$ is concentrated around its mean, in the following sense:
there exists $\theta_{\mathsf{Y}} > 0$ such that for any $y\geq 0$, we have
$\pi(y) \leq  e^{-\theta_{\mathsf{Y}} \cdot { y}}$.
\end{lemma}
\begin{proof}
It is enough to show that there exist constants $\ushort{y}>0$ and $\theta>1$ such that
for any $y\geq \ushort{y}$ we have $\pi(y)/\pi(y+1) >\theta$.

Let
$$\textsf{Up}(y) = \{u_{x,y}: x\geq 0\}, \ \ \textsf{Down}(y) = \{d_{x,y}: x\geq 0\}.$$
Also, let $\ushort{y}$ be the smallest positive integer for which $N_y < 1/3$. Now, see that
\begin{align*}
\frac{\min\{d: d\in \textsf{Down}(y+1)\}}{\max\{u: u\in \textsf{Up}(y)\}} &\geq \frac{m(1+\lambda)(2/3)}{m/3} > 2.
\end{align*}
Therefore, by the balance equations, we can set $\theta=2$.
\end{proof}

\subsection{Match rate under greedy matching}\label{sec.gmrate}

\begin{lemma}\label{lem.gmatchrate}
Under greedy matching, the  match rate of hard-to-match agents is $\frac{1}{1+\lambda} - O(\frac{1}{(1+\lambda)\sqrt{m}})$  and the match rate of  easy-to-match agents is $1-o(1)$.
\end{lemma}
\begin{proof}
First, we define some notation. Let $m_E=m$ and $m_H=(1+\lambda)m$.
We use $\yg_t, \xg_t$ respectively to denote the number of E, H agents under the greedy policy at time $t$.


For any type $\Theta\in \{\textrm{E}, \textrm{H}\}$, let ${\cal{G}}(\Theta)$ denote the match rate of $\Theta$ under greedy matching. The death rate (the rate of agents who depart the market unmatched) would then be $1-{\cal{G}}(\Theta)$. A straightforward calculation shows that
$$1-{\cal{G}}(\Theta)=\frac{\Theta^G}{m_{\Theta}}.$$
This is a consequence of the fact that the departure rate for all agents is equal to $1$. An application of linearity of expectation and the ergodic theorem imply the above equality. We then can use Theorem \ref{thm.expected} to write
$$1-{\cal{G}}(E)=\frac{E^G}{m}= O(1/m)$$
and
$$1-{\cal{G}}(H)=\frac{H^G}{(1+\lambda)m}= \frac{\lambda m+O(\sqrt{m})}{(1+\lambda)m}=\frac{\lambda}{1+\lambda}+O(1/\sqrt{m}).$$
This proves the claim for greedy matching.
\end{proof}

\subsection{Distribution of waiting time under greedy matching}\label{sec.gw8t}
We show that the waiting time for matching matched hard-to-match agents has an exponential distribution with rate $1+1/\lambda$.
This result is derived from the assumption that the ties between hard-to-match agents are broken randomly.
 As a consequence, we will be able to compute the expected waiting time of hard-to-match agents, conditioned on being matched or unconditionally. (Recall that the former is called the ``matching time" and the latter the ``waiting time".)

\begin{lemma}\label{lem.distofEGw8t}
Under the greedy policy, as $m$ approaches infinity, the waiting time and matching time of an E agent converge in distribution to the degenerate distribution at $0$.
\end{lemma}
\begin{proof}
First, we prove the result for waiting time.
Fix an E agent, $e$, and let $w_e$ denote the waiting time for $e$. For any fixed constant $t>0$, we will show that $\lim_{m\to\infty}\P{t>w_e}=0$. This will prove the claim. First, see that
\begin{align*}
\E{w_e} = \frac{1}{m}\cdot \EE{(x,y)\sim \pi}{x} = \frac{1}{m}\cdot O(1)
\end{align*}
where the second equality holds because of \aref{thm.expected}. Therefore, by Markov inequality, $$\P{w_e>t}<O(t/m)$$ holds, which means $\lim_{m\to\infty}\P{t>w_e}=0$.

Now, we prove the result on the matching time (i.e., the waiting time of matched agents).
Let $M_e$ denote the event in which $e$ leaves the pool with a match.
Our goal is to show that, for any fixed constant $t>0$,
$$\lim_{m\to\infty}\P{t>w_e|M_e}=0.$$
See that
$$\P{w_e=0|M_e}  = \frac{\P{w_e=0}}{\P{M_e}} = 1-o(1),$$
where the second equality holds because of \aref{thm.cncntrtn}. Therefore,
$$\P{w_e>0|M_e} = o(1),$$
which implies that
$$\P{w_e>t|M_e}=o(1).$$
This proves the claim.
\end{proof}

\begin{lemma}\label{lem.hard2matchwt}
As $m$ approaches infinity, the waiting time and matching time of hard-to-match agents converge in distribution to $\expdist(1+\frac{1}{\lambda})$.
\end{lemma}
We sketch the proof below. The formal proof is presented after the proof sketch.
\begin{proof}[Proof sketch]
Here we give a proof sketch.
We define a new process, namely $\calP$, in which there are no easy-to-match agents. Rather, an exponential clock is attached to each hard-to-match agent which ticks at rate $1/\lambda$. We suppose the agent is matched if the clock ticks before the agent departs.
It is not hard to show that the matching time of  agents in $\calP$ converges to the matching time of hard-to-match agents in the original process in distribution, as $m$ approaches infinity. 

Now, we compute the distribution for the waiting time of an agent in $\calP$.
Consider an agent $p$ and suppose it has entered the pool at time $t_0$.
Note that $p$ is matched iff it is matched before her departure clock ticks. Let $t_1, t_2$ be random variables such that $t_1\sim \exp(1/\lambda), t_2\sim\exp(1)$;  these random variables are interpreted as follows. The agent departs at time $t_0+t_2$ if she has not received any matches by then (i.e., her clock has not ticked). The time $t_0+t_1$ is the first time when the agent's {\em offer clock} ticks, i.e., the first time when the agent receives an offer. So, the agent is matched if and only if $t_1<t_2$. Alternatively, we can say the agent is matched iff $t_1=t_{\min}$ where $t_{\min}=\min\{t_1,t_2\}$.

First, see that $t_{\min}$ represents the waiting time of the agent. Therefore,
the waiting time of the agent  has distribution $\expdist(1+1/\lambda)$.

Next, we compute the distribution of the matching time.
Fix a constant $z>0$. The probability that an agent is matched before time $z$ conditioned on being matched is
\begin{align*}
\P{t_1 < z \big| t_1=t_{\min}} &= \frac{\P{t_{\min} < z \bigwedge t_1=t_{\min}}}{\P{t_1 =t_{\min}}}\\
&=\frac{\P{t_{\min} < z}\cdot \P{ t_1=t_{\min}}} { \P{t_1 =t_{\min}} } = \P{t_{\min} < z}.
\end{align*}
Since the above equality holds for any $z$, then the waiting time for an agent conditioned on being matched has the same distribution as the distribution of $t_{\min}$, which is $\expdist(1+1/\lambda)$.
\end{proof}

\begin{proof}[Proof of Lemma \ref{lem.hard2matchwt}]
First, we show that the waiting time of an H agent converges in distribution to $\expdist(1+1/\lambda)$. To this end, fix an H agent, $h$, upon her arrival at time $t_0$.
Let $\Q$ denote the stochastic process under greedy matching starting from $t_0$ and ending when $h$ leaves the system.
We couple another process with $\Q$, namely $\Q'$.
We define this coupling below. Roughly speaking, $\Q'$ is the same as the greedy matching process, with the exception that the arrival of $h$ is ``ignored'' in the sense that $h$ does not interfere with the evolution of $\Q'$.
\begin{itemize}
\item $\Q'$ runs from time $t_0$ to $t_0+\log m$. Furthermore, the departure clock of $h$ is set to tick at time $t_0+\log m$ in $\Q'$.
\item If $h$ finds a compatible match upon her arrival (at time $t_0$) in $\Q$, then we stop both $\Q,\Q'$. Otherwise,
we let $\Q$ evolve according to the greedy process. By definition, $\Q'$ has a sample path identical to $\Q$, until one of the following disjoint events happens:
\begin{enumerate}[{Event} (i)]
\item  $h$ receives an offer in $\Q$ before time $t_0+\log m$. In this case, we stop $\Q$. In $\Q'$, $h$ rejects the received offer as well as all  offers  she will receive in the future. Any E agent who gets rejected by $h$ will make an offer to the  next compatible agent in his (random) list. $\Q'$ will continue evolving according to the greedy process, with the exception of agent $h$, who does not interfere with the process.
\item The departure clock of $h$ ticks in $\Q$ before time $t_0+\log m$ and $h$ departs without being matched. In this case, we stop $\Q$ but continue to run  $\Q'$. $\Q'$ will continue evolving according to the greedy process, with the exception of agent $h$, who does not interfere with the process (i.e., agent $h$ rejects all the offers she receives, in the sense clarified above).
\end{enumerate}
\item We stop $\Q'$ when it reaches  time $t_0+\log m$.
\end{itemize}

For notational simplicity, we suppose $t_0=0$ without loss of generality.
Let $E_h(t)$ denote the event in which $h$ leaves the pool in $\Q$ before time $t$ (either matched or unmatched).
Also, let $E'_h(t)$ denote the event in which $h$ receives at least one offer in $\Q'$ before time $t$. For any constant $t>0$, we have
\begin{align}
\lim_{m\to\infty} \P{\over{E_h(t)}}  = e^{-t} \cdot \lim_{m\to\infty} \P{\over{E'_h(t)}} .\label{eq.ephconverges}
\end{align}
Therefore, if we show that $\lim_{m\to\infty} \P{\over{E'_h(t)}} = e^{-t/\lambda}$, we can use \eqref{eq.ephconverges} to imply that
\begin{align*}
\lim_{m\to\infty} \P{\over{E_h(t)}}  = e^{-t\cdot(1+1/\lambda)},
\end{align*}
which proves the claim on the distribution of waiting time. To complete the proof, the following claim must be proved.
\begin{claim}\label{clm.gwtdist}
For any constant $t>0$,
\begin{align}
\lim_{m\to\infty} \P{\over{E'_h(t)}}  = e^{-t/\lambda}.\label{eq.ephconverges}
\end{align}
\end{claim}
\begin{proof}
Our proof approach is as follows. First, we observe that the process $\Q'$ can be run from time $t_0=0$ to $\log m$ as follows: sample a state $(x,y)\sim \pi$. Then, let the stochastic system start from $(x,y)$ and evolve for $\log m$ units of time, under the greedy policy. By the PASTA property,\footnote{PASTA, or Poisson Arrivals See Time Averages, is a well-known property in the queueing literature; e.g., see \cite{Harchol-Balter}.} the sample paths generated by this process are identical to the sample paths of $\Q'$ (note that the arrival of $h$ is ``ignored'' in $\Q'$, in the sense that $h$ does not affect the evolution of $\Q'$).

By the above argument, we sample the state $(x,y) \sim \pi$ at time $t_0=0$, and let the process run until time $\log m$. We also consider an ``imaginary agent'' $h$, which exists in the H pool, but rejects all of the proposals that are made to her. Our goal is to show that the probability that $h$ receives no proposals in the period $[0, t]$, which we know by $\P{\over{E'_h(t)}}$, approaches $ e^{-t/\lambda}$ as $m$ approaches infinity.

Let $\H_t$ denote the history of the process until time $t$. (Note that the history does not include the offers made to $h$,  this agent does not change the evolution of the process.)
Let the random variable $n_t$ denote the number of E agents who arrive to the pool in the time interval $[0,t]$. Also, let $a_1,\ldots,a_{n_t}$ denote the arrival times of these agents. Define $x_1,\ldots,x_{n_t}$ to be the number of H agents in the pool at times $a_1,\ldots,a_{n_t}$, respectively.
Now, see that
\begin{align}
\P{{\over{E'_h(t)}} \big| \H_t} &\geq \Pi_{i=1}^{n_t} (1-\frac{1}{x_{i}}),\label{eq.mainHwinHgreedy1}\\
\P{{\over{E'_h(t)}} \big| \H_t} &\leq \Pi_{i=1}^{n_t} (1-\frac{1-(1-p)^{x_i}}{x_{i}}),\label{eq.mainHwinHgreedy2}
\end{align}
where \eqref{eq.mainHwinHgreedy1} and \eqref{eq.mainHwinHgreedy2} hold because the chance that a match happens at time $a_{i}$ is at most $1$ and at least ${1-(1-p)^{x_i}}$.

In the rest of the proof, we will use probabilistic bounds to show that in almost all histories $\H_t$, $n_t$ is close to $tm$ and $x^*$ is close to $\lambda m$, where $x^*=\min\{x_1,\ldots,x_{n_t}\}$. This will let us simplify \eqref{eq.mainHwinHgreedy1} and \eqref{eq.mainHwinHgreedy2}. The simplified forms will expose that in almost all histories $\H_t$,  $\P{{\over{E'_h(t)}}}$ approaches $e^{-t/\lambda}$ as $m$ approaches infinity.

Define  event $G_1$ as
\begin{align*}
G_{1}: \ \lambda m - \sqrt{m} \log^2 m \leq  x^* \leq \lambda m + \sqrt{m} \log^2 m.
\end{align*}
We will show that $\bar{G}_{1}$ happens wlp. To do this, we use the fact that the size of the H pool after the arrival of any H agent is close to $\lambda m$ (in the sense of Theorem \ref{thm.cncntrtn}). Formally, we use the PASTA property together with  Theorem \ref{thm.cncntrtn} and write a union bound over all arrivals of E agents in the  interval $[0,t]$. This implies that $\bar{G}_{1}$ happens wlp. More precisely, this implies that
\begin{align}
\P{\bar{G}_1} \leq e^{-O(\log^2 m)}.\label{eq.g1bingreedy}
\end{align}

Define  event  $G_2$ as
\begin{align*}
G_{2}: \ t m  - \sqrt{tm} \log m \leq  n_t \leq t m +\sqrt{tm} \log m .
\end{align*}
By  Chebyshev's inequality,
\begin{align}
\P{\bar{G}_{2}} \leq \log^{-2} m.\label{eq.g2bingreedy}
\end{align}
Note that the above inequality holds because Poisson distribution has equal mean and variance; in this case, the random variable $n_t$ has mean and variance $tm$.

Now, recall \eqref{eq.mainHwinHgreedy1} and \eqref{eq.mainHwinHgreedy2}, and that our goal is to show that in almost all histories $\H_t$, $n_t$ is close to $tm$ and $x^*$ is close to $\lambda m$. Inequalities \eqref{eq.g1bingreedy} and \eqref{eq.g2bingreedy} show just this.
More formally, they imply that in almost all histories $\H_t$ but a fraction $O(1/\log^2 m)$ of them, the events $G_1,G_2$ hold.
Therefore, we can use \eqref{eq.mainHwinHgreedy1} and \eqref{eq.mainHwinHgreedy2} to write
\begin{align*}
\P{{\over{E'_h(t)}} } &\geq (1-O(1/\log^2 m)) \cdot e^{-\frac{tm}{\lambda m}},\\
\P{{\over{E'_h(t)}} } &\leq e^{-\frac{tm }{\lambda m}\cdot (1-o(1))}.
\end{align*}
 The above equations imply that
$$ \lim_{m\to\infty} \P{\over{E'_h(t)}}  = e^{-t/\lambda}.$$
The proof is complete.
\end{proof}

To complete the proof of the lemma, it remains to show that the waiting time of matched H agents converges in distribution to $\expdist(1+\frac{1}{\lambda})$.
We follow the same idea used in the proof sketch.
Again, we fix an H agent, $h$, who arrives at time $t_0=0$. Define the (coupled) processes $\Q,\Q'$  as before.
Let $t_1$ be a random variable that denotes the first time at which $h$ receives an offer in $\Q'$. If $h$ does not receive an offer in $\Q'$, let $t_1=\log m$.
Also, let $t_2$ be an (independent) exponential random variable with rate $1$. This variable denotes the time at which the departure clock of $h$ ticks in $\Q$.

Define the random variable $t_{\min}=\min\{t_1,t_2\}$.
For any constant $z>0$, observe that
\begin{align*}
&\P{t_1 < z \big| \textrm{$h$ receives an offer in $\Q$}} =\\
& \P{t_1 < z \big| \textrm{$h$ receives an offer in $\Q$ before time $\log m$}}+\\
&\P{t_1 < z \big| \textrm{$h$ receives an offer in $\Q$ after time $\log m$}}.
\end{align*}
Therefore, we can write
\begin{align}
& \lim_{m\to\infty}\P{t_1 < z \big| \textrm{$h$ receives an offer in $\Q$}} \nonumber\\
=& \lim_{m\to\infty} \P{t_1 < z \big| \textrm{$h$ receives an offer in $\Q$ before time $\log m$}}\nonumber\\
=& \lim_{m\to\infty} \P{t_1 < z \big| t_1=t_{\min}}.\label{eq.gwtconglemma}
\end{align}

Next, observe that
\begin{align*}
\lim_{m\to\infty} \P{t_1 < z \big| t_1=t_{\min}} = \lim_{m\to\infty} \frac{\P{t_{\min} < z \bigwedge t_1=t_{\min}}}{\P{t_1 =t_{\min}}}.
\end{align*}
Now, recall from the first part of the proof that $t_1$ converges in distribution to $\expdist(1/\lambda)$. (This is essentially due to Claim \ref{clm.gwtdist}).
Now, note that since $t_2$ converges in distribution to $\expdist(1/\lambda)$, then $t_{\min}$ converges in distribution to $\expdist(1+1/\lambda)$.
Therefore, since $t_1,t_2$ are independent random variables and $t_{\min}=\min\{t_1,t_2\}$, we can simplify the above equation further as follows.
\begin{align*}
\lim_{m\to\infty} \P{t_1 < z \big| t_1=t_{\min}} &= \lim_{m\to\infty} \frac{\P{t_{\min} < z \bigwedge t_1=t_{\min}}}{\P{t_1 =t_{\min}}}\\
&= \lim_{m\to\infty} \frac{\P{t_{\min} < z}\cdot \P{ t_1=t_{\min}}} { \P{t_1 =t_{\min}} } = \lim_{m\to\infty} \P{t_{\min} < z} = e^{-z (1+1/\lambda)}.
\end{align*}
The above equation together with \eqref{eq.gwtconglemma} implies that
\begin{align}
\lim_{m\to\infty}\P{t_1 < z \big| \textrm{$h$ receives an offer in $\Q$}}=e^{-z (1+1/\lambda)}.
\end{align}
This finishes the proof.


\end{proof}

\begin{corollary}[of \aref{lem.hard2matchwt}]
The conditional expected waiting time for a hard-to-match agent conditioned on not being matched is $1+\lambda-\frac{1}{1+\lambda}$.
\end{corollary}
\begin{proof}
Let $w$ be a random variable denoting the waiting time of an agent $p$ and let $M$ be the event in which  agent $p$ is matched. Also, let $\bar{M}$ be the complement of $M$.
\begin{align}
\E{w} = \P{M}\cdot \E{w\big| M} + \P{\bar{M}}\cdot \E{w\big| \bar{M}}.\label{eq.ew0}
\end{align}
Note that the left-hand-side is equal to $\lambda$. Also, see that
$$\E{w\big| M} =\frac{\lambda}{1+\lambda},$$
which holds by \aref{lem.hard2matchwt}. Note that $\P{M}=\frac{1}{1+\lambda}$. Therefore, plugging the above equation into \eqref{eq.ew0} implies that $$\E{w\big| \bar{M}} = 1+\lambda-\frac{1}{1+\lambda}.$$
\end{proof}

\section{Analysis of Patient Matching}\label{sec.app.mainpresults}
After introducing some notation, we analyze the stochastic process corresponding to the patient policy. First, we present our core technical result, Theorem \ref{thm.pntcncntn}, and prove it in Subsection \ref{sec.proof4patient}.
After that, we prove our results about the match rate and the distribution of waiting time under patient matching in Subsections \ref{sec.mrunderPatient} and \ref{sec.pw8tdist}.

For the analysis in this section, we  use a two-dimensional \MC, $\mc$, to model the dynamics.
Let $V(\calM)$ denote the state space of $\calM$. We represent each state by a pair $(x,y)$ where $x,y$ respectively denote the number of H-agents and E-agents. In other words, we have
$$V(\mc)=\{(x,y): x,y\geq 0, \ x,y \in \mathbb{Z}\}.$$
The definition of \mc{} would be completed by defining the transition rates.
Similar to the \MC{} for the greedy policy, a transition can only happen from a state $(x,y)$ to its (at most) four neighbors, which are $$\{(x',y')\in V(\mc): |x-x'|+|y-y'|=1\}.$$ We do not define these transition rates explicitly here, since they are defined implicitly by the policy.
Define $(x^*,y^*)=\EE{(x,y)\in \pi}{(x,y)}$.

The following technical result forms the basis of our analysis and is proved in \aref{sec.proof4patient}.
\begin{theorem}\label{thm.pntcncntn}
There exists a constant $\sigma_{\mathsf{Y}} > 0$ such that for any $y\geq 0$, we have
\begin{align*}
\pi(y) &\leq  e^{-\sigma_{\mathsf{Y}} \cdot { y}}.
\end{align*}
\end{theorem}

The above result is essentially a concentration result for the E pool. The following concentration result holds for the H pool, which will not be used in our analysis. We state it for the sake of completeness.
\begin{theorem*}
There exists a constant $\sigma_{\mathsf{X}} > 0$ such that for any $x\geq 0$, we have
\begin{align*}
\pi(x) &\leq e^{-\sigma_{\mathsf{X}} \cdot\frac{ |x-x^*|}{\sqrt{m}} }.
\end{align*}
Furthermore, $|x^*-(1+\lambda)m|=o(m)$.
\end{theorem*}

Below we state a weaker version of the above theorem, which will be used in the analysis. The proof is presented in \aref{sec.proof4patient}.
\begin{lemma}\label{lem.xalmostlinp}
There exists a constant $\gamma_1 > 0$ such that for any $d\geq 0$, we have
\begin{align*}
\pi(\lambda m/2- d) &\leq e^{-\gamma_1 d}.
\end{align*}
\end{lemma}

%

\subsection{Proofs}\label{sec.proof4patient}
First, we prove Theorem \ref{thm.pntcncntn}.
\begin{proof}[Proof of Theorem \ref{thm.pntcncntn}]
We use a coupling technique to simplify the stochastic process. We use $\calP$ to denote the stochastic process governing the patient matching algorithm.
Define the stochastic process $\calP''$ to be the same as $\calP$, with the following differences in departures. In $\calP''$
\begin{enumerate}
\item E agents never leave the pool unless they have been  matched, and
\item  $H$ agents do not stay in the pool: an H agent leaves the pool immediately with probability $e^{-m}$, without searching for a match. With probability $1-e^{-m}$ she searches for a compatible agent (upon arrival). If a  match is found, both agents leave the pool. Otherwise, the H agent leaves the pool.
\end{enumerate}

We use the random variables $x'',y''$ to denote the number of H and E agents in $\calP''$.
Recall that we use the \MC{} \mc{} to model $\calP$, and that $\pi$ denotes the steady-state distribution of \mc{}. We use a similar \MC{}, $\mc'$, to model $\calP''$ and use $\pi''$ to denote its steady-state distribution. We use $\piy$ to denote the marginal distribution over $y$ in $\pi$. Similarly $\pippy$  denotes the marginal distribution over $y''$ in $\pi''$.

The proof has  two steps. In Step 1, we show that $\pippy$ stochastically dominates $\piy$. In Step 2, we will show that there exists a constant $\sigma''_{\mathsf{Y}}>0$ such that for all $y''\geq 0$,
$$\pi''(y'') \leq  e^{-\sigma''_{\mathsf{Y}} \cdot { y''}}.$$
This directly proves the lemma for $\sigma_{\mathsf{Y}}=\sigma''_{\mathsf{Y}}$, because of the stochastic dominance relation proved in Step 1.

Before proceeding to Step 1, we define a useful notation. Let $x_t,y_t$ denote the number of H and E agents in $\calP$ at time $t$.
Similarly, we use $x''_t,y''_t$ to denote the number of H and E agents in $\calP''$ at time $t$.

\paragraph{Step 1} First, we define a mediator process, $\calP'$, as follows. $\calP'$ is the same as $\calP$, with the following differences. In $\calP'$
\begin{enumerate}
\item E agents never leave the pool unless they have been matched, and
\item Upon arrival, $H$ agents draw the (exponential) random variable corresponding to their waiting time. If this variable is larger than $m$, they leave the pool immediately without searching for a match. Otherwise, the waiting time of the agent is set to be $m$, i.e., the agent  stays in the pool for $m$ units of time and searches for a match upon departure.
\end{enumerate}
We use the random variables $x',y'$ to denote the number of H and E agents in $\calP'$.
Also, we use $\pi'$ to denote the steady-state distribution corresponding to the process $\calP'$, and $\pipy$ to denote the marginal distribution induced by $\pi$ over $y'$.

To complete Step 1, we will show that (i) $\pippy$ stochastically dominates $\pipy$, and (ii) $\pipy=\piy$. Part (ii) is straightforward: $\calP'$ is just the same as $\calP$ except that the H agents who are allowed to enter the pool in $\calP'$ (those with waiting time shorter than $m$) will enter with a constant delay of $m$. Since the delay is constant, the arrival process remains a Poisson process (with the same rate).  It remains to show that part (i) holds, i.e., $\pippy$ stochastically dominates $\pipy$.


The proof proceeds by defining a coupling of the processes $\calP', \calP''$. The joint process, denoted by $\Q=(\calP', \calP'')$, will have two components corresponding to $\calP', \calP''$. This process $\Q$, in addition to being a valid coupling, will satisfy the following property:  $y'_t \leq y''_{t+m}$ for all $t\geq 0$, in all sample paths of $\Q$. If $\Q$ satisfies this condition, then $\pippy$ must stochastically dominate $\pipy$, and we are done with Step 1.\footnote{Note that $m$ is a constant with respect to $t$. Proving $y_t \leq y''_{t}$ obviously implies the desired stochastic dominance relation. Proving $y'_t \leq y''_{t+m}$ is just as good, since shifting the sample paths $y''_1,y''_2,\ldots$ in time does not change the stationary distribution of $y''$.}
So, all that remains is defining $\Q$ so that it satisfies the above-mentioned condition.

We define $\Q$ as follows. It starts with empty pools in both processes $\calP,' \calP''$, i.e., $x'_t=y'_t=x''_t=y''_t=0$. Both processes will have identical sequences for the arrival of agents, but different departure and matching processes. To define the matching process, we need some additional notation. Let $a'_1\leq a'_2\leq \ldots$ be the sequence of departure times in $\calP'$, i.e., $a'_i$ is the time that the $i$th arrival happens in $\calP'$. Also, let $a''_1\leq a''_2\leq \ldots$ be the sequence of departure times in $\calP''$. We make a final notational convention: in case of departure of an H agent at time $t$, we use $y'_t$ to denote the number of E agents just before that departure, i.e., $y'_t=\lim_{t^*\to t^{-}} y'_{t^*}$. The same definition holds for $y''_t$.\footnote{We use this convention since it allows us to conveniently distinguish between the departure process and the matching process.}

First, note that
\begin{align}
a'_i\leq a''_i\leq a'_i+m,\ \ \ \forall i\geq 0.\label{eq.proof-by-geo-matching}
\end{align}
This holds because $\calP',\calP''$ have the same arrival process; however, the departure of an agent in $\calP''$  can be delayed by up to $m$ units of time, compared to her departure time in $\calP'$.

The rest of the proof in Step 1 is straightforward. We couple $\calP'$ and $\calP''$ in such a way that the following property is satisfied: Upon the $i$th departure in $\calP''$, a compatible match is found iff (i) a compatible match is found upon the $i$th departure in $\calP'$, or (ii) $y''_{a''_i} > y'_{a'_i}$. We label this property  Property $\zeta$. If our coupling satisfies Property $\zeta$, then we are done with Step 1: this fact, the fact that $\calP',\calP''$ have the same arrival process, and \eqref{eq.proof-by-geo-matching} together would imply that $y'_t \leq y''_{t+m}$ holds for all $t\geq 0$.

It remains to show that the coupling $\Q$ can be defined in a way that Property $\zeta$ is satisfied. This is done inductively. The induction basis is $i=1$. See that $a'_1\leq a''_1$, and that  $y'_{a'_1} \leq y''_{a''_1}$. We consider two cases: either $y'_{a'_1} = y''_{a''_1}$ or $y'_{a'_1} < y''_{a''_1}$.
If $y'_{a'_1} = y''_{a''_1}$, then we use the same compatibility graph for the departing agent in both processes $\calP',\calP''$, i.e., the departing agent would be compatible with another H agent in $\calP'$, namely  agent $z$, iff the departing agent is compatible with $z$ in $\calP''$. It is then clear that Property $\zeta$ would be satisfied in this case.

So, suppose that the case $y'_{a'_1} < y''_{a''_1}$ holds. In this case, let $z_1,\ldots,z_{k'}$ denote the H agents in $\calP'$ at time $a'_1$, where $k'=y'_{a'_1}$. Similarly, let $z_1,\ldots,z_{k''}$ denote the H agents in $\calP''$ at time $a''_1$, where $k''=y''_{a''_1}$. According to this notation, the first $k'$ agents in $\calP''$ are the same agents as the $k'$ agents in $\calP'$. Indeed, our coupling treats these agents identically upon the $i$th departure in $\calP',\calP''$. This guarantees that if the departing agent is matched to an agent $z\in \{z_1,\ldots,z_{k'}\}$ in $\calP'$, then, the departing agent will also be matched to $z$ in $\calP''$. Therefore, Property $\zeta$ holds in this case as well.

To complete the induction, suppose that Property $\zeta$ holds for all departures before the $i$th departure. In the induction step, we will show that Property $\zeta$ will be satisfied after the $i$th departure. First, recall that $a'_i\leq a''_i$, by \eqref{eq.proof-by-geo-matching}. This fact, together with the induction hypothesis, implies that $y'_{a'_i} \leq y''_{a''_i}$. We consider two cases: either $y'_{a'_i} = y''_{a''_i}$ or $y'_{a'_i} < y''_{a''_i}$.  We define the coupling for each case separately. This definition is identical to the definition of our coupling in the induction basis. This completes the induction. Therefore, Property $\zeta$  holds, and Step 1 is complete.

\paragraph{Step 2}
As we mentioned in the beginning of the proof, in this step we will show that there exists a constant $\sigma''_{\mathsf{Y}}>0$ such that for all $y''\geq 0$,
\begin{align}
\pi''(y'') \leq  e^{-\sigma''_{\mathsf{Y}} \cdot { y''}}.\label{eq.ppybnd}
\end{align}

Process $\calP''$ is an easily tractable process. The random variable $y''$ evolves according to two Poisson processes (Poisson clocks). The first clock is the arrival clock, which ticks with rate $m$; upon each tick, the value of $y''$ increases by $1$. The second clock is the departure clock, which ticks with rate $(1+\lambda)m(1-e^{-m})$; upon each tick, the value of $y''$ decreases by $1$ with probability $1-M_{y''}$.

Let $y''_0>0$ be a constant such that $(1+\lambda)M_{y''_0}>1+\lambda/2$.
Therefore, for any $y''>y''_0$, the balance equations imply that
\begin{align*}
\frac{\pi''(y''+1)}{\pi''(y'')} \leq \frac{1}{1+\lambda/2}.
\end{align*}
The above equation, together with the fact that $y_0$ is a constant,  implies that there exists a constant $\sigma''_{\mathsf{Y}}$ such that \eqref{eq.ppybnd} holds. This finishes Step 2 and completes the proof.
\end{proof}

\begin{proof}[Proof of Lemma \ref{lem.xalmostlinp}]
Note that for any $x \leq \lambda m/2$, by the balance equations we have $$\frac{\pi(x+1)}{\pi(x)} \leq \frac{x+m}{(1+\lambda)m} \leq  \frac{1+\lambda/2}{1+\lambda}. $$ Letting $\gamma_1=\frac{1+\lambda/2}{1+\lambda}$ proves the claim.
\end{proof}


%

\subsection{Match rate under patient matching}\label{sec.mrunderPatient}
\begin{lemma}\label{lem.pmatchrate}
Under patient matching, the  match rate of hard-to-match agents is $\frac{1}{1+\lambda} - O(\frac{1}{(1+\lambda)\sqrt{m}})$  and the match rate of  easy-to-match agents is $1-o(1)$.
\end{lemma}
\begin{proof}
Let $m_E=m$ and $m_H=(1+\lambda)m$.
We use $\yp_t, \xp_t$ respectively to denote the number of E, H agents under the patient policy at time $t$.
Let $\yp, \xp$ denote the average (steady-state) number of E, H agents under the patient policy, respectively.
Also, let the steady-state number of matches between E and H agents be denoted by a random variable $M_{eh}$.

We will show that $\E{M_{eh}}=m-O(\sqrt{m})$. This would prove the claim on the match rates of E and H agents.
To this end, we define an  event $F$, as follows. Consider an E agent upon her arrival, namely agent $a$.
Let $F$ be the event in which $a$ is matched with another E agent.
Suppose $t$ denotes the time when $a$'s departure clock is set to tick (which is determined by a draw from $\expdist(1)$, right after her arrival).
Let $G$ be the event in which $a$ leaves the pool at time $t$, while being matched to another E agent. See that
\begin{align}
\P{F} \leq 2\P{G}.\label{bpfahwpgah}
\end{align}
This holds because any match between two E agents is formed upon the departure of one of them. (In the above inequality, the probabilities are unconditional steady-state probabilities, i.e.,  the probabilities are not conditional on the state of the pool that $a$ enters to, the waiting time of $a$, or compatibility of $a$ to the other agents.)

Next, we will show that $\P{G}=o(1)$. This holds essentially by Lemma \ref{lem.xalmostlinp}.  Roughly speaking, this lemma says that the steady-state H pool is large, and therefore, $a$ is compatible with an H agent wvhp. This would imply that $F$ holds wvlp.
More precisely, Lemma \ref{lem.xalmostlinp} is applicable since the marginal distribution of the size of the H pool just before an E departure (conditional on an E departure) is equal to the unconditional marginal steady-state distribution of the size of the H pool.\footnote{Recall that an E departure denotes the event of the departure of an E agent, whether because the agent is matched or because the agent's departure clock has ticked.} (This is a consequence of the PASTA property.)
Therefore, Lemma \ref{lem.xalmostlinp} is applicable and would imply that, wvhp, the size of the H pool is at least $\lambda m/4$ upon the departure of $a$. Therefore, wvhp, $a$ is compatible with at least one of the H agents in the pool. This means $G$ happens wvlp. This fact and  \eqref{bpfahwpgah} together imply that there exists a constant $\eta>0$ such that
$$\P{F} \leq e^{-\eta m}.$$
The above equation and linearity of expectation together imply that $\E{M_{eh}}\geq m-e^{-\eta m}$, which means $\E{M_{eh}}=m-O(\sqrt{m})$. (Note that $\E{M_{eh}}\leq m$.)
\end{proof}

\subsection{Distribution of waiting time}\label{sec.pw8tdist}
\begin{lemma}\label{lem.distofEPw8t}
Under the patient policy, as $m$ converges to infinity, the waiting time and matching time of an E agent  converge in distribution to the degenerate distribution at $0$.
\end{lemma}
\begin{proof}
Fix an E agent, $e$, and let $w_e$ denote the waiting time of $e$. For any fixed constant $t>0$, we will show that $\lim_{m\to\infty}\P{t>w_e}=0$. This will prove the claim. First, see that
\begin{align*}
\E{w_e} = \frac{1}{m}\cdot \EE{(x,y)\sim \pi}{x} = \frac{1}{m}\cdot O(1)
\end{align*}
where the second equality is by \aref{thm.pntcncntn}. Therefore, by Markov inequality, $$\P{w_e>t}<O(t/m)$$ holds, which means $\lim_{m\to\infty}\P{t>w_e}=0$.
\end{proof}

\begin{lemma}\label{lem.distofHPw8t}
Under the patient policy, as $m$ converges to infinity,  the waiting time and matching time of an $H$ agent converge in distribution to an exponential random variable with rate $1$.
\end{lemma}
\begin{proof}
Fix an H agent, namely $h$, upon  arrival. Let $E_h$ denote the event in which agent $h$ receives no offers before his departure clock ticks. More precisely, $E_h$ denotes the event in which no E agent will check to see whether agent $h$ is compatible with him.\footnote{We are implicitly assuming that E agents, upon departure, search for a compatible match by first visiting H agents one by one, and match with the first compatible agent. This is of course without loss of generality, by the principle of deferred decisions.} To prove the lemma, it is enough to show that
\begin{align}
\lim_{m\to \infty}\P{E_h} \to 0.\label{eq.mainHwinP}
\end{align}
\begin{claim}
If \eqref{eq.mainHwinP} holds, then the waiting time and matching time of an $H$ agent converge in distribution to an exponential random variable with rate~$1$.
\end{claim}
\begin{proof}
Let $w$ denote the waiting time of $h$ and let $d$ denote the time that his departure clock is set to tick (which will be set upon the arrival of $h$).
Note that for any $t$ we have
\begin{align*}
\P{d < t} \leq \P{w < t} &\leq \P{d < t} +\P{(d>t) \wedge E_h} \\
&\leq \P{d < t} +\P{E_h}.
\end{align*}
The above inequality and \eqref{eq.mainHwinP} together imply that
\begin{align*}
\lim_{m\to\infty} \P{w < t} = \P{d < t},
\end{align*}
which proves the claim.
\end{proof}


To prove that \eqref{eq.mainHwinP} holds, we slightly modify the stochastic process after the arrival of $h$.
Let $\Q$ denote the original process, which starts upon the arrival of $h$, namely at time $t_0$, and ends when $h$ leaves the system.
The modified process, namely $\Q'$, is similar to $\Q$, except for the differences clarified below.
Roughly speaking, $\Q'$ is the same as the patient matching process with the exception that the arrival of $h$ is ``ignored'' in the sense that $h$ does not interfere with the evolution of $\Q'$.
\begin{itemize}
\item $\Q'$ runs from time $t_0$ to $t_0+\log m$. Furthermore, the departure clock of $h$ is set to tick at time $t_0+\log m$ in $\Q'$.
\item By definition, $\Q'$ has a sample path identical to $\Q$, until one of the following disjoint events happens:
\begin{enumerate}[{Event} (i)]
\item  $h$ receives an offer in $\Q$ before time $t_0+\log m$. In this case, we stop $\Q$. In $\Q'$, $h$ rejects the received offer as well as all of the offers that she will receive in the future. Any E agent who gets rejected by $h$ will make an offer to the  next compatible agent in his (random) list. $\Q'$ will continue evolving according to the patient matching process, with the exception of agent $h$, who does not interfere with the process.
\item The departure clock of $h$ ticks in $\Q$ before time $t_0+\log m$, and no offers are made to $h$ prior her departure.  In this case, we stop $\Q$ but continue to run  $\Q'$. $\Q'$ will continue evolving according to the patient matching process, with the exception of agent $h$, who does not interfere with the process (i.e., agent $h$ rejects all the offers she receives, in the sense clarified above).
\end{enumerate}
\item We stop $\Q'$ when it reaches  the time $t_0+\log m$.
\end{itemize}

Let $E'_h$ denote the event in which $h$ receives at least one offer in $\Q'$ before time $t_0+\log m$. A straightforward coupling argument shows that
\begin{align*}
\P{E_h} \leq \P{E'_h} + e^{-\log m}.
\end{align*}
Therefore, to prove that \eqref{eq.mainHwinP} holds, it is enough to show that
\begin{align}
\lim_{m\to \infty}\P{E'_h} \to 0.\label{eq.mainHwinPQ}
\end{align}


Our approach for proving \eqref{eq.mainHwinPQ} is as follows. First, we observe that, when the arrival of $h$ is ignored, the process $\Q'$ can be run from time $t_0$ to $t_0+\log m$ as follows: sample a state $(x,y)\sim \pi$. Then, let the stochastic system start from $(x,y)$ and evolve for $\log m$ units of time, under the patient policy. By the PASTA property, the sample paths generated by this process are identical to the sample paths of $\Q'$ (when the arrival of $h$ is ignored).

By the above argument, we sample the state $(x,y) \sim \pi$ at time $t_0$ and let the process run until time $t_0+\log m$. We also consider an ``imaginary agent'' $h$, which exists in the H pool, but rejects all of the proposals that are made to her. Our goal is showing that the probability that $h$ receives at least one proposal over the period $[t_0, t_0+\log m]$, which we denoted by $\P{E'_h}$, approaches $0$ as $m$ approaches infinity.

Without loss of generality, let $t_0=0$ and $t=\log m$ for notational simplicity. Let $\H_t$ denote the history of the process until time $t$. (Note that the history does not include the offers made to $h$, since this agent does not change the evolution of the process.)
Let the random variable $n_t$ denote the number of E agents whose departure clock ticks in the interval $[0,t]$. Also, let $d_1,\ldots,d_{n_t}$ denote the times that the departure clocks of these agents tick. Define $x_1,\ldots,x_{n_t}$ to be the number of H agents in the pool at times $d_1,\ldots,d_{n_t}$, respectively.

First, see that
\begin{align}
\P{\bar{E_h} \big| \H_t} \geq \Pi_{i=1}^{n_t} (1-1/x_{i}).\label{eq.mainHwinH}
\end{align}
In the rest of the proof, we will use probabilistic bounds to show that in almost all histories $\H_t$, $n_t$ is small and $x^*$ is large, where $x^*=\min\{x_1,\ldots,x_{n_t}\}$. These two facts, together with \eqref{eq.mainHwinH}, will finish the proof.

Define  event $G_{1}$ as $$x^* >\lambda m/4.$$ The goal is to show that $\bar{G}_{1}$ happens wvlp. For this, we will show that the size of the E pool after the arrival of any E agent is small, wvhp. Formally, we can use the PASTA property together with Lemma \ref{lem.xalmostlinp} and write a union bound over all arrivals of E agents in the interval $[0,t]$. This implies that
\begin{align}
\P{\bar{G}_1} = O(mt e^{-\gamma_3 m}),  \label{eq.g1ub}
\end{align}
where $\gamma_3=\lambda \gamma_2/2$.

Define  event $G_2$ as $n_t< t \log m$. By \aref{thm.pntcncntn}, $\E{n_t}=O(t)$. Therefore, by Markov inequality,
\begin{align}
\P{\bar{G}_2} = O(1/\log m). \label{eq.g2ub}
\end{align}

Now, recall \eqref{eq.mainHwinH}, and that our goal is to show that in almost all histories $\H_t$, $n_t$ is small and $x^*$ is large. \eqref{eq.g1ub} and \eqref{eq.g2ub} show just this. More formally, they imply that in almost all histories $\H_t$ but a fraction $O(1/\log m)$ of them,  events $G_1,G_2$ hold. Therefore, using \eqref{eq.mainHwinH}, we can write
\begin{align}
\P{\bar{E_h}} & \geq  \left(1-O({1}/{\log m})\right) \cdot \left(1-\frac{1}{\lambda m/4.}\right)^{t \log m},\\
& \geq  \left(1-O({1}/{\log m})\right) \cdot \left(1-\frac{{t \log m}}{\lambda m/4.}\right)=1-o(1),
\end{align}
which implies that $\P{{E_h}}=o(1)$. This completes the proof.
\end{proof}

\section{Proof of Proposition \ref{thm.opt2}, Part (ii)}\label{appendix.batching}
We break the proof  into four parts: \aref{prp.batchingmr} and \aref{prp.batchingwt} prove the bounds on the match rate and waiting time of H agents, and \aref{prp.batchingmrEagents} and \aref{prp.batchingwtEagents} prove the bounds on the match rate and waiting time of E agents. These propositions are stated and proved after the following preliminary analysis.

%
%


Suppose that time is indexed by real numbers and starts from $0$. The batching policy makes matches every $T$ units of time. Define {\it round} $i$ to be the time window  between time $(i-1)T+1$  right after the matches are made (if any) and time $iT$ right before the matches are made.

To give an upper bound on the match rate of the batching policy, we can analyze a simpler process instead, which we call process $\calB$ and define as follows.
$\calB$ is similar to the original batching policy, with the following differences:
\begin{enumerate}
\item Easy-to-match nodes are not compatible.
\item The probability of compatibility between an easy-to-match node and a hard-to-match node is $1$.
\end{enumerate}

A straightforward coupling argument shows that the match rate of H agents in $\calB$ is smaller and their (average) waiting time larger than in the batching policy. (This holds in each sample path of the coupled process.) Therefore, to provide the promised bounds for H~agents in the batching policy, it suffices to analyze the simpler process, $\calB$. All the following definitions  are therefore defined for the process $\calB$.

Let $E_i, H_i$ respectively denote the number of E agents and H agents at the beginning of round $i$. Let $X_i=H_i-E_i$. (Note that $ E_i \cdot H_i=0$ must always hold.) It is straightforward to verify that $\calC=\langle (E_1,H_1),(E_2,H_2),\ldots\rangle$ is an ergodic Markov chain with its state space being the set of pairs of non-negative integers. Let $\pi$ denote the steady-state distribution of $\calC$.

Let $E'_i, H'_i$ respectively denote the number of E agents and H agents who were present in the pool both at the beginning and at the end of round $i$ .
Define $X'_i=H'_i-E'_i$. Moreover, let $E''_i, H''_i$ respectively denote the number of E agents and H agents who arrived after the beginning of round $i$ and were present in the pool at the end of round $i$. Define $X''_i=H''_i-E''_i$.

Next, observe that
\begin{align}
X_{i+1} = X'_{i}+X''_{i}.
\end{align}
Taking expectation (with respect to the stationary distribution $\pi$) from both sides of the above equality implies that
\begin{align}
\EE{\pi}{X_{i+1}} = \EE{\pi}{X'_{i}}+\EE{\pi}{X''_{i}}. \label{eq.x}
\end{align}
We will compute $\EE{\pi}{X_{i}}$ using the above equation and the following observations.

First, observe that
\begin{align}
\EE{\pi}{X'_{i}}=  \EE{\pi}{H'_{i}}-\EE{\pi}{E'_{i}} =\EE{\pi}{X_{i}}\cdot e^{-\gamma T},\label{eq.xp}
\end{align}
where $\gamma=1/d$. This equality holds simply because an agent who is present at the beginning of round $i$ will also be present at the end of round $i$ with probability $ e^{-\gamma  T}$.

Let $N_i$ denote the number of arrivals.  Second, observe that
\begin{align}
\EE{\pi}{X''_{i}}&=\EE{\pi}{H''_{i}}-\EE{\pi}{E''_{i}}\nonumber\\
&=\sum_{n=1}^{n=\infty} \P{H''_i=n} \cdot n\cdot \int_0^{T}\frac{1}{T}\cdot e^{-\gamma (T-t)} \opd t - \sum_{n=1}^{n=\infty} \P{E''_i=n} \cdot n\cdot \int_0^{T}\frac{1}{T}\cdot e^{-\gamma (T-t)} \opd t \nonumber\\
&= \int_0^{T}\frac{1}{T}\cdot e^{-\gamma (T-t)} \opd t \cdot \left(\EE{\pi}{H''_i}-\EE{\pi}{E''_i}\right) = (\frac{1}{T}-\frac{e^{-\gamma T}}{T})\cdot \lambda m T/\gamma. \label{eq.xpp}
\end{align}

We then can use \eqref{eq.x}, \eqref{eq.xp}, and \eqref{eq.xpp} to write
\begin{align}
\EE{\pi}{X_{i+1}} =\EE{\pi}{X_{i}}\cdot e^{-\gamma T}+(\frac{1}{T}-\frac{e^{-\gamma T}}{T})\cdot \lambda m T/\gamma.\nonumber
\end{align}
We can simplify the above equation further using the fact that $\EE{\pi}{X_{i+1}}=\EE{\pi}{X_{i}}$, which holds since $\pi$ is the steady-state distribution:
\begin{align}
\EE{\pi}{X_{i}} =\frac{(\frac{1}{T}-\frac{e^{-\gamma T}}{T})\cdot \lambda m T/\gamma}{1-e^{-\gamma T}}=\lambda m/\gamma.\label{eq.expxi}
\end{align}

Next, we use \eqref{eq.expxi} to compute an upper bound on the match rate of H agents under the batching policy. Equivalently, we compute a lower bound on the death rate of H agents, where the death rate is the expected fractions of H agents per unit of time who depart the market without getting matched. To this end, we just need to compute the expected number of such departures (i.e., {\em deaths}) per round.

Fix a round $i$ and let  $D_i$ denote the number of deaths in round $i$. The expected number of deaths of H agents during round $i$ is just equal to the expected number of H agents who were present at the beginning of round $i$ but not at the end, plus the expected number of H agents who arrived sometime in round $i$ but were not present at the end of round $i$. That is,
\begin{align}
\EE{\pi}{D_i}= \EE{\pi}{\max\{X_i,0\}\cdot (1-e^{-\gamma T})} + (1+\lambda)m T \cdot (1-\frac{1}{\gamma T}+\frac{e^{-\gamma T}}{\gamma  T}). \label{eq.expecteddeaths}
\end{align}
Verify that the second summand on the right-hand-side is indeed the expected number of H agents who arrived sometime in round $i$ but were not present at the end of round $i$ by the following equality:
\begin{align}
\sum_{n=1}^{n=\infty} \P{H''_i=n} \cdot n\cdot \int_0^{T}\frac{1}{T}\cdot (1-e^{-\gamma (T-t)}) \opd t = (1+\lambda)m T\cdot (1-\frac{1}{\gamma T}+\frac{e^{-\gamma T}}{\gamma  T}).\nonumber
\end{align}

We now can use \eqref{eq.expecteddeaths} to write the inequality
\begin{align}
\EE{\pi}{D_i}\geq \EE{\pi}{X_i}\cdot (1-e^{-\gamma T}) + (1+\lambda)m T\cdot (1-\frac{1}{\gamma T}+\frac{e^{-\gamma T}}{\gamma T}), \nonumber
\end{align}
which can be simplified using \eqref{eq.expxi} to
\begin{align}
\EE{\pi}{D_i}\geq \lambda m /\gamma \cdot (1-e^{-\gamma T}) + (1+\lambda)m T\cdot (1-\frac{1}{\gamma T}+\frac{e^{-\gamma T}}{\gamma T}).
\end{align}
The above inequality implies
\begin{align}
\frac{\EE{\pi}{D_i}}{(1+\lambda)m T} &\geq \frac{\lambda  (1-e^{-\gamma T})}{(1+\lambda) T\gamma} + (1-\frac{1}{\gamma T}+\frac{e^{-\gamma T}}{\gamma T})\nonumber
\end{align}
which is the promised lower bound on the death rate of H agents (i.e., the left-hand-side of the above inequality). This implies that the match rate of the H agents is bounded from above by
\begin{align}
\frac{1-e^{-\gamma T}}{\gamma T} - \frac{\lambda  (1-e^{-\gamma T})}{(1+\lambda)\gamma T}.\label{eq.mr}
\end{align}

\begin{proposition}\label{prp.batchingmr}
Fix $T>0$. There exists a constant $r<\frac{1}{1+\lambda}$ such that for any $m,\lambda>0$, the match rate of H agents under the batching policy is smaller than $r$.
\end{proposition}
\begin{proof}
The proof uses the upper bound \eqref{eq.mr} on the match rate of the batching policy. Let $u$ denote this upper bound.  We will show that $u<\frac{1}{1+\lambda}$ holds for any $m,\lambda>0$. Simple algebra reveals that this is equivalent to showing that
\begin{align}
1 + e^{\gamma T} (-1 + \gamma T)>0.\label{eq.tobeproved}
\end{align}
Let the left-hand-side of \eqref{eq.tobeproved} be denoted by $f(T)$. Observe that $f(T)$ is strictly increasing in $T$ since $f'(T)=\gamma^2 T e^T>0$. Since $f(0)=0$ and $T>0$, then $f(T)>0$. This proves the claim.
\end{proof}

\begin{proposition}\label{prp.batchingwt}
Fix $T>0$. There exists a constant $r>\frac{\lambda}{1+\lambda}$ such that for any $m,\lambda>0$, the expected waiting time of H agents under the batching policy is larger than $r$.
\end{proposition}
\begin{proof}
Let $W_i$ denote the waiting time incurred by all H agents in round $i$. The expected waiting time of H agents is just equal to $\frac{\EE{\pi}{W_i}}{(1+\lambda) m T}$. We therefore prove the claim by computing a lower bound on $\EE{\pi}{W_i}$. Observe that
\begin{align}
\EE{\pi}{W_i} = \EE{\pi}{W'_i} + \EE{\pi}{W''_i},\label{eq.expwi}
\end{align}
where $W'_i$  denotes the waiting time incurred in round $i$ by the H agents who were present at the beginning of round $i$ and
 $W''_i$  denotes the waiting time incurred in round $i$ by the H agents who were not present at the beginning of round $i$.

Verify that
\begin{align}
\EE{\pi}{W'_i}=\EE{\pi}{\max\{X_i,0\}} \cdot \left(\int_0^T \gamma t e^{-\gamma  t}\opd t + Te^{-\gamma  T} \right)\geq \lambda m (1-e^{-\gamma T})/\gamma^2,\label{eq.expwpi}
\end{align}
where the left-hand-side is the product of the expected number of H agents that are present at the beginning of round $i$ and the expected waiting time incurred by any such agent in round $i$. Also, verify that
\begin{align}
\EE{\pi}{W''_i}&=(1+\lambda)m T \cdot \left( \int_0^T \frac{1}{T} \left(\int_0^{T-s} \gamma t e^{-\gamma t}\opd t + (T-s) e^{-\gamma (T-s)}\right)\opd s \right)\nonumber\\
&=(1+\lambda)m  ({e^{-\gamma T}+\gamma T-1})/{\gamma^2} ,\label{eq.expwppi}
\end{align}
where on the right-hand-side of the first equality the first term is the expected number of H agents who arrive sometime in round $i$ and the second term is the expected waiting time that each such agent incurs. The second equality follows from simple calculations.

We now can use \eqref{eq.expwi}, \eqref{eq.expwpi}, and \eqref{eq.expwppi} to write the following lower bound on the waiting time of batching policies:
\begin{align}
 \frac{\EE{\pi}{W_i}}{(1+\lambda) m T}\geq \frac{ \gamma T (1+\lambda) +e^{-\gamma T}-1  }{\gamma^2 (\lambda+1) T}.
\end{align}
Simple algebra  shows that  the right-hand-side of the above inequality is larger than $\frac{\lambda/\gamma}{1+\lambda}$ if and only if
\begin{align*}
1 + e^{\gamma T} (-1 + \gamma T)>0.
\end{align*}
Let the LHS of the above inequality be denoted by $f(T)$.
Observe that $f(T)$ is strictly increasing in $T$ since $f'(T)=\gamma^2 T e^T>0$. Since $f(0)=0$ and $T>0$, then $f(T)>0$. This proves the claim.
\end{proof}

\begin{proposition}\label{prp.batchingmrEagents}
Fix $T>0$. The match rate of E agents is at most $\frac{1-e^{-\gamma T}}{\gamma T}$.
\end{proposition}
\begin{proof}
First we compute the chance that an E agent who arrives during round $i$ does not become critical before round $i$ ends:
\begin{align*}
\int_0^T \frac{1}{T} e^{-\gamma t}\opd t=\frac{1-e^{-\gamma T}}{\gamma T}.
\end{align*}
This implies that, conditioned on arriving in period $i$, an E agent does not get matched with probability at least $1-\frac{1-e^{-\gamma T}}{\gamma T}$. This proves the claim.
\end{proof}

\begin{proposition}\label{prp.batchingwtEagents}
Fix $T>0$. The waiting time of E agents is at least $\frac{1}{\gamma}-\frac{1-e^{-\gamma T}}{\gamma^2 T}$.
\end{proposition}
\begin{proof}
We compute the expected waiting time that an E agent who arrives in round $i$ incurs during round $i$:
\begin{align*}
\int_0^T \frac{1}{T}\cdot  \left(\int_0^{T-s}\gamma t e^{-\gamma t}\opd t + (T-s) e^{-\gamma (T-s)}\right)\opd s & =\frac{ e^{-\gamma T} + \gamma T-1}{\gamma^2 T}\\
&= \frac{1}{\gamma}-\frac{1-e^{-\gamma T}}{\gamma^2 T}.
\end{align*}
This proves the claim.
\end{proof}

\section{Few Hard-to-Match Agents ($\lambda<0$)}
\label{sec:small-lambda}

For completeness we analyze  the behavior of the patient and greedy matching policies when $\lambda<0$.
\begin{theorem}
\label{thm:otherlambda}
Let  $\lambda \in [-1,0)$.
\begin{enumerate}
\item The  match rate for all agents is $1-O(1/m)$ under both policies.
\item Under greedy matching, the expected waiting time and matching time of H agents are $O(1/m)$, and the expected waiting time and matching time of E agents are $O\left(\frac{\log^{ m}_{1/(1-p)}}{m}\right)$.
\item Under  patient matching, the expected waiting time and matching time of H agents are $O(1/m)$; moreover, the distribution of waiting time and matching time of E agents converge to the  exponential random variable with mean $\frac{1}{1-\lambda}$.
\end{enumerate}
\end{theorem}

We omit the proof of this theorem. It is a simpler version of the proof of Theorem \ref{thm.optimalpolicy}.

\section{Simulations Based on the Stylized Model}
\label{sec-simulations-stylized}

This section presents  simulation results based on the  stylized model. The findings illustrate that the  theoretical predictions  hold even in  small markets.
As in the model, no two  hard-to-match agents are  compatible with  each other. Every two easy-to-match agents are compatible with probability $q=0.04$, independently. Every pair of easy-to-match and hard-to-match agents is compatible with probability $p=0.1$, independently. The choice of these probabilities is  motivated from some components of the kidney exchange pool.\footnote{Consider patient-donor pairs with pair of blood types A-O, respectively,  and O-A pairs. Since the donor and the patient of an  A-O pair   are ABO compatible, the patient is likely to be sensitized and thus  two O-A pairs will have a low chance of being compatible  with  each other. Since the donor of an O-A pair is ABO incompatible, the patient is less sensitized, and therefore such a pair has  a higher probability to being compatible with an A-O pair.}\,\footnote{This choice of parameters is also  interesting because it shows that the  convergence rates are  fast even when $p>q$, and even faster for larger $q$, e.g., when $p=0.04$, $q=0.1$.}

Other parameters of the model are set as follows. Agents arrive according to a Poisson process with an average of 1 per time period (day) and become critical according to an  exponential distribution with mean $200$ days.\footnote{When no matches are made at all, the expected pool size is therefore about $200$. Typical kidney exchange networks are fairly thicker; for example, recall the estimates for arrival and departure rates from our data set in Section \ref{sec:simulationsData} (i.e., 1 arrival per day and a departure rate of 1/360). According to these estimates, when no matches are made, the expected pool size would be $360$.} We set $\lambda = 0.5$ in the   first set of simulations and then  vary  $\lambda$. Both the greedy and patient matching policies are simulated  until $70,000$ agents arrive and  the waiting and matching times are recorded for all but the first 5000 agents.\footnote{We ignore the first 5000 agents  to ensure that our samples are taken when the stochastic process is (almost) in the steady state, i.e., the corresponding Markov chain is mixed.}

Under  the greedy policy the waiting time and matching time distributions of hard-to-match agents  fit exponential distributions  with means 64.63 and 64.6, respectively, where the mean predicted by Proposition \ref{thm.opt2} (albeit in a thick enough market) is $\frac{\lambda}{1+\lambda}\cdot 200\approx 66.66$.
The empirical distribution of waiting times, the empirical distribution of matching times,  and the CDF of the exponential distribution with mean 66.66 are depicted in Figure \ref{fig.UD1} (note that they  almost coincide).

Under  the patient policy  the waiting time and matching time distributions of hard-to-match agents fit exponential distributions  with means 189.09 and 190.12, respectively.  Figure \ref{fig.UD1} also plots these empirical distributions, as well as the exponential distribution with mean 200 (here too, the CDFs almost coincide).

The fraction  of hard-to-match agents that are matched is approximately $0.67$ under both policies, the fraction  of easy-to-match agents that are matched is approximately $1$ under both policies, and in particular, almost all easy-to-match agents are matched with hard-to-match agents.

\begin{figure}[H]
\centering
{\includegraphics[scale=0.5]{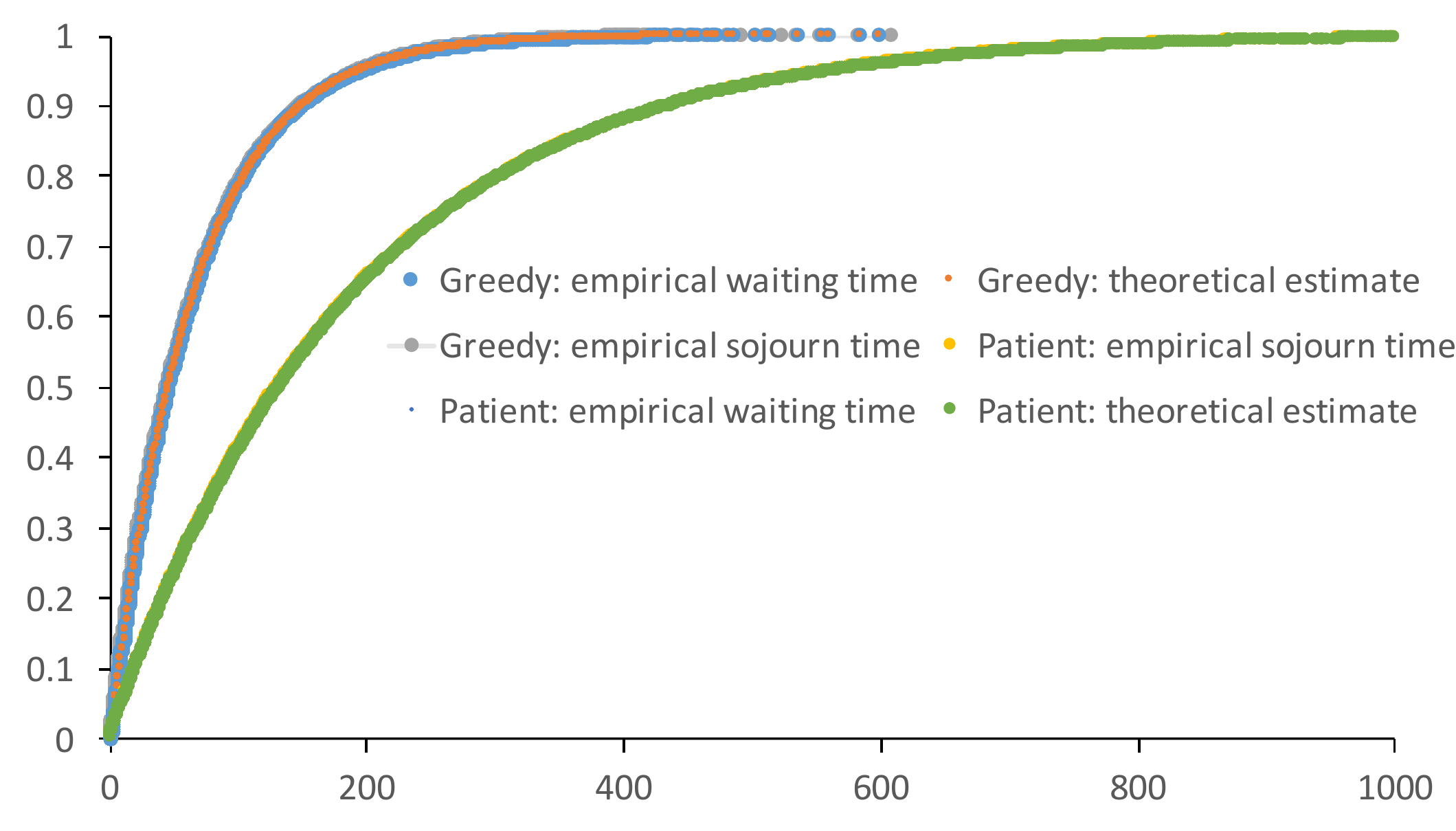}}%
\caption{\footnotesize{Waiting time and matching time distributions under the greedy and patient policies for $\lambda=0.5$.}\label{fig.UD1}}
\end{figure}

Table \ref{tab:simuldata} presents similar results  for different values of $\lambda$, while fixing the total arrival rate to $200$.
The values of $\lambda$ are derived by setting the fraction of easy-to-match agents arriving to the pool (i.e., $\frac{1}{2+\lambda}$) to $0.25, 0.35$ and $0.45$, which respectively correspond to $\lambda=2$, $\lambda=0.857$, and $\lambda=0.222$.
For any type of agent $\Theta\in\{\textrm{E}, \textrm{H}\}$, let $WT(\Theta)$ be the average waiting time of matched agents of type $\Theta$,  $MT(\Theta)$ be the average matching time of agents of type  $\Theta$, and  $M(\Theta)$ be the fraction  of  agents of type $\Theta$ who match. Also let $\hat{MT}(\Theta)$ be the {\em estimated average matching time} of type $\Theta$ that is obtained by applying the limit result of Proposition \ref{thm.opt2}, i.e., $\hat{MT}(E)=0$ and $\hat{MT}(H)=\frac{\lambda}{1+\lambda}\cdot 200$. Although the formal limit result holds  when $m$ approaches infinity, it can  be used to estimate the matching time for small $m$. As Table \ref{tab:simuldata} shows, these estimates are reasonably close to the empirical average matching time  and the average waiting time.


{\linespread{1}
\begin{table}[H]
  \centering
  \footnotesize
    \begin{tabular}{r@{\hskip 20pt}r@{\hskip 20pt}rrrrrrrr}
    Policy & $\lambda$ & MT(H) & WT(H) & $\hat{MT}$(H) & MT(E) & WT(E) & $\hat{MT}$(E) &  M(H) & M(E) \\
    \hline
Patient & 0.222 & 181.8 & 182.4 & 200 & 18.9 & 18.9 & 0 & 0.814 & 0.99 \\
Greedy & 0.222 & 36.31 & 36.87 & 36.36 & 0.28 & 0.29 & 0 & 0.82 & 1 \\
Patient & 0.857 & 188.9 & 192.7 & 200 & 10.43 & 10.43 & 0 & 0.54 & 1 \\
Greedy & 0.857 & 92.5 & 92.7 & 92.3 & 0 & 0 & 0 & 0.54 & 1 \\
Patient & 2 & 191.3 & 196.3 & 200 &  7.43 & 7.43 & 0 & 0.33 & 1 \\
   Greedy & 2 & 132.6 & 132.8 & 133.3 & 0 & 0& 0 & 0.33 & 1
    \end{tabular}%
   \caption{\footnotesize{Statistics for the stylized model for different values of $\lambda$.}}
    \label{tab:simuldata}
\end{table}%
}

The waiting time and matching time distributions  follow   exponential  distributions with the corresponding averages reported in Table \ref{tab:simuldata} (plots are omitted). Note that under  the patient policy  the  average waiting time and the average matching time of hard-to-match agents are similar for all values of  $\lambda$ and are close to the predicted waiting time, 200, which is the exogenous average time, in which  hard-to-match agents  become critical.
Under  greedy matching, for each $\lambda$ the   average waiting time and average  matching time for hard-match-agents are similar and are close to the predicted values,  $200\cdot\frac{\lambda}{1+\lambda}$.

%
%
%


\end{document}